\documentclass[final,3p]{elsarticle}

\usepackage[english]{babel}

\usepackage{amsmath}
\usepackage{amssymb}
\usepackage{amsthm}
\usepackage{graphicx}
\usepackage{enumitem,linegoal}
\usepackage{mathrsfs}
\usepackage{verbatim}
\usepackage{subcaption}
\usepackage{float}
\usepackage{cleveref}
\usepackage[T1]{fontenc} %%To have polishhook for polish names
\usepackage[utf8]{inputenc}

\usepackage{tikz,tkz-graph,amsmath,amsfonts,amssymb,graphics}
\usepackage{subcaption}
\usepackage[numbers]{natbib}
\usepackage{cleveref}[2011/12/24]
\usepackage{enumitem,linegoal}
\usetikzlibrary{patterns}
\usetikzlibrary{arrows.meta}
\usetikzlibrary{arrows}
\usetikzlibrary{arrows.meta}
\usetikzlibrary{decorations.pathmorphing}
\usetikzlibrary{decorations.markings}
\usetikzlibrary{calc}

\tikzset{
  circ/.style = {circle,draw,fill,inner sep=1pt},
  lcirc/.style = {circle,draw,fill,inner sep=4pt},
  scirc/.style = {circle,draw,fill,inner sep=.5pt},
  mcirc/.style = {circle,draw,fill,inner sep=2pt},
  invisible/.style = {circle,draw=none,inner sep=0pt,font=\tiny}
}

\tikzset{clause/.pic={
	\node[lcirc] (q1) at (4,0) [] {};
	\node[lcirc] (p2) at (8,0) [] {};
	\node[lcirc] (p1) at (2,3) [] {};
	\node[lcirc] (o) at (6,3) [] {};
	\node[lcirc] (q2) at (10,3) [] {};
	\node[lcirc] (q3) at (4,6) [] {};
	\node[lcirc] (p3) at (8,6) [] {};
    \draw[-]
    (q1) -- (o) -- (p2) -- (q2) -- (o) -- (p3) -- (q3) -- (o) -- (p1) -- (q1)
    (p1) -- (0,4)
    (q1) -- (4,-2)
    (p2) -- (8,-2)
    (q2) -- (12,4)
    (p3) -- (10,7)
    (q3) -- (2,7);
	\draw
	(0,4) circle (0.4) circle (0.6)
	(4,-2) circle (0.4) circle (0.6)
	(8,-2) circle (0.4) circle (0.6)
	(12,4) circle (0.4) circle (0.6)
	(10,7) circle (0.4) circle (0.6)
	(2,7) circle (0.4) circle (0.6);
}}

\tikzset{clausebis/.pic={
	\node[mcirc] (q1) at (4,0) [] {};
	\node[mcirc] (p2) at (8,0) [] {};
	\node[mcirc] (p1) at (2,3) [] {};
	\node[mcirc] (o) at (6,3) [] {};
	\node[mcirc] (q2) at (10,3) [] {};
	\node[mcirc] (q3) at (4,6) [] {};
	\node[mcirc] (p3) at (8,6) [] {};
    \draw[-]
    (q1) -- (o) -- (p2) -- (q2) -- (o) -- (p3) -- (q3) -- (o) -- (p1) -- (q1)
    (p1) -- (0,4)
    (q1) -- (4,-2)
    (p2) -- (8,-2)
    (q2) -- (12,4)
    (p3) -- (10,7)
    (q3) -- (2,7);
	\draw
	(0,4) circle (0.25) circle (0.4)
	(4,-2) circle (0.25) circle (0.4)
	(8,-2) circle (0.25) circle (0.4)
	(12,4) circle (0.25) circle (0.4)
	(10,7) circle (0.25) circle (0.4)
	(2,7) circle (0.25) circle (0.4);
}}

\tikzset{clause_terminals/.pic={
	\path pic (c) {clause};
	\node[lcirc, cyan] (-t1) at (0,0) [] {};
	\node[lcirc, cyan] (-t2) at (12,0) [] {};
	\node[lcirc, cyan] (-t3) at (6,9) [] {};
    \draw[cyan]
    (2,3) -- (-t1) -- (4,0)
    (8,0) -- (-t2) -- (10,3)
    (8,6) -- (-t3) -- (4,6);
}}

\tikzset{variable/.pic={
	\node[lcirc] (a) at (6,2) [] {};
	\node[lcirc] (b) at (6,6) [] {};
	\node[lcirc] (c) at (2,6) [] {};
	\node[lcirc] (d) at (2,2) [] {};
	\node[lcirc] (e) at (4,4) [] {};
    \draw[-]
    (7,1) -- (a) -- (b) -- (c) -- (d) -- (a) -- (e) -- (b) -- (7,7)
    (1,7) -- (c) -- (e) -- (d) -- (1,1);
    \draw
	(1,1) circle (0.4) circle (0.6)
	(1,7) circle (0.4) circle (0.6)
	(7,1) circle (0.4) circle (0.6)
	(7,7) circle (0.4) circle (0.6);
}}

\tikzset{variablebis/.pic={
	\node[mcirc] (a) at (6,2) [] {};
	\node[mcirc] (b) at (6,6) [] {};
	\node[mcirc] (c) at (2,6) [] {};
	\node[mcirc] (d) at (2,2) [] {};
	\node[mcirc] (e) at (4,4) [] {};
    \draw[-]
    (7,1) -- (a) -- (b) -- (c) -- (d) -- (a) -- (e) -- (b) -- (7,7)
    (1,7) -- (c) -- (e) -- (d) -- (1,1);
    \draw
	(1,1) circle (0.25) circle (0.4)
	(1,7) circle (0.25) circle (0.4)
	(7,1) circle (0.25) circle (0.4)
	(7,7) circle (0.25) circle (0.4);
}}

\tikzset{variable_terminals/.pic={
	\path pic (v) {variable};
	\node[lcirc, cyan] (-t1) at (0,4) [] {};
	\node[lcirc, cyan] (-t2) at (4,0) [] {};
	\node[lcirc, cyan] (-t3) at (8,4) [] {};
    \draw[cyan] (2,6) -- (-t1) -- (2,2) -- (-t2) -- (6,2) -- (-t3) -- (6,6);
}}

\tikzset{clause_oppose/.pic={
    \draw[thick,<->,>=stealth] (0,2) -- (2,2) -- (2,0) node[near start,above] {};
    \draw[thick,->,>=stealth] (-1,1) -- (1,1) -- (1,2) node[near start,above] {};
    \draw[thick,->,>=stealth] (0,3) -- (0,1) node[near start,above] {};
    \draw[thick,->,>=stealth] (3,-1) -- (3,1) -- (2,1) node[near start,above] {};
    \draw[thick,->,>=stealth] (1,0) -- (3,0) node[near start,above] {};
    \draw[thick,->,>=stealth] (2,4) -- (2,2) node[near start,above] {};
    \draw[thick,->,>=stealth] (4,2) -- (2,2) node[near start,above] {};
    
    \draw[red,thick,<->,>=stealth] (2,3) -- (3,3) -- (3,2);
    \draw[red,thick,<->,>=stealth] (-3,0) -- (0,0) -- (0,1);
    \draw[red,thick,<->,>=stealth] (6,0) -- (3,0);
    
    \draw
    (-1,1) circle (0.2) circle (0.3)
    (0,3) circle (0.2) circle (0.3)
    (3,-1) circle (0.2) circle (0.3)
    (1,0) circle (0.2) circle (0.3)
    (2,4) circle (0.2) circle (0.3)
    (4,2) circle (0.2) circle (0.3);
}}

\tikzset{clause_other/.pic={
    \draw[thick,<->,>=stealth] (0,-2) -- (0,0) -- (2,0) node[near start,above] {};
    \draw[thick,->,>=stealth] (-1,-2) -- (1,-2) node[near start,above] {};
    \draw[thick,->,>=stealth] (1,-3) -- (1,-1) -- (0,-1) node[near start,above] {};
    \draw[thick,->,>=stealth] (2,-1) -- (2,1) node[near start,above] {};
    \draw[thick,->,>=stealth] (3,1) -- (1,1) -- (1,0) node[near start,above] {};
    \draw[thick,->,>=stealth] (0,2) -- (0,0) node[near start,above] {};
    \draw[thick,->,>=stealth] (-2,0) -- (0,0) node[near start,above] {};
    
    \draw[red,thick,<->,>=stealth] (-1,0) -- (-1,1) -- (0,1);
    \draw[red,thick,<->,>=stealth] (2,4) -- (2,1);
    \draw[red,thick,<->,>=stealth] (5,-2) -- (1,-2);
    
    \draw
    (-1,-2) circle (0.2) circle (0.3)
    (1,-3) circle (0.2) circle (0.3)
    (2,-1) circle (0.2) circle (0.3)
    (3,1) circle (0.2) circle (0.3)
    (0,2) circle (0.2) circle (0.3)
    (-2,0) circle (0.2) circle (0.3);
}}

\tikzset{001_oppose/.pic={
    \draw[thick,<->,>=stealth] (0,1) -- (1,1) node[near start,above] {};
    \draw[thick,->,>=stealth] (-2,2) -- (1,2) -- (1,1) node[near start,above] {};
    \draw[thick,->,>=stealth] (-2,1) -- (0,1) -- (0,2) node[near start,above] {};
    \draw[thick,->,>=stealth] (3,1) -- (1,1) -- (1,0) node[near start,above] {};
    \draw[thick,->,>=stealth] (3,0) -- (0,0) -- (0,1) node[near start,above] {};
    
    \draw[red,thick,<->,>=stealth] (-1,1) -- (-1,0) -- (0,0);
    \draw[red,thick,<->,>=stealth] (1,2) -- (2,2) -- (2,1);
    \draw[red,thick,<->,>=stealth] (0,5) -- (0,2);
    
    \draw
    (-2,2) circle (0.2) circle (0.3)
    (-2,1) circle (0.2) circle (0.3)
    (3,1) circle (0.2) circle (0.3)
    (3,0) circle (0.2) circle (0.3);
}}

\tikzset{001_other/.pic={
    \draw[thick,<->,>=stealth] (1,0) -- (1,1) node[near start,above] {};
    \draw[thick,->,>=stealth] (0,-2) -- (0,1) -- (1,1) node[near start,above] {};
    \draw[thick,->,>=stealth] (1,-2) -- (1,0) -- (0,0) node[near start,above] {};
    \draw[thick,->,>=stealth] (1,3) -- (1,1) -- (2,1) node[near start,above] {};
    \draw[thick,->,>=stealth] (2,3) -- (2,0) -- (1,0) node[near start,above] {};
    
    \draw[red,thick,<->,>=stealth] (1,-1) -- (2,-1) -- (2,0);
    \draw[red,thick,<->,>=stealth] (0,1) -- (0,2) -- (1,2);
    \draw[red,thick,<->,>=stealth] (-3,0) -- (0,0);
    
    \draw
    (0,-2) circle (0.2) circle (0.3)
    (1,-2) circle (0.2) circle (0.3)
    (1,3) circle (0.2) circle (0.3)
    (2,3) circle (0.2) circle (0.3);
}}

\tikzset{110_oppose/.pic={
    \draw[thick,<->,>=stealth] (1,0) -- (1,1) node[near start,above] {};
    \draw[thick,->,>=stealth] (0,-1) -- (0,1) -- (1,1) node[near start,above] {};
    \draw[thick,->,>=stealth] (1,-1) -- (1,0) -- (0,0) node[near start,above] {};
    \draw[thick,->,>=stealth] (1,3) -- (1,1) -- (2,1) node[near start,above] {};
    \draw[thick,->,>=stealth] (2,3) -- (2,0) -- (1,0) node[near start,above] {};
    
    \draw[red,thick,<->,>=stealth] (0,1) -- (0,2) -- (1,2);
    \draw[red,thick,<->,>=stealth] (-3,0) -- (0,0);
    \draw[red,thick,<->,>=stealth] (6,1) -- (2,1);
    
    \draw
    (0,-1) circle (0.2) circle (0.3)
    (1,-1) circle (0.2) circle (0.3)
    (1,3) circle (0.2) circle (0.3)
    (2,3) circle (0.2) circle (0.3);
}}

\tikzset{110_other/.pic={
    \draw[thick,<->,>=stealth] (1,0) -- (1,1) node[near start,above] {};
    \draw[thick,->,>=stealth] (0,-2) -- (0,1) -- (1,1) node[near start,above] {};
    \draw[thick,->,>=stealth] (1,-2) -- (1,0) -- (0,0) node[near start,above] {};
    \draw[thick,->,>=stealth] (1,2) -- (1,1) -- (2,1) node[near start,above] {};
    \draw[thick,->,>=stealth] (2,2) -- (2,0) -- (1,0) node[near start,above] {};
    
    \draw[red,thick,<->,>=stealth] (1,-1) -- (2,-1) -- (2,0);
    \draw[red,thick,<->,>=stealth] (-3,0) -- (0,0);
    \draw[red,thick,<->,>=stealth] (3,5) -- (3,1) -- (2,1);
    
    \draw
    (0,-2) circle (0.2) circle (0.3)
    (1,-2) circle (0.2) circle (0.3)
    (1,2) circle (0.2) circle (0.3)
    (2,2) circle (0.2) circle (0.3);
}}

\theoremstyle{plain}
\newtheorem{theorem}[equation]{Theorem}
\newtheorem{corollary}[equation]{Corollary}
\newtheorem{lemma}[equation]{Lemma}
\newtheorem{observation}{Observation}{\normalfont\bfseries}{\itshape}
\newtheorem{Claim}{Claim}{\normalfont\bfseries}{\itshape}
\newtheorem{fact}{Fact}
\newenvironment{cproof}[1]{\par\indent{\textit{Proof.}}\space#1}{\hfill $\diamondsuit$}

\begin{document}
\begin{frontmatter}
%===========================================================

\title{CPG graphs: Some structural and hardness results}

\author[1]{Nicolas Champseix}
\author[2]{Esther Galby}\ead{esther.galby@unifr.ch}
\author[3]{Andrea Munaro}\ead{a.munaro@qub.ac.uk}
\author[2]{Bernard Ries}\ead{bernard.ries@unifr.ch}

\address[1]{\'Ecole Normale Sup\'erieure de Lyon, France}
\address[2]{University of Fribourg, Switzerland}
\address[3]{Queen's University Belfast, United Kingdom}

%===========================================================

\begin{abstract}
In this paper we continue the systematic study of \textit{Contact graphs of Paths on a Grid} (\textit{CPG graphs}) initiated in \cite{cpg}. A CPG graph is a graph for which there exists a collection of pairwise interiorly disjoint paths on a grid in one-to-one correspondence with its vertex set such that two vertices are adjacent if and only if the corresponding paths touch at a grid-point. If every such path has at most $k$ bends for some $k \geq 0$, the graph is said to be $B_k$-CPG. 

We first show that, for any $k \geq 0$, the class of $B_k$-CPG graphs is strictly contained in the class of $B_{k+1}$-CPG graphs even within the class of planar graphs, thus implying that there exists no $k \geq 0$ such that every planar CPG graph is $B_k$-CPG. The main result of the paper is that recognizing CPG graphs and $B_k$-CPG graphs with $k \geq 1$ is $\mathsf{NP}$-complete. Moreover, we show that the same remains true even within the class of planar graphs in the case $k \geq 3$. We then consider several graph problems restricted to CPG graphs and show, in particular, that {\sc Independent Set} and {\sc Clique Cover} remain $\mathsf{NP}$-hard for $B_0$-CPG graphs. Finally, we consider the related classes $B_k$-EPG of edge-intersection graphs of paths with at most $k$ bends on a grid. Although it is possible to optimally color a $B_0$-EPG graph in polynomial time, as this class coincides with that of interval graphs, we show that, in contrast, {\sc 3-Colorability} is $\mathsf{NP}$-complete for $B_1$-EPG graphs. 
\end{abstract}

\begin{keyword} CPG graphs \sep EPG graphs \sep Planar graphs \sep Recognition \sep $\mathsf{NP}$-hardness  
\end{keyword}

\end{frontmatter}
%===========================================================

\section{Introduction}

\citet{golumbic} introduced the class of \textit{Edge intersection graphs of Paths on a Grid} (\textit{EPG} for short) that is, graphs for which there exists a collection of paths on a grid in one-to-one correspondence with their vertex set such that two vertices are adjacent if and only if the corresponding paths intersect on at least one grid-edge. Since every graph is EPG \cite{golumbic}, a natural restriction which was forthwith considered consists in limiting the number of \textit{bends} (90 degree turns at a grid-point) that the paths may have: a graph is $B_k$-EPG, for some integer $k \geq 0$, if it has an EPG representation in which each path has at most $k$ bends. 

Later on, \citet{asinowski} considered the closely related class of \textit{Vertex intersection graphs of Paths on a Grid} (\textit{VPG} for short). Likewise, the vertices of such graphs may be represented by paths on a grid but two vertices are adjacent if and only if the corresponding paths intersect on at least one grid-point. In fact, this class coincides with that of intersection graphs of curves in the plane, also known as string graphs \citep{asinowski}. Similarly to EPG, they then defined the class $B_k$-VPG, for $k \geq 0$, consisting of those VPG graphs admitting a VPG representation where each path has at most $k$ bends. 

Since their introduction, EPG and VPG graphs have been extensively studied (see, e.g., \citep{alcon,NCA,asinowski,biedl,bonomo,cohen,epstein,felsner,francis,heldt1,heldt2,PR17}). 
One of the main considered problems consists in determining the \textit{bend-number with respect to EPG} (resp. \textit{VPG}) \textit{representations} of a given graph class $\mathcal{G}$: the bend number of $\mathcal{G}$ with respect to EPG (resp. VPG) representations is the minimum integer $k \geq 0$ such that every graph in $\mathcal{G}$ is $B_k$-EPG (resp. $B_k$-VPG). For instance, \citet{gonc} showed that planar graphs have a bend number with respect to VPG representations of $1$; the bend number of planar graphs with respect to EPG representations has yet to be determined (it is either $3$ or $4$ by~\cite{heldt1}).

In this paper, we are interested in the contact counterpart of these graph classes, namely the class of \textit{Contact graphs of Paths on a Grid} (\textit{CPG} for short): a graph $G=(V,E)$ is a \textit{CPG graph} if there exists a collection $\mathcal{P}$ of pairwise interiorly disjoint paths on a grid $\mathcal{G}$ in one-to-one correspondence with $V$ such that two vertices are adjacent if and only if the corresponding paths touch. If every such path has at most $k$ bends, for some $k \geq 0$, then the graph is \textit{$B_k$-CPG}. The pair $\mathcal{R} = (\mathcal{G}, \mathcal{P})$ is a \textit{CPG representation of $G$} and, more specifically, a \textit{$k$-bend CPG representation of $G$} if every path in $\mathcal{P}$ has at most $k$ bends. The \textit{bend number with respect to CPG representations} is defined mutatis mutandis. Clearly, any $B_k$-CPG graph is also a $B_k$-VPG graph and a $B_{k+2}$-EPG graph. 

Although CPG graphs have been considered in the literature (see, e.g., \cite{chaplick12,chaplick13,felsner,planbip,gonc,kobourov}), a systematic study of this class was first initiated in~\cite{cpg}. We note that \citet{aerts} studied the family of graphs admitting a \textit{Vertex Contact representation of Paths on a Grid}, which constitutes a subclass of CPG graphs as paths in such representations are not allowed to share a common endpoint. Such graphs are easily seen to be planar and in fact this class is strictly contained in that of planar CPG graphs \cite{cpg}. It is then natural to ask whether every planar graph is CPG and whether there exists $k \geq 0$ such that every planar CPG graph is $B_k$-CPG. The former was answered in the negative in \cite{cpg}. The latter is settled in \Cref{sec:planar} by providing, for any $k \geq 0$, a planar graph which is $B_{k+1}$-CPG but not $B_k$-CPG. Since every planar graph is $B_1$-VPG \cite{gonc}, it follows that for any $k \geq 1$, $B_k\mbox{-CPG} \subsetneq B_k\mbox{-VPG}$ even within the class of planar graphs (and in fact, it is easy to see that this remains true for $k=0$). 

Great attention has been devoted to the complexity of recognizing (subclasses of) VPG and EPG graphs. For each fixed $k \geq 0$, recognizing $B_{k}$-VPG graphs is $\mathsf{NP}$-complete \citep{CJKV12,Kra94}. Similarly, recognizing VPG graphs is $\mathsf{NP}$-hard \citep{Kra91} and in $\mathsf{NP}$ \citep{SSS03}. Notice that membership in $\mathsf{NP}$ remained elusive for a while as, contrary to the case of fixed $k$ where as a polynomial certificate it is enough to provide a list of coordinates corresponding to endpoints and bend-points of paths, there exist string graphs on $O(n)$ vertices whose every string representation contains at least $2^{n}$ intersection points between the curves \citep{KM91}. Since $B_{0}$-EPG graphs are exactly interval graphs, they can be recognized in polynomial time \citep{BL76}. On the other hand, recognizing $B_{k}$-EPG graphs is $\mathsf{NP}$-complete for $k = 1, 2$ \citep{heldt2,PR17} and it is an open problem to settle the case $k \geq 3$. In \cite{cpg}, the authors showed that \textsc{Recognition} is $\mathsf{NP}$-complete for $B_0$-CPG graphs. In \Cref{sec:recognition}, we completely settle the remaining cases by showing that \textsc{Recognition} is $\mathsf{NP}$-complete for $B_k$-CPG graphs with $k \geq 1$. For $k \geq 3$, we show that this remains true even within the class of planar graphs. As a corollary of these results, we show that recognizing CPG graphs and planar CPG graphs are $\mathsf{NP}$-complete problems.

\citet{asinowski} provided several complexity results for problems restricted to $B_0$-VPG graphs. More specifically, they showed that {\sc Independent Set}, {\sc Hamiltonian Cycle}, {\sc Hamiltonian Path} and {\sc Colorability} are all $\mathsf{NP}$-complete, whereas {\sc Clique} is polynomial-time solvable. It is then natural to consider the complexity of such problems when restricted to $B_0$-CPG graphs. The fact that every planar bipartite graph is $B_0$-CPG \cite{planbip} immediately implies that problems which are $\mathsf{NP}$-complete for this class, such as {\sc Dominating Set}, {\sc Feedback Vertex Set}, {\sc Hamiltonian Cycle} and {\sc Hamiltonian Path} (see, e.g., \cite{clark,munaro1}), remain $\mathsf{NP}$-complete in $B_0$-CPG. It was furthermore shown in \cite{cpg} that {\sc 3-Colorability} is $\mathsf{NP}$-complete in $B_0$-CPG. In \Cref{sec:complexity}, we consider {\sc Independent Set} and {\sc Clique Cover} and show that they are $\mathsf{NP}$-complete in $B_0$-CPG.

It was shown in \cite{epstein} that \textsc{Colorability} is $\mathsf{NP}$-complete in $B_1$-EPG. On the other hand, the problem is polynomial-time solvable in $B_0$-EPG as this class coincides with that of interval graphs. In \Cref{sec:b1epg}, we provide a complexity dichotomy for the related {\sc 3-Colorability} problem restricted to EPG graphs. The problem is $\mathsf{NP}$-complete in $B_k$-EPG, for each $k \geq 2$, since it is $\mathsf{NP}$-complete in $B_0$-CPG \citep{cpg}. We complete the picture by showing that {\sc 3-Colorability} is $\mathsf{NP}$-complete in $B_1$-EPG.

We remark that all our complexity results hold even if a CPG (resp. EPG) representation of the graph is given.
%===========================================================

\section{Preliminaries}
\label{sec:prel}

Throughout the paper, the considered graphs are undirected, finite and simple. For any graph theoretical notion not defined here, we refer the reader to \cite{diestel}.

Let $G$ be a graph with vertex set $V$ and edge set $E$. The \textit{degree} of a vertex $v\in V$ is the size of its neighborhood and a vertex is \textit{cubic} if it has degree 3. For $k \geq 1$, the \textit{$k$-subdivision} of an edge $uv \in E$ consists in replacing $uv$ by a path $uv_1\cdots v_kv$, where $v_1, \ldots, v_k$ are new vertices; the $k$-subdivision of $G$ is the graph obtained by $k$-subdividing every edge of $G$. The \textit{line graph} $L(G)$ of $G$ has vertex set $E$ and two vertices $e$ and $e'$ are adjacent if and only if they have a common endvertex in $G$. 

A graph is \textit{subcubic} if every vertex has degree at most 3. A graph is \textit{planar} if it admits an embedding in the plane where no two edges cross. Given a graph $H$, we say that a graph is \textit{H-free} if it contains no induced subgraph isomorphic to $H$. The \textit{triangle} is the complete graph on three vertices and the \textit{diamond} is the graph obtained by removing an edge from the complete graph on four vertices.

A subset $S \subseteq V$ is an \textit{independent set} of $G$ if any two vertices in $S$ are nonadjacent; the maximum size of an independent set of $G$ is denoted by $\alpha(G)$. A subset $V' \subseteq V$ is a \textit{vertex cover} of $G$ if each edge of $G$ is incident to at least one vertex in $V'$; the minimum size of a vertex cover of $G$ is denoted by $\beta(G)$. A subset $K \subseteq V$ is a \textit{clique} if any two vertices in $K$ are adjacent. A \textit{clique cover} of $G$ is a set of cliques such that each vertex of $G$ belongs to at least one of them; the minimum number of cliques in a clique cover of $G$ is denoted by $\theta(G)$. 

Let $G=(V,E)$ be a CPG graph and $\mathcal{R} = (\mathcal{G},\mathcal{P})$ be a CPG representation of $G$. The grid-point lying on row $x$ and column $y$ is denoted by $(x, y)$. The path in $\mathcal{P}$ representing some vertex $u \in V$ is denoted by $P_u$. A \textit{segment} of a path is either a vertical or horizontal line segment in the polygonal curve constituting the path. We say that a path $P$ \textit{touches} a path $P'$ if one endpoint of $P$ belongs to $P'$, and that $P$ and $P'$ \textit{touch} if $P$ and $P'$ have a common grid-point. We denote by $\partial(P)$ the endpoints of a path $P$. An \textit{interior point} of a path $P$ is a point belonging to $P$ and different from its endpoints; the \textit{interior} of $P$, denoted by $\mathring{P}$, is the set of all its interior points. A grid-point is said to be a \textit{contact point} if it belongs to more than one path, and a \textit{$k$-contact point} if it belongs to exactly $k$ paths in $\mathcal{R}$. We denote by $T(\mathcal{R})$ the set of contact points in $\mathcal{R}$. A grid-point is of \textit{type II.a} if it is a 3-contact point, an interior point of one path but a bend-point for no path; a grid-point is of \textit{type II.b} if it is a 3-contact point and the bend-point of some path (see \Cref{intpoint}). 

\begin{figure}[ht]

\centering

\begin{subfigure}[b]{.45\textwidth}

\centering

\begin{tikzpicture}[scale=.7]

\node at (0,0) (p) [label=below left:{\scriptsize $p$}] {};

\draw[thick,-,>=stealth] (0,1)--(0,-1);

\draw[thick,dotted] (0,1.2)--(0,-1.2);

\draw[thick,<-,>=stealth] (0,0)--(1,0);

\draw[thick,dotted] (1,0)--(1.2,0);

\draw[thick,<-,>=stealth] (0,0)--(-1,0);

\draw[thick,dotted] (-1,0)--(-1.2,0);

\end{tikzpicture}

\caption{Type II.a}

\label{intpointI}

\end{subfigure}
\hspace*{1cm}
\begin{subfigure}[b]{.45\textwidth}

\centering

\begin{tikzpicture}[scale=.7]

\node at (3,0) (p) [label=below left:{\scriptsize $p$}] {};

\draw[thick,-,>=stealth] (3,1)--(3,0)--(4,0);

\draw[thick,dotted] (3,1)--(3,1.2);

\draw[thick,dotted] (4,0)--(4.2,0);

\draw[thick,<-,>=stealth] (3,0)--(3,-1);

\draw[thick,dotted] (3,-1)--(3,-1.2);

\draw[thick,<-,>=stealth] (3,0)--(2,0);

\draw[thick,dotted] (1.8,0)--(2,0);

\end{tikzpicture}

\caption{Type II.b}

\label{intpointII}

\end{subfigure}

\caption{Type II grid-points.}

\label{intpoint}

\end{figure}
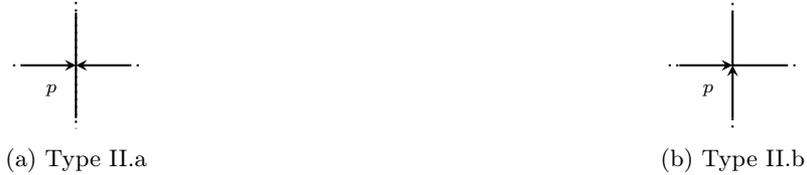

We say that an endpoint of a path is \textit{free} in the representation $\mathcal{R}$ if it does not belong to any other path in $\mathcal{R}$. For a vertex $u \in V$, the \textit{weight of $u$ with respect to} $\mathcal{R}$, denoted by $w_u$, is defined as follows. Let $q_u^i$ ($i=1,2$) be the endpoints of the corresponding path $P_u$ and consider, for $i=1,2$, 
\[w_u^i = |\{P \in \mathcal{P} : q_u^i \in \mathring{P}\}| + \frac{1}{2} \cdot |\{P \in \mathcal{P} : P \neq P_u \text{ and } q_u^i \in \partial(P)\}|.\]
Then $w_u = w_u^1 + w_u^2$. It is easy to see that $w_u^i \in \{0,\frac{1}{2},1,\frac{3}{2}\}$, for any $u \in V$ and $i=1,2$, and that $|E| \leq \sum_{u \in V} w_u$ (see \cite{cpg}). 

Given a $B_k$-CPG graph $G$ and a $k$-bend representation $\mathcal{R} = (\mathcal{G}, \mathcal{P})$ of $G$, the grid $\mathcal{G}$ and the paths in $\mathcal{P}$ are encoded as follows. For the grid $\mathcal{G}$, we only keep track of the \textit{grid-step $\sigma$}, that is, the length of a grid-edge. A path $P \in \mathcal{P}$ is given by a sequence $s(P)$ of grid-points $(x_1,y_1), (x_2,y_2), \dots, (x_{\ell_P},y_{\ell_P})$, where $(x_1,y_1)$ and $(x_{\ell_P},y_{\ell_P})$ are the endpoints of $P$ and all the other points are the bend-points of $P$ in their order of appearance while traversing $P$ from $(x_1,y_1)$ to $(x_{\ell_P},y_{\ell_P})$ i.e., for any $i \in [\ell_P -1]$, $[(x_i,y_i),(x_{i+1},y_{i+1})]$ is a segment of $P$. Clearly, $\ell_P \leq k + 2$. 

The \textit{length} of a path in $\mathcal{P}$ is the number of grid-edges it contains. The \textit{refinement of a grid $\mathcal{G}$} having grid-step $\sigma$ is the operation adding a new column (resp. row) between any pair of consecutive columns (resp. row) in $\mathcal{G}$ and setting the grid-step to $\sigma/2$. Notice that this operation does not change the sequences $s(P)$. 

%===========================================================

\section{Planar CPG graphs}
\label{sec:planar}

In this section we show that, for any $k \geq 0$, the class of $B_{k+1}$-CPG graphs strictly contains that of $B_k$-CPG graphs, even within the class of planar graphs. To this end, for $k \geq 0$, consider the planar graph $G_k$ depicted in \Cref{fig:Gk}. We refer to the vertices $\alpha_i$, for $1 \leq i \leq 20$, as the \textit{secondary vertices}, and to the vertices $u_j^i$, for $1 \leq j \leq k+2$ and a given $1 \leq i \leq 19$, as the \textit{$(i,i+1)$-sewing vertices}. A $(k+1)$-bend CPG representation of $G_k$ is given in \Cref{fig:representation} (the blue paths correspond to sewing vertices and the red paths correspond to secondary vertices).

\begin{figure}[htb]
\begin{center}
\begin{subfigure}{.6\textwidth}
\centering
\begin{tikzpicture}[scale=0.35]
\fill[green!20!white] (-4.75,0) ellipse (1.25cm and 1cm);
\fill[green!20!white] (-2.25,0) ellipse (1.25cm and 1cm);
\fill[green!20!white] (4.75,0) ellipse (1.25cm and 1cm);
\node at (-4.75,0) {$H_1$};
\node at (-2.25,0) {$H_2$};
\node at (4.75,0) {$H_{19}$};

\node[circ] (a) at (0,3) [label=above:{\tiny $a$}] {};
\node[circ] (b) at (0,-3) [label=below:{\tiny $b$}] {};

\node[circ] (a1) at (-6,0) [label=below left:{\tiny $\alpha_1$}] {};
\node[circ] (a2) at (-3.5,0) [label=below right:{\tiny $\alpha_2$}] {};
\node[circ] (a3) at (-1,0) [label=below right:{\tiny $\alpha_3$}] {};
\node[circ] (a19) at (3.5,0) [label=below left:{\tiny $\alpha_{19}$}] {};
\node[circ] (a20) at (6,0) [label=below right:{\tiny $\alpha_{20}$}] {};

\draw (a) .. controls (-6,2) .. (a1);
\draw (a) .. controls (-3.5,2) .. (a2);
\draw (a) .. controls (-1,2) .. (a3);
\draw (a) .. controls (3.5,2) .. (a19);
\draw (a) .. controls (6,2) .. (a20);
\draw (a) .. controls (-10,4) and (-10,-4) .. (b);
\draw (b) .. controls (-6,-2) .. (a1);
\draw (b) .. controls (-3.5,-2) .. (a2);
\draw (b) .. controls (-1,-2) .. (a3);
\draw (b) .. controls (3.5,-2) .. (a19);
\draw (b) .. controls (6,-2) .. (a20);

\node[scirc] at (1,0) {};
\node[scirc] at (1.25,0) {};
\node[scirc] at (1.5,0) {};

\node[invisible] at (2,-5.5) {};
\node[invisible] at (2,4.8) {};
\end{tikzpicture}
\caption{The planar graph $G_k$.}
\label{fig:overall}
\end{subfigure}
\hspace*{1cm}
\begin{subfigure}{.3\textwidth}
\centering
\begin{tikzpicture}[scale=0.47]
\node[circ] (ai) at (0,0) [label=left:{\tiny $\alpha_i$}] {};
\node[circ] (ai1) at (4,0) [label=right:{\tiny $\alpha_{i+1}$}] {};

\node[circ] (u1) at (2,3) [label=above right:{\tiny $u_1^i$}] {};
\node[circ] (u2) at (2,2) [label=above right:{\tiny $u_2^i$}] {};
\node[circ] (u3) at (2,1) [label=above right:{\tiny $u_3^i$}] {};
\node[circ] (uk1) at (2,-2) [label=above right:{\tiny $u_{k+1}^i$}] {};
\node[circ] (uk2) at (2,-3) [label=below right:{\tiny $u_{k+2}^i$}] {};

\draw (ai) .. controls (0.5,3) .. (u1);
\draw (ai) .. controls (0.5,2) .. (u2);
\draw (ai) .. controls (0.5,1) .. (u3);
\draw (ai) .. controls (0.5,-2) .. (uk1);
\draw (ai) .. controls (0.5,-3) .. (uk2);
\draw (ai1) .. controls (3.5,3) .. (u1);
\draw (ai1) .. controls (3.5,2) .. (u2);
\draw (ai1) .. controls (3.5,1) .. (u3);
\draw (ai1) .. controls (3.5,-2) .. (uk1);
\draw (ai1) .. controls (3.5,-3) .. (uk2);

\draw[-]
(u1) -- (u2) -- (u3)
(uk1) -- (uk2);

\node[scirc] at (2,0) {};
\node[scirc] at (2,-0.5) {};
\node[scirc] at (2,-1) {};

\end{tikzpicture}
\caption{The gadget $H_i$.}
\label{fig:gadget}
\end{subfigure}
\caption{The construction of $G_k$.}
\label{fig:Gk}
\end{center}
\end{figure}
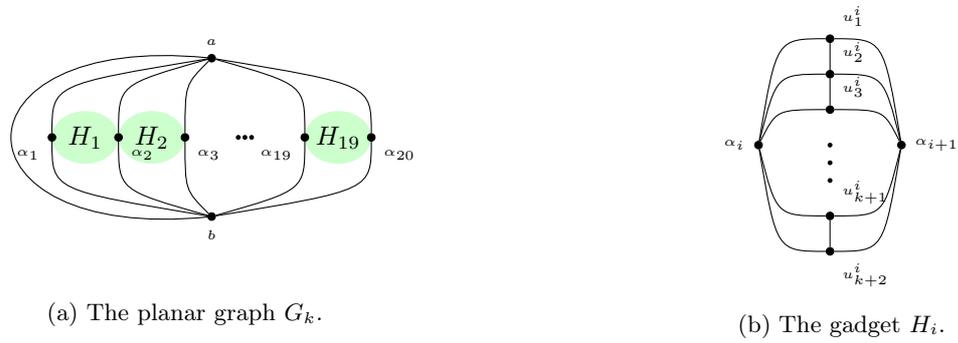

Given a CPG representation $\mathcal{R} = (\mathcal{G},\mathcal{P})$ of $G_k$, a path in $\mathcal{P}$ corresponding to a secondary vertex (resp. an $(i,i+1)$-sewing vertex) is called a \textit{secondary path} (resp. an \textit{$(i,i+1)$-sewing path}). A secondary path $P_{\alpha_i}$ is \textit{pure} if no endpoint of $P_a$ or $P_b$ belongs to $P_{\alpha_i}$. We now use the graph $G_k$ to prove the following.

\begin{figure}
\centering
\begin{subfigure}{.45\linewidth}
\centering
\begin{tikzpicture}[scale=0.3]
\draw[thick,<->,>=stealth] (12,-1) -- (1,-1) node[left] {\tiny $P_a$};
\draw[thick,<->,>=stealth] (12,-1) -- (12,10) node[above] {\tiny $P_b$};

\draw[red,thick,<-,>=stealth] (2,-1) -- (2,5) node[midway,left] {\tiny $P_{\alpha_1}$} -- (3,5) -- (3,6) -- (4,6) -- (4,6.5);
\draw[red,thick,<-,>=stealth] (3,-1) -- (3,4) -- (4,4) -- (4,5) -- (5,5) -- (5,5.5);
\draw[red,thick,<-,>=stealth] (4,-1) -- (4,3) -- (5,3) -- (5,4) -- (6,4) -- (6,4.5);
\draw[red,thick,<-,>=stealth] (6,-1) -- (6,1) -- (7,1) -- (7,2) -- (8,2) -- (8,2.5);
\draw[red,thick,<-,>=stealth] (7,-1) -- (7,0) -- (8,0) -- (8,1) node[midway,right] {\tiny $P_{\alpha_{20}}$} -- (9,1) -- (9,1.5);

\draw[red,thick,->,>=stealth] (5.5,8) -- (6,8) -- (6,9) -- (12,9);
\draw[red,thick,->,>=stealth] (6.5,7) -- (7,7) -- (7,8) -- (12,8);
\draw[red,thick,->,>=stealth] (7.5,6) -- (8,6) -- (8,7) -- (12,7);
\draw[red,thick,->,>=stealth] (9.5,4) -- (10,4) -- (10,5) -- (12,5);
\draw[red,thick,->,>=stealth] (10.5,3) -- (11,3) -- (11,4) -- (12,4);

\draw[cyan,thick,<->,>=stealth] (2,4) -- (3,4);
\draw[cyan,thick,<->,>=stealth] (3,5) -- (3,4);
\draw[cyan,thick,<->,>=stealth] (3,5) -- (4,5);
\draw[cyan,thick,<->,>=stealth] (4,6) -- (4,5);
\draw[cyan,thick,<-,>=stealth] (4,6) -- (4.5,6);
\draw[cyan,thick,<-,>=stealth] (6,8) -- (6,7.5);
\draw[cyan,thick,<->,>=stealth] (6,8) -- (7,8);
\draw[cyan,thick,<->,>=stealth] (7,9) -- (7,8);

\draw[cyan,thick,<->,>=stealth] (3,3) -- (4,3);
\draw[cyan,thick,<->,>=stealth] (4,4) -- (4,3);
\draw[cyan,thick,<->,>=stealth] (4,4) -- (5,4);
\draw[cyan,thick,<->,>=stealth] (5,5) -- (5,4);
\draw[cyan,thick,<-,>=stealth] (5,5) -- (5.5,5);
\draw[cyan,thick,<-,>=stealth] (7,7) -- (7,6.5);
\draw[cyan,thick,<->,>=stealth] (7,7) -- (8,7);
\draw[cyan,thick,<->,>=stealth] (8,8) -- (8,7);

\draw[cyan,thick,<-,>=stealth] (4,2) -- (4.5,2);
\draw[cyan,thick,<-,>=stealth] (5,3) -- (5,2.5);
\draw[cyan,thick,<-,>=stealth] (5,3) -- (5.5,3);
\draw[cyan,thick,<-,>=stealth] (6,4) -- (6,3.5);
\draw[cyan,thick,<-,>=stealth] (6,4) -- (6.5,4);
\draw[cyan,thick,<-,>=stealth] (8,6) -- (8,5.5);
\draw[cyan,thick,<-,>=stealth] (8,6) -- (8.5,6);
\draw[cyan,thick,<-,>=stealth] (9,7) -- (9,6.5);

\draw[cyan,thick,->,>=stealth] (5.5,1) -- (6,1);
\draw[cyan,thick,->,>=stealth] (6,1.5) -- (6,1);
\draw[cyan,thick,->,>=stealth] (6.5,2) -- (7,2);
\draw[cyan,thick,->,>=stealth] (7,2.5) -- (7,2);
\draw[cyan,thick,->,>=stealth] (9.5,5) -- (10,5);
\draw[cyan,thick,->,>=stealth] (10,5.5) -- (10,5);

\draw[cyan,thick,<->,>=stealth] (6,0) -- (7,0);
\draw[cyan,thick,<->,>=stealth] (7,1) -- (7,0);
\draw[cyan,thick,<->,>=stealth] (7,1) -- (8,1);
\draw[cyan,thick,<->,>=stealth] (8,2) -- (8,1);
\draw[cyan,thick,<-,>=stealth] (8,2) -- (8.5,2);
\draw[cyan,thick,<-,>=stealth] (10,4) -- (10,3.5);
\draw[cyan,thick,<->,>=stealth] (10,4) -- (11,4);
\draw[cyan,thick,<->,>=stealth] (11,5) -- (11,4);

\node[scirc] at (4.7,0) {};
\node[scirc] at (5,0) {};
\node[scirc] at (5.3,0) {};

\node[scirc] at (11,5.7) {};
\node[scirc] at (11,6) {};
\node[scirc] at (11,6.3) {};

\node[scirc] at (5.7,5.7) {};
\node[scirc] at (6,6) {};
\node[scirc] at (6.3,6.3) {};

\node[scirc] at (9,2.4) {};
\node[scirc] at (9.3,2.7) {};
\node[scirc] at (9.6,3) {};
\end{tikzpicture}
\caption{$k$ is even.}
\end{subfigure}
\hspace*{1cm}
\begin{subfigure}{.45\linewidth}
\centering
\begin{tikzpicture}[scale=0.3]
\draw[thick,<->,>=stealth] (13,-1) -- (1,-1) node[left] {\tiny $P_a$};
\draw[thick,<->,>=stealth] (13,-1) -- (13,10) -- (1,10) node[left] {\tiny $P_b$};

\draw[red,thick,<-,>=stealth] (2,-1) -- (2,5) node[midway,left] {\tiny $P_{\alpha_1}$} -- (3,5) -- (3,6) -- (4,6) -- (4,6.5);
\draw[red,thick,<-,>=stealth] (3,-1) -- (3,4) -- (4,4) -- (4,5) -- (5,5) -- (5,5.5);
\draw[red,thick,<-,>=stealth] (4,-1) -- (4,3) -- (5,3) -- (5,4) -- (6,4) -- (6,4.5);
\draw[red,thick,<-,>=stealth] (6,-1) -- (6,1) -- (7,1) -- (7,2) -- (8,2) -- (8,2.5);
\draw[red,thick,<-,>=stealth] (7,-1) -- (7,0) -- (8,0) -- (8,1) node[midway,right] {\tiny $P_{\alpha_{20}}$} -- (9,1) -- (9,1.5);

\draw[red,thick,->,>=stealth] (5.5,8) -- (6,8) -- (6,9) -- (7,9) -- (7,10);
\draw[red,thick,->,>=stealth] (6.5,7) -- (7,7) -- (7,8) -- (8,8) -- (8,10);
\draw[red,thick,->,>=stealth] (7.5,6) -- (8,6) -- (8,7) -- (9,7) -- (9,10);
\draw[red,thick,->,>=stealth] (9.5,4) -- (10,4) -- (10,5) -- (11,5) -- (11,10);
\draw[red,thick,->,>=stealth] (10.5,3) -- (11,3) -- (11,4) -- (12,4) -- (12,10);

\draw[cyan,thick,<->,>=stealth] (2,4) -- (3,4);
\draw[cyan,thick,<->,>=stealth] (3,5) -- (3,4);
\draw[cyan,thick,<->,>=stealth] (3,5) -- (4,5);
\draw[cyan,thick,<->,>=stealth] (4,6) -- (4,5);
\draw[cyan,thick,<-,>=stealth] (4,6) -- (4.5,6);
\draw[cyan,thick,<-,>=stealth] (6,8) -- (6,7.5);
\draw[cyan,thick,<->,>=stealth] (6,8) -- (7,8);
\draw[cyan,thick,<->,>=stealth] (7,9) -- (7,8);
\draw[cyan,thick,<->,>=stealth] (7,9) -- (8,9);

\draw[cyan,thick,<->,>=stealth] (3,3) -- (4,3);
\draw[cyan,thick,<->,>=stealth] (4,4) -- (4,3);
\draw[cyan,thick,<->,>=stealth] (4,4) -- (5,4);
\draw[cyan,thick,<->,>=stealth] (5,5) -- (5,4);
\draw[cyan,thick,<-,>=stealth] (5,5) -- (5.5,5);
\draw[cyan,thick,<-,>=stealth] (7,7) -- (7,6.5);
\draw[cyan,thick,<->,>=stealth] (7,7) -- (8,7);
\draw[cyan,thick,<->,>=stealth] (8,8) -- (8,7);
\draw[cyan,thick,<->,>=stealth] (8,8) -- (9,8);

\draw[cyan,thick,<-,>=stealth] (4,2) -- (4.5,2);
\draw[cyan,thick,<-,>=stealth] (5,3) -- (5,2.5);
\draw[cyan,thick,<-,>=stealth] (5,3) -- (5.5,3);
\draw[cyan,thick,<-,>=stealth] (6,4) -- (6,3.5);
\draw[cyan,thick,<-,>=stealth] (6,4) -- (6.5,4);
\draw[cyan,thick,<-,>=stealth] (8,6) -- (8,5.5);
\draw[cyan,thick,<-,>=stealth] (8,6) -- (8.5,6);
\draw[cyan,thick,<-,>=stealth] (9,7) -- (9,6.5);
\draw[cyan,thick,<-,>=stealth] (9,7) -- (9.5,7);

\draw[cyan,thick,->,>=stealth] (5.5,1) -- (6,1);
\draw[cyan,thick,->,>=stealth] (6,1.5) -- (6,1);
\draw[cyan,thick,->,>=stealth] (6.5,2) -- (7,2);
\draw[cyan,thick,->,>=stealth] (7,2.5) -- (7,2);
\draw[cyan,thick,->,>=stealth] (9.5,5) -- (10,5);
\draw[cyan,thick,->,>=stealth] (10,5.5) -- (10,5);
\draw[cyan,thick,->,>=stealth] (10.5,6) -- (11,6);

\draw[cyan,thick,<->,>=stealth] (6,0) -- (7,0);
\draw[cyan,thick,<->,>=stealth] (7,1) -- (7,0);
\draw[cyan,thick,<->,>=stealth] (7,1) -- (8,1);
\draw[cyan,thick,<->,>=stealth] (8,2) -- (8,1);
\draw[cyan,thick,<-,>=stealth] (8,2) -- (8.5,2);
\draw[cyan,thick,<-,>=stealth] (10,4) -- (10,3.5);
\draw[cyan,thick,<->,>=stealth] (10,4) -- (11,4);
\draw[cyan,thick,<->,>=stealth] (11,5) -- (11,4);
\draw[cyan,thick,<->,>=stealth] (11,5) -- (12,5);

\node[scirc] at (4.7,0) {};
\node[scirc] at (5,0) {};
\node[scirc] at (5.3,0) {};

\node[scirc] at (9.7,8.5) {};
\node[scirc] at (10,8.5) {};
\node[scirc] at (10.3,8.5) {};

\node[scirc] at (5.7,5.7) {};
\node[scirc] at (6,6) {};
\node[scirc] at (6.3,6.3) {};

\node[scirc] at (9,2.4) {};
\node[scirc] at (9.3,2.7) {};
\node[scirc] at (9.6,3) {};

\node[invisible] at (1,11) {};
\end{tikzpicture}
\caption{$k$ is odd.}
\end{subfigure}
\caption{A $(k+1)$-bend CPG representation of $G_k$ (the endpoints are marked by an arrow).}
\label{fig:representation}
\end{figure}
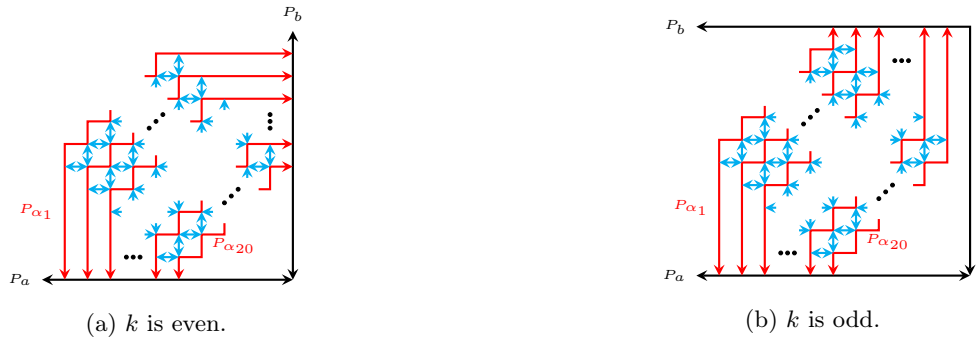

\begin{theorem}
\label{thm:unbound}
For any $k \geq 0$, there exists a planar graph which is $B_{k+1}$-CPG but not $B_k$-CPG.
\end{theorem}

\begin{proof}
We show that in any CPG representation $\mathcal{R} = (\mathcal{G},\mathcal{P})$ of $G_k$, there exists a path with at least $k+1$ bends. Since $G_k$ is $B_{k+1}$-CPG (see \Cref{fig:representation}), this would conclude the proof. We first observe the following.

\begin{observation}
\label{pureendpoints}
If a path $P_{\alpha_i}$ is pure, then one endpoint of $P_{\alpha_i}$ belongs to $P_a$ and the other endpoint belongs to $P_b$.
\end{observation}

We next prove the following two claims.

\begin{Claim}
\label{bendpoint}
Let $P_{\alpha_i}$ and $P_{\alpha_{i+1}}$, with $1 \leq i \leq 19$,  be two pure paths and let $u_j^i$ and $u_{j+1}^i$, with $1 \leq j \leq k + 1$, be two $(i,i+1)$-sewing vertices. If a grid-point $x$ belongs to $P_{u_j^i} \cap P_{u_{j+1}^i}$, then $x$ corresponds to an endpoint of both $P_{u_j^i}$ and $P_{u_{j+1}^i}$, and a bend-point of either $P_{\alpha_i}$ or $P_{\alpha_{i+1}}$.
\end{Claim}

\begin{cproof}
It follows from Observation \ref{pureendpoints} and the fact that $u_j^i$ and $u_{j+1}^i$ are nonadjacent to both $a$ and $b$, that no endpoint of $P_{\alpha_i}$ or $P_{\alpha_{i+1}}$ belongs to $P_{u_j^i}$ or $P_{u_{j+1}^i}$. Consequently, one endpoint of $P_{u_j^i}$ (resp. $P_{u_{j+1}^i}$) belongs to $P_{\alpha_i}$ and the other endpoint belongs to $P_{\alpha_{i+1}}$. By definition, $x$ corresponds to an endpoint of at least one of $P_{u_j^i}$ and $P_{u_{j+1}^i}$, which implies that $x$ belongs to either $\mathring{P_{\alpha_i}}$ or $\mathring{P_{\alpha_{i+1}}}$. But then $x$ must be an endpoint of both $P_{u_j^i}$ and $P_{u_{j+1}^i}$; in particular, $x$ is a grid-point of either type II.a or type II.b. Without loss of generality, we may assume that $x \in \mathring{P_{\alpha_i}}$. We denote by $y_{u_j^i}$ (resp. $y_{u_{j+1}^i}$) the endpoint of $P_{u_j^i}$ (resp. $P_{u_{j+1}^i}$) belonging to $P_{\alpha_{i+1}}$. Suppose, to the contrary, that $x$ is of type II.a. The union of $P_{u_j^i}$, $P_{u_{j+1}^i}$ and the portion of $P_{\alpha_{i+1}}$ between $y_{u_j^i}$ and $y_{u_{j+1}^i}$ defines a closed curve $\mathcal{C}$ dividing the plane into two regions. Since $P_a$ and $P_b$ touch neither $P_{u_j^i}$ nor $P_{u_{j+1}^i}$, $P_a$ and $P_b$ lie entirely in one of these regions. Moreover, since $a$ and $b$ are adjacent and $P_{\alpha_{i+1}}$ is pure, $P_a$ and $P_b$ in fact belong to the same region. On the other hand, since one endpoint of $P_{u_j^i}$ (resp. $P_{u_{j+1}^i}$) belongs to $P_{\alpha_i}$ while the other endpoint belongs to $P_{\alpha_{i+1}}$, and both endpoints of $P_{\alpha_i}$ are in $P_a \cup P_b$, it follows that $x \in \mathcal{C}$ is the only contact point between $P_{u_j^i}$ (resp. $P_{u_{j+1}^i}$) and $P_{\alpha_i}$. But $\alpha_i$ and $\alpha_{i+1}$ are nonadjacent, which implies that $P_{\alpha_i}$ crosses $\mathcal{C}$ only once and has therefore one endpoint in each region. However, both endpoints of $P_{\alpha_i}$ belong to $P_a \cup P_b$, contradicting the fact that $P_a$ and $P_b$ lie in the same region. Hence, $x$ is of type II.b, thus concluding the proof.
\end{cproof}

\begin{Claim}
\label{bends}
If two paths $P_{\alpha_i}$ and $P_{\alpha_{i+1}}$ are pure, for some $1 \leq i \leq 19$, then one of them contains at least $\lfloor \frac{k+1}{2} \rfloor$ bends and the other contains at least $\lceil \frac{k+1}{2} \rceil$ bends. Moreover, all these bend-points belong to \textit{(i,i+1)}-sewing paths.
\end{Claim}

\begin{cproof}
For each $1 \leq j \leq k+1$, consider a point $x_j \in P_{u_j^i} \cap P_{u_{j+1}^i}$. By Claim \ref{bendpoint}, $x_j$ is a bend-point of either $P_{\alpha_i}$ or $P_{\alpha_{i+1}}$. Since $x_j$ and $x_{j+1}$ are the endpoints of $P_{u_{j+1}^i}$, one of them belongs to $P_{\alpha_i}$ while the other belongs to $P_{\alpha_{i+1}}$. Therefore, $\{x_j : 1 \leq j \leq k+1 \text{ and } j \equiv 0 \pmod 2\}$ is a subset of one of the considered secondary path and $\{x_j : 1 \leq j \leq k+1 \text{ and } j \equiv 1 \pmod 2\}$ is a subset of the other secondary path.
\end{cproof}

\medskip

Finally, we claim that there exists an index $1 \leq j \leq 17$ such that $P_{\alpha_j}$, $P_{\alpha_{j+1}}$, $P_{\alpha_{j+2}}$ and $P_{\alpha_{j+3}}$ are all pure. Indeed, if it weren't the case, there would be at least $\lfloor 20/4 \rfloor = 5$ secondary paths which are not pure. But at most $4$ secondary paths can contain endpoints of $P_a$ or $P_b$, a contradiction. It now follows from Claim \ref{bends} that $P_{\alpha_{j+1}}$ has at least $\lfloor \frac{k+1}{2} \rfloor$ bends (which belong to $(j,j+1)$-sewing paths) and that $P_{\alpha_{j+2}}$ has at least $\lfloor \frac{k+1}{2} \rfloor$ bends (which belong to $(j+2,j+3)$-sewing paths). Furthermore, one of $P_{\alpha_{j+1}}$ and $P_{\alpha_{j+2}}$ has at least $\lceil \frac{k+1}{2} \rceil$ bends which are endpoints of $(j+1,j+2)$-sewing paths. Therefore, there is a path with at least $\lfloor \frac{k+1}{2} \rfloor + \lceil \frac{k+1}{2} \rceil = k+1$ bends.
\end{proof}

\begin{corollary}
For any $k\geq 0$, $B_k$-CPG is strictly contained in $B_{k+1}$-CPG, even within the class of planar graphs.
\end{corollary}

Although not every planar graph is CPG \cite{cpg}, we may still define the bend number, with respect to CPG representations, of the class of planar graphs which are CPG. From Theorem \ref{thm:unbound}, we deduce the following.

\begin{corollary}
The class of planar CPG graphs has unbounded bend number with respect to CPG representations.
\end{corollary}

%===========================================================

\section{Recognition}
\label{sec:recognition}

In this section, we show that {\sc Recognition} is $\mathsf{NP}$-complete for $B_k$-CPG graphs with $k \geq 1$, as well as for planar $B_k$-CPG graphs with $k \geq 3$. As a corollary, we also show that {\sc Recognition} is $\mathsf{NP}$-complete for CPG graphs and remains $\mathsf{NP}$-complete for planar CPG graphs.

We begin by showing membership in $\mathsf{NP}$. This is trivial in the case of $B_{k}$-CPG graphs. Indeed, as a polynomial certificate it is enough to consider, for each path $P$, a sequence $s(P)$ of grid-points $(x_1,y_1), (x_2,y_2), \dots, (x_{\ell_P},y_{\ell_P})$, where $(x_1,y_1)$ and $(x_{\ell_P},y_{\ell_P})$ are the endpoints of $P$ and all the other points are the bend-points of $P$ in their order of appearance while traversing $P$ from $(x_1,y_1)$ to $(x_{\ell_P},y_{\ell_P})$ (notice that $\ell_P \leq k + 2$). In the case of CPG graphs the situation is more subtle and we need the following adaptation of \citep[Theorem~1]{Hlineny}.
 
\begin{theorem}\label{innp} Every CPG graph on $n$ vertices admits a CPG representation on a grid of area $O(n^2)$. In particular, recognizing CPG graphs belongs to $\mathsf{NP}$.
\end{theorem}

\begin{proof} Let $G$ be a CPG graph on $n$ vertices with CPG representation $\mathcal{R} = (\mathcal{G}, \mathcal{P})$. We first build a plane graph $G'$ as follows. The vertices of $G'$ are the contact points in $\mathcal{R}$ and the free endpoints of paths in $\mathcal{P}$. The edges of $G'$ are the portions of paths in $\mathcal{P}$ between the vertices. Since there are at most $2n$ contact points in $\mathcal{R}$, $|V(G')| = O(n)$. Since $G'$ has maximum degree at most $4$, it has $O(n)$ edges. Therefore, denoting by $H$ the $1$-subdivision of $G'$, we have that $H$ is a simple planar graph with maximum degree at most $4$ and $|V(H)| = O(n)$. 

By \citep{tamassia}, every planar graph on $n$ vertices with maximum degree at most $4$ admits an embedding on the grid where vertices are mapped to grid-points, edges are mapped to pairwise interiorly disjoint grid-paths connecting the two grid-points corresponding to the endvertices and the area of the smallest rectangle enclosing the embedding is $O(n^2)$. Therefore, $H$ admits such an embedding on a grid of area $O(n^2)$ and so $G$ admits a CPG representation on a grid of area $O(n^2)$. In this representation, every path $P$ has a polynomial (in $n$) number of bends and we provide the sequences $s(P)$ as a polynomial certificate. \end{proof}

We now show $\mathsf{NP}$-hardness of {\sc Recognition} for $B_k$-CPG graphs. Our hardness proof is based on that of \citep[Theorem~4]{Hlineny}. We reduce from \textsc{Planar Exactly 3-Bounded 3-Sat}, shown to be $\mathsf{NP}$-complete in~\cite{Dahlhaus}. This problem is defined as a special case of the classical \textsc{Planar 3-Sat} used in \citep[Theorem~4]{Hlineny}: the instance is a formula $\Phi$ with variable set $V$ and clause set $C$ such that each clause has size at most 3, each variable appears in exactly 3 clauses and the bipartite graph $H$ with vertex set $V \cup C$ and an edge between $x \in V$ and $c \in C$ if either $x \in c$ or $\neg x \in c$ is planar. Before turning to the proof, we first provide some preliminary results. Note that the following easy observation holds up to symmetry and rotation.

\begin{observation}
Let $G$ be a $B_1$-CPG graph and $\mathcal{R}$ be a 1-bend CPG representation of $G$. Suppose that there exist two distinct grid-points $p=(x_1,y)$ and $q=(x_2,y)$ with $x_1 < x_2$ and three paths $P$, $P'$ and $P''$ such that (see \Cref{fig:grid-points}):
\begin{itemize}
    \item $p$ is an endpoint of $P$ and $P$ uses the grid-edge above $p$;
    \item $p$ is an endpoint of $P'$ and $P'$ uses the grid-edge below $p$;
    \item $q$ is an endpoint of $P''$ and $P''$ does not use the grid-egde to the left of $q$.
\end{itemize}
Then one of the three cases depicted in \Cref{fig:grid-points} occurs. If (a) occurs, then either $P''$ and $P$ touch or $P''$ and $P'$ touch but not both; and if (b) (resp. (c)) occurs, then $P''$ cannot touch $P'$ (resp. $P$).
\label{obs:grid-points}
\end{observation}

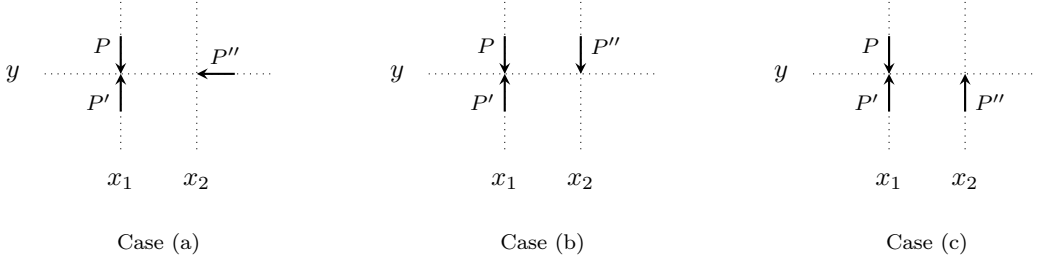
\begin{figure}[h]
\centering
\scalebox{1}
{
    \begin{subfigure}{.3\textwidth}
    \centering
    \begin{tikzpicture}[scale=1]
    
        \draw[thick,->,>=stealth] (-.5,.5) -- (-.5,0) node[near start,left] {\footnotesize $P$};
        \draw[thick,->,>=stealth] (-.5,-.5) -- (-.5,0) node[near start,left] {\footnotesize $P'$};
        \draw[thick,->,>=stealth] (1,0) -- (.5,0) node[near start,above] {\footnotesize $P''$};
        
        \draw[dotted]
            (-1.5,0) -- (1.5,0)
            (-.5,-1) -- (-.5,1)
            (.5,-1) -- (.5,1);
        
        \node[below] at (-.5,-1.2) {$x_1$};
        \node[below] at (.5,-1.2) {$x_2$};
        \node[left] at (-1.7,0) {$y$};

        \node[draw=none] at (0,-2.25) {\footnotesize Case (a)};

    \end{tikzpicture}
    \end{subfigure}
    
    \begin{subfigure}{.3\textwidth}
    \centering
    \begin{tikzpicture}[scale=1]
    
        \draw[thick,->,>=stealth] (-.5,.5) -- (-.5,0) node[near start,left] {\footnotesize $P$};
        \draw[thick,->,>=stealth] (-.5,-.5) -- (-.5,0) node[near start,left] {\footnotesize $P'$};
        \draw[thick,->,>=stealth] (.5,.5) -- (.5,0) node[near start,right] {\footnotesize $P''$};
        
        \draw[dotted]
            (-1.5,0) -- (1.5,0)
            (-.5,-1) -- (-.5,1)
            (.5,-1) -- (.5,1);
        
        \node[below] at (-.5,-1.2) {$x_1$};
        \node[below] at (.5,-1.2) {$x_2$};
        \node[left] at (-1.7,0) {$y$};
        
        \node[draw=none] at (0,-2.25) {\footnotesize Case (b)};

    \end{tikzpicture}
    \end{subfigure}
    
    \begin{subfigure}{.3\textwidth}
    \centering
    \begin{tikzpicture}[scale=1]
    
        \draw[thick,->,>=stealth] (-.5,.5) -- (-.5,0) node[near start,left] {\footnotesize $P$};
        \draw[thick,->,>=stealth] (-.5,-.5) -- (-.5,0) node[near start,left] {\footnotesize $P'$};
        \draw[thick,->,>=stealth] (.5,-.5) -- (.5,0) node[near start,right] {\footnotesize $P''$};
        
        \draw[dotted]
            (-1.5,0) -- (1.5,0)
            (-.5,-1) -- (-.5,1)
            (.5,-1) -- (.5,1);
        
        \node[below] at (-.5,-1.2) {$x_1$};
        \node[below] at (.5,-1.2) {$x_2$};
        \node[left] at (-1.7,0) {$y$};

        \node[draw=none] at (0,-2.25) {\footnotesize Case (c)};

    \end{tikzpicture}
    \end{subfigure}
}
\caption{$P$, $P'$ and $P''$ cannot pariwise touch in a 1-bend CPG representation.}
\label{fig:grid-points}
\end{figure}

The following is a variation of \citep[Lemma~4.1]{Hlineny}.

\begin{lemma}
\label{lem:freeendpoint}
Let $\mathcal{R}$ be a CPG representation of a graph $G$ containing $f$ free endpoints of paths and no 4-contact point. Then $\mathcal{R}$ must contain at least $|E(G)| - 2 \cdot |V(G)| + f$ contact points corresponding to pairwise edge-disjoint triangles in $G$.
\end{lemma}

\begin{proof}
We construct a map $t\colon E(G) \to T(\mathcal{R})$ that assigns to each $e \in E(G)$ a contact point of $\mathcal{R}$ representing the edge $e$. If an edge $e$ is represented by more than one contact point, we arbitrarily choose one of them. Then, for any 2-contact point $x$, $|t^{-1}(x)| \leq 1$; and for any 3-contact point $x$, $|t^{-1}(x)| \leq 3$. 

Let $T(\mathcal{R}) = T_1 \cup T_2 \cup T_3$ be a partition of the contact points in $\mathcal{R}$ where $T_1 = \{x\in T(\mathcal{R}): |t^{-1}(x)| \leq 1\}$, $T_2=\{x\in T(\mathcal{R}): |t^{-1}(x)| = 2\}$ and $T_3=\{x\in T(\mathcal{R}): |t^{-1}(x)| = 3\}$. Note that any contact point in $T_3$ corresponds to a triangle in $G$ and that two such triangles cannot share an edge by definition of $t$. Further observe that in any contact point in $T_1$, at least one path ends; and since any contact point in $T_2$ or $T_3$ is a 3-contact point, at least two paths end in it. Thus, by counting the number of endpoints of paths in $\mathcal{R}$, we obtain that $|T_1| + 2 |T_2| + 2 |T_3| \leq 2 |V(G)| - f$. But from the map $t$ we know that $|E(G)| \leq |T_1| + 2 |T_2| + 3 |T_3|$ and so $|E(G)| - 2 |V(G)| + f \leq |T_3|$, which concludes the proof.
\end{proof}

The rest of this section is organized as follows. In \Cref{subsec:endeat}, we introduce three gadgets, referred to as \textit{end-eating graphs}, which we use in subsequent constructions. In \Cref{subsec:reduction}, we describe the actual reduction.

\subsection{The end-eating graphs}
\label{subsec:endeat}

The end-eating graph $\mathcal{E}_1$ depicted in Figure \ref{fig:endeating1} will be used in the proof of the $\mathsf{NP}$-hardness of {\sc Recognition} restricted to $B_1$-CPG graphs; the end-eating graph $\mathcal{E}_2$ depicted in Figure \ref{fig:endeating2} will be used in the proof of the $\mathsf{NP}$-hardness of {\sc Recognition} restricted to $B_2$-CPG graphs; and the end-eating graph $\mathcal{E}_3$ depicted in Figure \ref{fig:endeating3} will be used in the proof of the $\mathsf{NP}$-hardness of {\sc Recognition} restricted to planar $B_k$-CPG graphs, for $k \geq 3$. These gadgets are used to capture endpoints of paths: if a vertex $u$ of some other graph is adjacent to the special vertex $v$ of $\mathcal{E}_1$ (or to any vertex $v$ of $\mathcal{E}_2$ or $\mathcal{E}_3$), then the edge $uv$ can only be represented by using one endpoint of the path $P_u$ and an interior point of $P_v$, as shown in the following Lemmas.

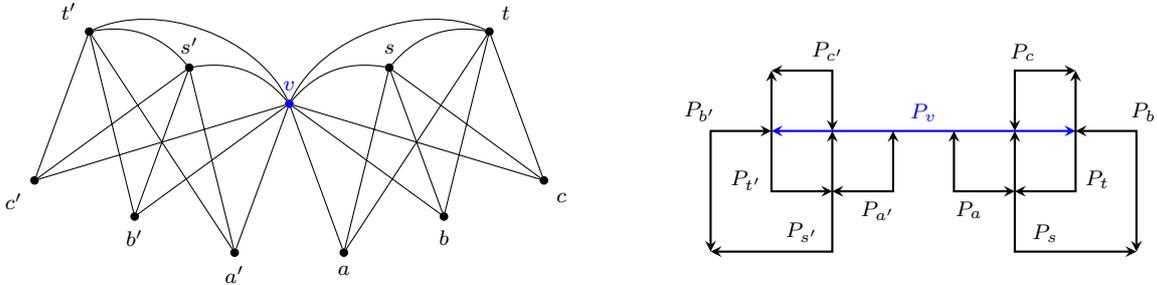
\begin{figure}[h]
\centering
\scalebox{1}
{
    \begin{subfigure}{.5\textwidth}
    \centering
    \begin{tikzpicture}[scale=0.7]
    
        \begin{scope}[shift={(0,0)},rotate=20,scale=1]
        \node[blue,circ] (v) at (0,0) [label=above:{\color{blue}\footnotesize $v$}] {};
        \node[circ] (s) at (2,0) [label=above:{\footnotesize $s$}] {};
        \node[circ] (t) at (4,0) [label=above right:{\footnotesize $t$}] {};
        \node[circ] (a) at (0,-3) [label=below:{\footnotesize $a$}] {};
        \node[circ] (b) at (2,-3) [label=below:{\footnotesize $b$}] {};
        \node[circ] (c) at (4,-3) [label=below right:{\footnotesize $c$}] {};

        \draw[-]
        (a) -- (v) -- (b) -- (s) -- (a) -- (t) -- (b)
        (v) -- (c) -- (t)
        (c) -- (s);

        \draw (v) .. controls (0.7,0.4) and (1.3,0.4) .. (s);
        \draw (s) .. controls (2.7,0.4) and (3.3,0.4) .. (t);
        \draw (t) .. controls (3,1) and (1,1) .. (v);
        \end{scope}
        
        \begin{scope}[shift={(0,0)},rotate=-20,scale=1]
        \node[circ] (s') at (-2,0) [label=above:{\footnotesize $s'$}] {};
        \node[circ] (t') at (-4,0) [label=above left:{\footnotesize $t'$}] {};
        \node[circ] (a') at (0,-3) [label=below:{\footnotesize $a'$}] {};
        \node[circ] (b') at (-2,-3) [label=below:{\footnotesize $b'$}] {};
        \node[circ] (c') at (-4,-3) [label=below left:{\footnotesize $c'$}] {};

        \draw[-]
        (a') -- (v) -- (b') -- (s') -- (a') -- (t') -- (b')
        (v) -- (c') -- (t')
        (c') -- (s');

        \draw (v) .. controls (-0.7,0.4) and (-1.3,0.4) .. (s');
        \draw (s') .. controls (-2.7,0.4) and (-3.3,0.4) .. (t');
        \draw (t') .. controls (-3,1) and (-1,1) .. (v);
        \end{scope}

    \end{tikzpicture}
    \end{subfigure}

    \begin{subfigure}{.5\textwidth}
    \centering
    \begin{tikzpicture}[scale=0.8]

        \draw[blue,thick,<->,>=stealth] (-3,2) -- (2,2) node[midway,above] {\footnotesize $P_v$};
        \draw[thick,<->,>=stealth] (1,2) -- (1,0) -- (3,0) node[near start,above] {\footnotesize $P_s$};
        \draw[thick,<->,>=stealth] (1,1) -- (2,1) -- (2,3) node[near start,below right] {\footnotesize $P_t$};
        \draw[thick,<->,>=stealth] (0,2) -- (0,1) -- (1,1) node[near start,below] {\footnotesize $P_a$};
        \draw[thick,<->,>=stealth] (3,0) -- (3,2) -- (2,2) node[near start,above right] {\footnotesize $P_b$};
        \draw[thick,<->,>=stealth] (1,2) -- (1,3) -- (2,3) node[midway,above left] {\footnotesize $P_c$};
        \draw[thick,<->,>=stealth] (-2,2) -- (-2,0) -- (-4,0) node[near start,above] {\footnotesize $P_{s'}$};
        \draw[thick,<->,>=stealth] (-2,1) -- (-3,1) -- (-3,3) node[near start,below left] {\footnotesize $P_{t'}$};
        \draw[thick,<->,>=stealth] (-1,2) -- (-1,1) -- (-2,1) node[near start,below] {\footnotesize $P_{a'}$};
        \draw[thick,<->,>=stealth] (-4,0) -- (-4,2) -- (-3,2) node[near start,above left] {\footnotesize $P_{b'}$};
        \draw[thick,<->,>=stealth] (-2,2) -- (-2,3) -- (-3,3) node[midway,above right] {\footnotesize $P_{c'}$};

    \end{tikzpicture}
    \end{subfigure}
}
\caption{The end-eating graph $\mathcal{E}_1$ and a 1-bend CPG representation of it.}
\label{fig:endeating1}
\end{figure}

\begin{figure}[ht]
\centering
\scalebox{1}
{
    \begin{subfigure}[b]{.45\textwidth}
    \centering
    \begin{tikzpicture}[scale=.55]
        \node [circ] at (1,0) (v1)  [label=below left:{\tiny $1$}] {};
        \node [circ] at (0,2) (v2)  [label=below left:{\tiny $2$}] {};
        \node [circ] at (0,4) (v3)  [label=above left:{\tiny $3$}] {};
        \node [circ] at (1,6) (v4)  [label=above left:{\tiny $4$}] {};
        \node [circ] at (6,0) (v5)  [label=below right:{\tiny $5$}] {};
        \node [circ] at (7,2) (v6)  [label=below right:{\tiny $6$}] {};
        \node [circ] at (7,4) (v7)  [label=above right:{\tiny $7$}] {};
        \node [circ] at (6,6) (v8)  [label=above right:{\tiny $8$}] {};
        \node [circ] at (5,1) (v9)  [label=below:{\tiny $9$}] {};
        \node [circ] at (2,1) (v10)  [label=below:{\tiny $10$}] {};
        \node [circ] at (2,5) (v12)  [label=above:{\tiny $12$}] {};
        \node [circ] at (5,5) (v13)  [label=above:{\tiny $13$}] {};
        \node [circ] at (2,3) (v11)  [label=right:{\tiny $11$}] {};
        \node [circ] at (5,3) (v14)  [label=left:{\tiny $14$}] {};
        \node [circ] at (4,2) (v15)  [label=below:{\tiny $15$}] {};
        \node [circ] at (3,2) (v16)  [label=below:{\tiny $16$}] {};
        \node [circ] at (3,4) (v17)  [label=above:{\tiny $17$}] {};
        \node [circ] at (4,4) (v18)  [label=above:{\tiny $18$}] {};
	
        \draw[-]
            (v1) -- (v2) -- (v3) -- (v4) -- (v8) -- (v7) -- (v6) -- (v5) -- (v1) -- (v4) -- (v2) -- (v11) -- (v3) -- (v12) -- (v4) -- (v13) -- (v7) -- (v14) -- (v6) -- (v9) -- (v5) -- (v8) -- (v6) -- (v15) -- (v14) -- (v13) -- (v8) -- (v12) -- (v13) -- (v18) -- (v7) -- (v5) -- (v10) -- (v9) -- (v1) -- (v10) -- (v2) -- (v16) -- (v10) -- (v11) -- (v16) -- (v15) -- (v9) -- (v14) -- (v18) -- (v15) -- (v17) -- (v18) -- (v16) -- (v17) -- (v11) -- (v12) -- (v17) -- (v3) -- (v1);
    \end{tikzpicture} 
    \end{subfigure}

    \begin{subfigure}[b]{.45\textwidth}
    \centering
    \begin{tikzpicture}[scale=.6]
        \draw[thick,<->,>=stealth] (1,5)--(0,5)--node[left=.3pt]{\tiny{$P_1$}}(0,2)--(2,2);
        \draw[thick,<->,>=stealth] (1,5)--(1,6)--node[above=.3pt]{\tiny{$P_4$}}(4,6)--(4,4);
        \draw[thick,<->,>=stealth] (1,5)--node[left=.3pt]{\tiny{$P_2$}}(1,3)--(2,3);
        \draw[thick,<->,>=stealth] (1,5)--node[above=.3pt]{\tiny{$P_3$}}(3,5)--(3,4);
        \draw[thick,<->,>=stealth] (5,1)--(5,0)--node[below=.3pt]{\tiny{$P_5$}}(2,0)--(2,2);
        \draw[thick,<->,>=stealth] (5,1)--(6,1)--node[right=.3pt]{\tiny{$P_8$}}(6,4)--(4,4);
        \draw[thick,<->,>=stealth] (5,1)--node[below=.3pt]{\tiny{$P_6$}}(3,1)--(3,2);
        \draw[thick,<->,>=stealth] (5,1)--node[right=.3pt]{\tiny{$P_7$}}(5,3)--(4,3);
        \draw[thick,<->,>=stealth] (3,2)--node[right=7pt]{\tiny{$P_{14}$}}(4,2)--(4,3);
        \draw[thick,<->,>=stealth] (2,3)--node[above=7pt]{\tiny{$P_{11}$}}(2,4)--(3,4);
        \draw[thick,<->,>=stealth] (3,3)--node[right=-2pt]{\tiny{$P_{15}$}}(3,2);
        \draw[thick,<->,>=stealth] (3,3)--node[above=-2pt]{\tiny{$P_{18}$}}(4,3);
        \draw[thick,<->,>=stealth] (3,3)--node[left=-3pt]{\tiny{$P_{17}$}}(3,4);
        \draw[thick,<->,>=stealth] (3,3)--node[below=0pt]{\tiny{$P_{16}$}}(2,3);
        \draw[thick,<->,>=stealth] (2,2)--node[below=-2pt]{\tiny{$P_{9}$}}(3,2);
        \draw[thick,<->,>=stealth] (2,2)--node[left=-2pt]{\tiny{$P_{10}$}}(2,3);
        \draw[thick,<->,>=stealth] (4,4)--node[right=-2pt]{\tiny{$P_{13}$}}(4,3);
        \draw[thick,<->,>=stealth] (4,4)--node[above=-2pt]{\tiny{$P_{12}$}}(3,4);
    \end{tikzpicture}
    \end{subfigure}
}
\caption{The end-eating graph $\mathcal{E}_2$ and a 2-bend CPG representation of it.}
\label{fig:endeating2}
\end{figure}
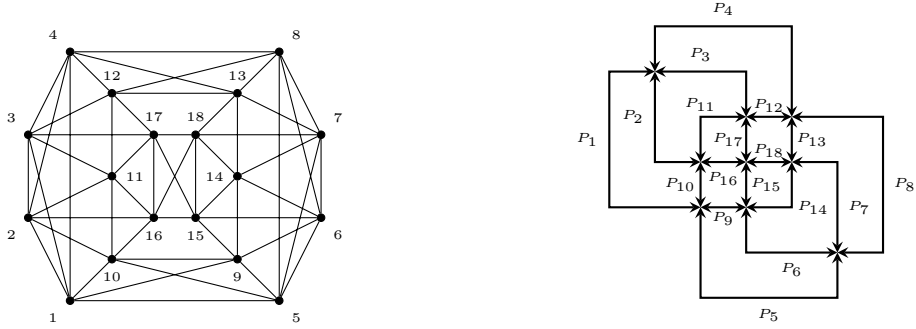

\begin{figure}[h]
\centering
\scalebox{1}
{
    \begin{subfigure}{.45\textwidth}
    \centering
    \begin{tikzpicture}[scale=0.24] 
    
        \node[circ] (1') at (0,16) [label={[label distance= .01cm]10:\tiny $1'$}] {};
        \node[circ] (2') at (0,22) [label=above:{\tiny $2'$}] {};
        \node[circ] (3') at (6,24) [label=right:{\tiny $3'$}] {};
        \node[circ] (4') at (10,19) [label=above:{\tiny $4'$}] {};
        \node[circ] (5') at (6,14) [label=right:{\tiny $5'$}] {};
        \node[circ] (a') at (4,16.5) [] {};
        \node[circ] (b') at (7,17.5) [] {};
        \node[circ] (c') at (2,19) [] {};
        \node[circ] (d') at (4.5,19) [] {};
        \node[circ] (e') at (7,20.5) [] {};
        \node[circ] (f') at (4,21.5) [] {};
        \node[circ] (4) at (3,2) [label={[label distance=.01cm]280:\tiny $4$}] {};
        \node[circ] (3) at (3,8) [label=above right:{\tiny $3$}] {};
        \node[circ] (2) at (9,10) [label={[label distance=.04cm]268:\tiny $2$}] {};
        \node[circ] (1) at (13,5) [label=below:{\tiny $1$}] {};
        \node[circ] (5) at (9,0) [label=below right:{\tiny $5$}] {};
        \node[circ] (b) at (7,2.5) [] {};
        \node[circ] (a) at (10,3.5) [] {};
        \node[circ] (e) at (5,5) [] {};
        \node[circ] (d) at (7.5,5) [] {};
        \node[circ] (c) at (10,6.5) [] {};
        \node[circ] (f) at (7,7.5) [] {};

        \draw[-]
        (1) -- (2) -- (3) -- (4) -- (5) -- (1) -- (c) -- (2) -- (f) -- (3) -- (e) -- (4) -- (b) -- (5) -- (a) -- (1)
        (a) -- (c) -- (f) -- (e) -- (b) -- (a) -- (d) -- (b)
        (c) -- (d) -- (e)
        (d) -- (f)
        (1') -- (2') -- (3') -- (4') -- (5') -- (1') -- (c') -- (2') -- (f') -- (3') -- (e') -- (4') -- (b') -- (5') -- (a') -- (1')
        (a') -- (c') -- (f') -- (e') -- (b') -- (a') -- (d') -- (b')
        (c') -- (d') -- (e')
        (d') -- (f')
        (1') -- (2)
        (1') -- (3)
        (1') -- (4)
        (1') .. controls (0,1) and (2,-1) .. (5)
        (2') .. controls (-4,0) and (2,-3) .. (5)
        (2') .. controls (8,30) and (16,22) .. (1)
        (3') .. controls (12,21) .. (1)
        (4') -- (1)
        (5') .. controls (10,11) .. (1)
        (5') -- (2);
        
    \end{tikzpicture}
    \end{subfigure}
    
    \begin{subfigure}{.45\textwidth}
    \centering
    \begin{tikzpicture}[scale=0.5] 
        \draw[red,thick,<->,>=stealth] (7,8) -- (7,9) -- (10,9) node[midway,above] {\tiny $P_{1'}$} -- (10,2) -- (5,2);
        \draw[red,thick,<->,>=stealth] (7,9) -- (7,10) -- (4,10) -- (4,11) -- (10,11) node[midway,above] {\tiny $P_{2'}$};
        \draw[red,thick,<->,>=stealth] (4,10) -- (2,10) -- (2,14) -- (11,14) node[midway,below] {\tiny $P_{3'}$} -- (11,13);
        \draw[red,thick,<->,>=stealth] (10,1) -- (10,0) -- (0,0) node[midway,below] {\tiny $P_{4'}$} -- (0,10) -- (2,10);
        \draw[red,thick,<->,>=stealth] (10,2) -- (10,1) -- (1,1) node[midway,below] {\tiny $P_{5'}$} -- (1,8) -- (2,8);

        \draw[cyan,thick,<->,>=stealth] (10,2) -- (11,2) -- (11,11) -- (13,11);
        \draw[cyan,thick,<->,>=stealth] (10,1) -- (13,1) -- (13,13) -- (12,13);
        \draw[cyan,thick,<->,>=stealth] (10,9) -- (10,11) -- (11,11);
        \draw[cyan,thick,<->,>=stealth] (11,11) -- (11,12) -- (12,12) -- (12,13);
        \draw[cyan,thick,<->,>=stealth] (11,13) -- (12,13) -- (12,15) -- (1,15) -- (1,10);
        \draw[cyan,thick,<->,>=stealth] (10,11) -- (10,13) -- (11,13) -- (11,12);

        \draw[thick,<->,>=stealth] (2,10) -- (2,8) -- (4,8) node[midway,above] {\tiny $P_1$} -- (4,10);
        \draw[thick,<->,>=stealth] (2,8) -- (2,2) -- (5,2) node[midway,above] {\tiny $P_2$} -- (5,4);
        \draw[thick,<->,>=stealth] (7,7) -- (7,8) -- (9,8) -- (9,3) -- (5,3) node[midway,below] {\tiny $P_3$};
        \draw[thick,<->,>=stealth] (5,8) -- (7,8) node[midway,above] {\tiny $P_4$};
        \draw[thick,<->,>=stealth] (4,8) -- (5,8) -- (5,9) -- (7,9) node[midway,above] {\tiny $P_5$};

        \draw[cyan,thick,<->,>=stealth] (4,8) -- (4,5) -- (5,5) -- (5,7);
        \draw[cyan,thick,<->,>=stealth] (5,8) -- (5,7) -- (6,7);
        \draw[cyan,thick,<->,>=stealth] (3,8) -- (3,4) -- (5,4) -- (5,5);
        \draw[cyan,thick,<->,>=stealth] (6,7) -- (6,6) -- (7,6) -- (7,5) -- (5,5);
        \draw[cyan,thick,<->,>=stealth] (6,8) -- (6,7) -- (7,7);
        \draw[cyan,thick,<->,>=stealth] (7,6) -- (7,7) -- (8,7) -- (8,4) -- (5,4);
    \end{tikzpicture}
    \end{subfigure}
}
\caption{The end-eating graph $\mathcal{E}_3$ and a 3-bend CPG representation of it.}
\label{fig:endeating3}
\end{figure}
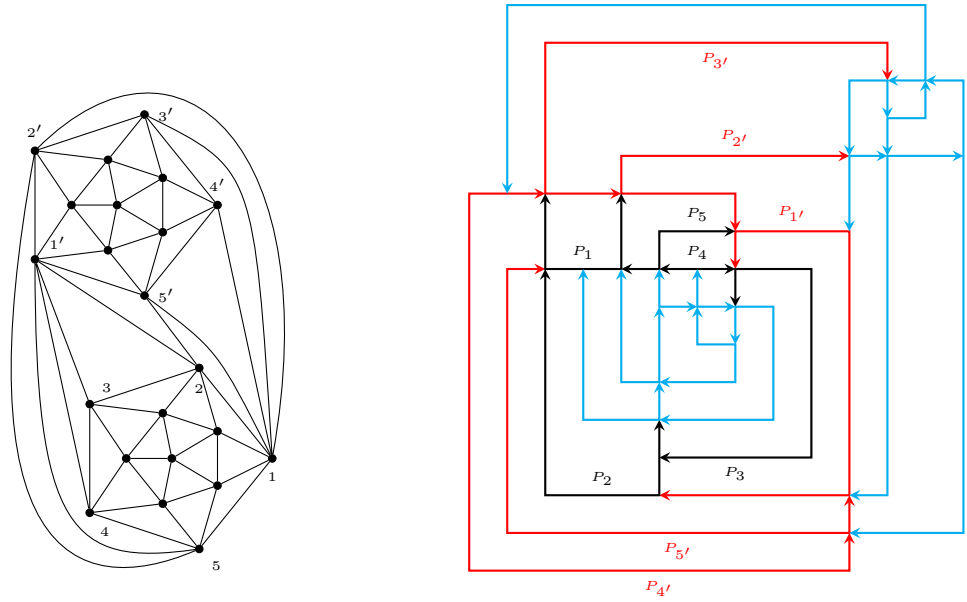

\begin{lemma}
\label{lem:Ei}
There exists no CPG representation of $\mathcal{E}_i$, for $i=2,3$, in which a path has a free endpoint.
\end{lemma}

\begin{proof}
The result for $i=2$ follows from the fact that $\mathcal{E}_2$ is 6-regular; indeed, if there exists a path $P_u$ in a CPG representation of $\mathcal{E}_2$ having a free endpoint, then $w_u \leq 3/2$ (see \Cref{sec:prel}) and thus $$|E(\mathcal{E}_2)| \leq \sum_{v \in V} w_v \leq 3 |V(\mathcal{E}_2)| - \frac{3}{2},$$ a contradiction.

The result for $i=3$ follows from the fact that $\mathcal{E}_3$ is $K_4$-free and has 22 vertices, 60 edges but only 16 pairwise edge-disjoint triangles: hence, by Lemma \ref{lem:freeendpoint}, no CPG representation of $\mathcal{E}_3$ contains free endpoints of paths. 
\end{proof}

\begin{lemma}
\label{lem:E1}
There exists no 1-bend CPG representation of $\mathcal{E}_1$ in which the path $P_v$ has a free endpoint.
\end{lemma}

\begin{proof}
Notice that $\mathcal{E}_1$ consists of two copies of the graph $G(v)$ depicted in Figure \ref{fig:Gv}, where the copies of $v$ have been identified. Hence, it suffices to prove that in any 1-bend CPG representation of $G(v)$, the path $P_v$ has at most one free endpoint.

\begin{figure}[h]
\centering
\scalebox{1}
{
    \begin{subfigure}{.5\textwidth}
    \centering
    \begin{tikzpicture}[scale=1]
    
        \node[circ] (a) at (0,0) [label=below left:{\footnotesize $a$}] {};
        \node[circ] (b) at (2,0) [label=below:{\footnotesize $b$}] {};
        \node[circ] (c) at (4,0) [label=below right:{\footnotesize $c$}] {};
        \node[blue,circ] (v) at (0,3) [label=above left:{\color{blue}\footnotesize $v$}] {};
        \node[circ] (s) at (2,3) [label=above:{\footnotesize $s$}] {};
        \node[circ] (t) at (4,3) [label=above right:{\footnotesize $t$}] {};

        \draw[-]
        (a) -- (v) -- (b) -- (s) -- (a) -- (t) -- (b)
        (v) -- (c) -- (t)
        (c) -- (s);

        \draw (v) .. controls (0.7,3.4) and (1.3,3.4) .. (s);
        \draw (s) .. controls (2.7,3.4) and (3.3,3.4) .. (t);
        \draw (t) .. controls (3,4) and (1,4) .. (v);

    \end{tikzpicture}
    \end{subfigure}

    \begin{subfigure}{.5\textwidth}
    \centering
    \begin{tikzpicture}[scale=1]

        \draw[blue,thick,<->,>=stealth] (0,2) -- (3,2) node[near start,above left] {\footnotesize $P_v$};
        \draw[thick,<->,>=stealth] (2,2) -- (2,0) -- (4,0) node[near start,above] {\footnotesize $P_s$};
        \draw[thick,<->,>=stealth] (2,1) -- (3,1) -- (3,3) node[near start,below right] {\footnotesize $P_t$};
        \draw[thick,<->,>=stealth] (1,2) -- (1,1) -- (2,1) node[near start,below left] {\footnotesize $P_a$};
        \draw[thick,<->,>=stealth] (4,0) -- (4,2) -- (3,2) node[near start,above right] {\footnotesize $P_b$};
        \draw[thick,<->,>=stealth] (2,2) -- (2,3) -- (3,3) node[midway,above left] {\footnotesize $P_c$};

    \end{tikzpicture}
    \end{subfigure}
}
\caption{The graph $G(v)$ and a 1-bend CPG representation of it where $P_v$ only has one free endpoint}
\label{fig:Gv}
\end{figure}
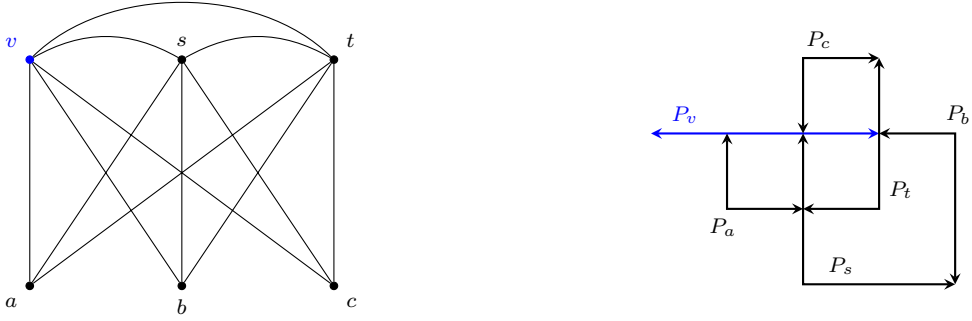

Let $\mathcal{R}$ be a 1-bend CPG representation of $G(v)$. Suppose, to the contrary, that both endpoints of $P_v$ belong to no other path in $\mathcal{R}$. Then, every other path touches the interior of $P_v$ and since any clique of size 4 contains $v$ and both endpoints of $P_v$ are free, $\mathcal{R}$ contains no 4-contact point. Now $G(v)$ has 6 vertices, 12 edges and at least 2 free endpoints and so, by \Cref{lem:freeendpoint}, $\mathcal{R}$ must contain at least two 3-contact points $p$ and $q$ representing edge-disjoint triangles in $G(v)$, which we denote by $T_p$ and $T_q$, respectively. Note that since every triangle of $G(v)$ shares an edge with the triangle $(v,s,t)$, neither $T_p$ nor $T_q$ is $(v,s,t)$. But then, both $T_p$ and $T_q$ share an edge with $(v,s,t)$ and so $T_p$ and $T_q$ have exactly one common vertex, which is either $v$, $s$ or $t$. If it is $v$, then we may assume without loss of generality that $T_p$ is $(v,s,a)$ and $T_q$ is $(v,t,b)$, as $a$, $b$ and $c$ have a symmetric role. Otherwise, it is $s$ or $t$, and we may then assume without loss of generality that $T_p$ is $(v,s,a)$ and $T_q$ is $(s,t,b)$. Thus, the following holds.

\begin{fact}
\label{a:triangles}
There exist two 3-contact points $p$ and $q$ in $\mathcal{R}$ such that $p \in P_v \cap P_s \cap P_a$ and one of the following holds:
\begin{itemize}
    \item $q \in P_v \cap P_t \cap P_b$;
    \item $q \in P_s \cap P_t \cap P_b$.
\end{itemize}
\end{fact}

As a consequence, we now show the following.

\begin{Claim}
There exists a path $P \in \{P_b,P_c\}$ such that $P$ touches the interior of $P_s$ and contains an endpoint of $P_t$.
\label{cl:path_P}
\end{Claim}

\begin{cproof}
Since $P_v$ has two free endpoints, $p$ must be an interior point of $P_v$ and an endpoint of $P_s$ and $P_a$. But then $P_s$ has only one endpoint left and so $P_s$ cannot touch both the interior of $P_b$ and the interior of $P_c$. It follows that one path $P \in \{P_b,P_c\}$ must touch the interior of $P_s$ and since $P$ touches the interior of $P_v$, $P$ cannot touch the interior of $P_t$.
\end{cproof}

\bigskip

Back to the proof of \Cref{lem:E1}, we now distinguish two cases, depending on whether $p$ is a bend-point of $P_v$ or not.\\

\textit{Case 1. $p$ is a bend-point of $P_v$.} By possibly rotating and reflecting, we may assume without loss of generality that $P_v$ lies on the left and below $p$, with $P_s$ coming from above and $P_a$ coming from the right. Figure \ref{fig:bend_cases} depicts the four possible ways for $P_t$ to touch the interior of $P_v$. In subcase 1, $P_t$ and $P_a$ necessarily touch as depicted in Figure \ref{fig:bend_case_1}. By Claim \ref{cl:path_P}, there exists a path $P$ touching an endpoint of $P_t$ and the interior of $P_s$. In order to touch the interior of $P_s$, $P$ must belong to the green area depicted in Figure \ref{fig:bend_case_1}; but then, $P$ cannot touch an endpoint of $P_t$, a contradiction. In subcase 2 (resp. subcase 3), an observation similar to Observation \ref{obs:grid-points} shows that $P_t$ and $P_s$ (resp. $P_t$ and $P_a$) cannot touch, a contradiction. Finally, in subcase~4, $P_t$ and $P_s$ necessarily touch as depicted in Figure \ref{fig:bend_case_4}. Thus, $P_s \cap P_t$ is reduced to one point $i_{s,t}$ and $P_v \cap P_t$ is reduced to one point $i_{v,t}$ (see Figure \ref{fig:bend_case_4}). Then, by Fact \ref{a:triangles}, either $P_b$ contains $i_{s,t}$ and touches $P_v$, or $P_b$ contains $i_{v,t}$ and touches $P_s$. In both cases, $P_b$ needs to have at least two bends, a contradiction.\\

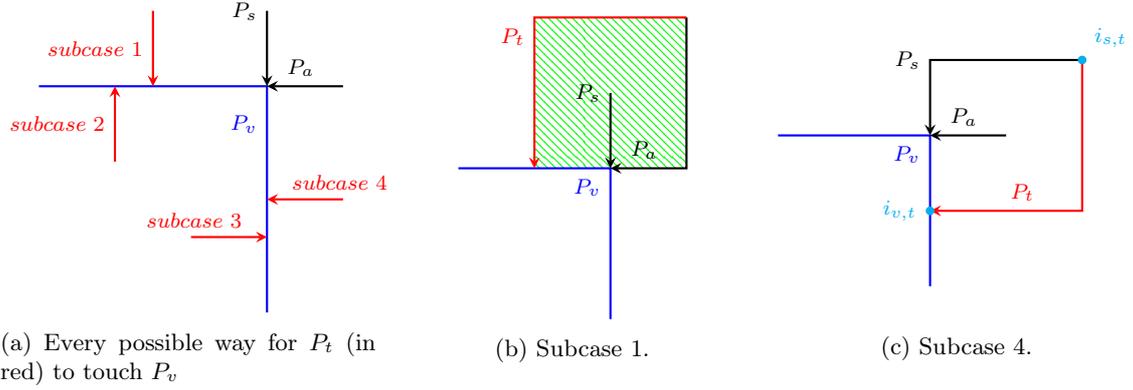
\begin{figure}[h]
\centering
\scalebox{1}
{
    \begin{subfigure}{.3\textwidth}
    \centering
    \begin{tikzpicture}[scale=1]
    
        \draw[blue,thick,-,>=stealth] (0,3) -- (3,3) -- (3,0) node[near start,above left] {\footnotesize $P_v$};
        \draw[thick,->,>=stealth] (3,4) -- (3,3) node[near start,above left] {\footnotesize $P_s$};
        \draw[thick,->,>=stealth] (4,3) -- (3,3) node[near start,above left] {\footnotesize $P_a$};
        \draw[red,thick,->,>=stealth] (1.5,4) -- (1.5,3) node[midway, left] {\footnotesize $subcase~1$};
        \draw[red,thick,->,>=stealth] (1,2) -- (1,3) node[midway,left] {\footnotesize $subcase~2$};
        \draw[red,thick,->,>=stealth] (2,1) -- (3,1) node[pos=.8,above left] {\footnotesize $subcase~3$};
        \draw[red,thick,->,>=stealth] (4,1.5) -- (3,1.5) node[pos=.8,above right] {\footnotesize $subcase~4$};

    \end{tikzpicture}
    \caption{Every possible way for $P_t$ (in red) to touch $P_v$}
    \label{fig:bend_cases}
    \end{subfigure}

    \begin{subfigure}{.3\textwidth}
    \centering
    \begin{tikzpicture}[scale=1]

        \fill [pattern=north west lines, pattern color=green] (2,3) rectangle (4,5);    
        \draw[blue,thick,-,>=stealth] (1,3) -- (3,3) -- (3,1) node[near start,above left] {\footnotesize $P_v$};
        \draw[thick,->,>=stealth] (3,4) -- (3,3) node[near start,above left] {\footnotesize $P_s$};
        \draw[thick,->,>=stealth] (4,5) -- (4,3) -- (3,3) node[near start,above left] {\footnotesize $P_a$};
        \draw[red,thick,->,>=stealth] (4,5) -- (2,5) -- (2,3) node[near start,above left] {\footnotesize $P_t$};
    
    \end{tikzpicture}
    \caption{Subcase 1.}
    \label{fig:bend_case_1}
    \end{subfigure}

    \begin{subfigure}{.3\textwidth}
    \centering
    \begin{tikzpicture}[scale=1]
    
        \draw[blue,thick,-,>=stealth] (1,3) -- (3,3) -- (3,1) node[near start,above left] {\footnotesize $P_v$};
        \draw[thick,->,>=stealth] (5,4) -- (3,4) -- (3,3) node[near start,above left] {\footnotesize $P_s$};
        \draw[thick,->,>=stealth] (4,3) -- (3,3) node[near start,above left] {\footnotesize $P_a$};
        \draw[red,thick,->,>=stealth] (5,4) -- (5,2) -- (3,2) node[near start,above left] {\footnotesize $P_t$};
        
        \node[cyan,circ] (vt) at (3,2) [label=left:{\color{cyan}\footnotesize $i_{v,t}$}] {};
        \node[cyan,circ] (st) at (5,4) [label=above right:{\color{cyan}\footnotesize $i_{s,t}$}] {};

        \node[invisible] at (3,.6) {};

    \end{tikzpicture}
    \caption{Subcase 4.}
    \label{fig:bend_case_4}
    \end{subfigure}
}
\caption{Case 1. $p$ is a bend-point of $P_v$.}
\label{fig:bend}
\end{figure}

\textit{Case 2. $p$ is not a bend-point of $P_v$.} By possibly rotating and reflecting, we may assume without loss of generality that $P_v$ goes through $p$ horizontally, with $P_s$ coming from above and $P_a$ from below (see Figures \ref{fig:up} and \ref{fig:down}). If $P_t$ touches the horizontal segment of $P_v$, it then follows from Observation \ref{obs:grid-points} that either $P_t$ and $P_s$ cannot touch, or $P_t$ and $P_a$ cannot touch, a contradiction. Thus, $P_v$ must have a bend and $P_t$ must touch a point of the vertical segment of $P_v$ distinct from the bend-point by Observation \ref{obs:grid-points}. Let us assume, without loss of generality, that the bend-point of $P_v$ is on the right extremity of its horizontal segment.

Suppose first that the vertical segment of $P_v$ is above its horizontal segment (see Figure \ref{fig:up}). If $P_t$ touches $P_v$ from the left, then $P_t$ and $P_a$ necessarily touch as depicted in Figure \ref{fig:up_left}. By Claim \ref{cl:path_P}, there is a path $P$ touching an endpoint of $P_t$ and the interior of $P_s$. In order to touch the interior of $P_s$, $P$ must belong to the green area depicted in Figure \ref{fig:up_left}; but then $P$ cannot touch an endpoint of $P_t$, a contradiction. Otherwise, $P_t$ touches $P_v$ from the right and so $P_t$ and $P_a$ necessarily touch as depicted in Figure \ref{fig:up_right}. Then, the only way for $P_s$ and $P_t$ to touch is depicted in Figure \ref{fig:up_right}; but then, $P_s \cap P_t$ and $P_v \cap P_t$ are reduced to one same point $i$ as shown in Figure \ref{fig:up_right} and so, by Fact \ref{a:triangles}, $P_b$ contains $i$, a contradiction.

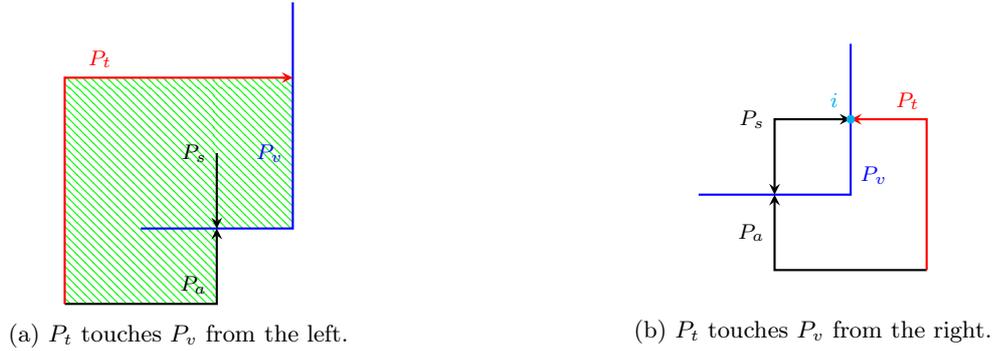
\begin{figure}[htb]
\centering
\scalebox{1}
{
    \begin{subfigure}{.5\textwidth}
    \centering
    \begin{tikzpicture}[scale=1]

	    \fill [pattern=north west lines, pattern color=green] (-1,-1) rectangle (1,2);
	    \fill [pattern=north west lines, pattern color=green] (1,0) rectangle (2,2);
    
        \draw[blue,thick,-,>=stealth] (0,0) -- (2,0) -- (2,3) node[near start,above left] {\footnotesize $P_v$};
        \draw[thick,->,>=stealth] (1,1) -- (1,0) node[near start,above left] {\footnotesize $P_s$};
        \draw[thick,->,>=stealth] (-1,-1) -- (1,-1) -- (1,0) node[near start,left] {\footnotesize $P_a$};
        \draw[red,thick,->,>=stealth] (-1,-1) -- (-1,2) -- (2,2) node[near start,above left] {\footnotesize $P_t$};

    \end{tikzpicture}
    \caption{$P_t$ touches $P_v$ from the left.}
    \label{fig:up_left}
    \end{subfigure}

    \begin{subfigure}{.5\textwidth}
    \centering
    \begin{tikzpicture}[scale=1]
    
        \draw[blue,thick,-,>=stealth] (0,0) -- (2,0) -- (2,2) node[near start,below right] {\footnotesize $P_v$};
        \draw[thick,<->,>=stealth] (2,1) -- (1,1) -- (1,0) node[near start,above left] {\footnotesize $P_s$};
        \draw[thick,->,>=stealth] (3,-1) -- (1,-1) -- (1,0) node[near start,above left] {\footnotesize $P_a$};
        \draw[red,thick,->,>=stealth] (3,-1) -- (3,1) -- (2,1) node[near start,above] {\footnotesize $P_t$};
        
        \node[cyan,circ] (i) at (2,1) [label=above left:{\color{cyan}\footnotesize $i$}] {};

        \node[invisible] at (2,-1.4) {};
        \node[invisible] at (2,2.5) {};

    \end{tikzpicture}
    \caption{$P_t$ touches $P_v$ from the right.}
    \label{fig:up_right}
    \end{subfigure}
}
\caption{$p$ is not a bend-point of $P_v$ and the vertical segment of $P_v$ is above its horizontal segment.}
\label{fig:up}
\end{figure}

Finally, suppose that the vertical segment of $P_v$ lies below its horizontal segment (see Figure \ref{fig:down}). If $P_t$ touches $P_v$ from the left, then $P_t$ and $P_s$ necessarily touch as depicted in Figure \ref{fig:down_left}. But then, $P_s \cap P_t$ is reduced to one point $i_{s,t}$ and $P_v \cap P_t$ is reduced to one point $i_{v,t}$ (see Figure \ref{fig:down_left}) and so, by Fact \ref{a:triangles}, either $P_b$ contains $i_{s,t}$ and touches $P_v$, or $P_b$ contains $i_{v,t}$ and touches $P_s$. In both cases $P_b$ needs to have at least two bends, a contradiction. Otherwise, $P_t$ touches $P_v$ from the right and so $P_t$ and $P_s$ necessarily touch as depicted in Figure \ref{fig:down_right}. But then, $P_s \cap P_t$ is reduced to one point $i_{s,t}$ and $P_v \cap P_t$ is reduced to one point $i_{v,t}$ (see Figure \ref{fig:down_right}) and so, by Fact \ref{a:triangles}, either $i_{s,t}$ belongs to $P_b$ and $P_b$ touches $P_v$, or $i_{v,t}$ belongs to $P_b$ and $P_b$ touches $P_s$. In both cases $P_b$ needs to have at least two bends, a contradiction.

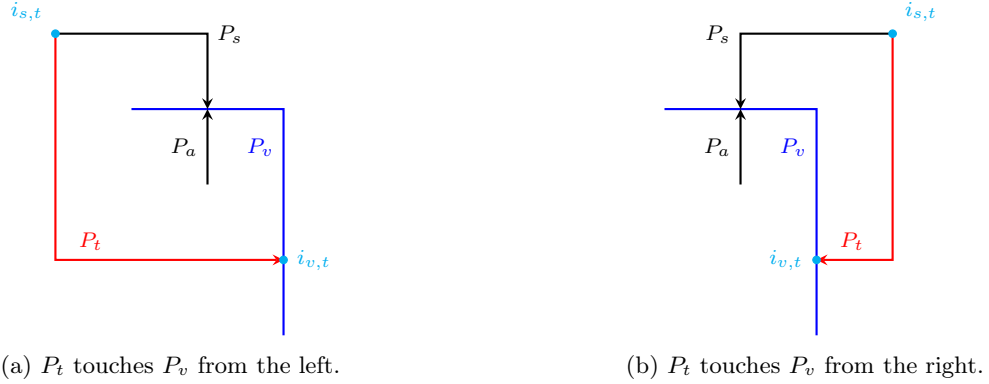
\begin{figure}[htb]
\centering
\scalebox{1}
{
    \begin{subfigure}{.5\textwidth}
    \centering
    \begin{tikzpicture}[scale=1]
    
        \draw[blue,thick,-,>=stealth] (0,0) -- (2,0) -- (2,-3) node[near start,above left] {\footnotesize $P_v$};
        \draw[thick,->,>=stealth] (-1,1) -- (1,1) -- (1,0) node[near start,above right] {\footnotesize $P_s$};
        \draw[thick,->,>=stealth] (1,-1) -- (1,0) node[near start,above left] {\footnotesize $P_a$};
        \draw[red,thick,->,>=stealth] (-1,1) -- (-1,-2) -- (2,-2) node[near start,above left] {\footnotesize $P_t$};
        
        \node[cyan,circ] (vt) at (2,-2) [label=right:{\color{cyan}\footnotesize $i_{v,t}$}] {};
        \node[cyan,circ] (st) at (-1,1) [label=above left:{\color{cyan}\footnotesize $i_{s,t}$}] {};

    \end{tikzpicture}
    \caption{$P_t$ touches $P_v$ from the left.}
    \label{fig:down_left}
    \end{subfigure}

    \begin{subfigure}{.5\textwidth}
    \centering
    \begin{tikzpicture}[scale=1]
    
        \draw[blue,thick,-,>=stealth] (0,0) -- (2,0) -- (2,-3) node[near start,above left] {\footnotesize $P_v$};
        \draw[thick,->,>=stealth] (3,1) -- (1,1) -- (1,0) node[near start,above left] {\footnotesize $P_s$};
        \draw[thick,->,>=stealth] (1,-1) -- (1,0) node[near start,above left] {\footnotesize $P_a$};
        \draw[red,thick,->,>=stealth] (3,1) -- (3,-2) -- (2,-2) node[near start,above left] {\footnotesize $P_t$};
        
        \node[cyan,circ] (vt) at (2,-2) [label=left:{\color{cyan}\footnotesize $i_{v,t}$}] {};
        \node[cyan,circ] (st) at (3,1) [label=above right:{\color{cyan}\footnotesize $i_{s,t}$}] {};

    \end{tikzpicture}
    \caption{$P_t$ touches $P_v$ from the right.}
    \label{fig:down_right}
    \end{subfigure}
}
\caption{$p$ is not a bend-point of $P_v$ and the vertical segment of $P_v$ is below its horizontal segment.}
\label{fig:down}
\end{figure}

As in both Case 1 and Case 2 we obtained a contradiction, $P_v$ cannot have two free endpoints in $\mathcal{R}$.
\end{proof}

\subsection{The reduction}
\label{subsec:reduction}

Given an instance $\Phi$ of \textsc{Planar Exactly 3-Bounded 3-Sat}, with variable set $V$ and clause set $C$, we construct a graph $G_1(\Phi)$ (resp. $G_2(\Phi)$; $G_3(\Phi)$) which is $B_1$-CPG (resp. $B_2$-CPG; planar $B_k$-CPG with $k \geq 3$) if and only if $\Phi$ is satisfiable. We may assume that each variable has at most 2 positive occurrences and at most 2 negated occurrences (recall that each variable appears in 3 clauses): indeed, if some variable has only positive (resp. negated) occurrences, we may reduce the formula by setting this variable to TRUE (resp. FALSE). The graphs $G_1(\Phi)$, $G_2(\Phi)$ and $G_3(\Phi)$ differ only in the end-eating graph used and are constructed as follows.
Starting from the planar bipartite graph $H$, with $V(H) = V \cup C$ and $E(H) = \{xc : x \in c \text{ or } \neg x \in c\}$, we replace each variable vertex $x \in V$ with the \textit{variable gadget $V_x$} depicted in Figure \ref{fig:variablegadget}, and each clause vertex $c \in C$ with the \textit{clause gadget $C_c$} depicted in Figure \ref{fig:clausegadget}. In every subsequent figure, the double-circles schematically represent the copy of the end-eating graph $\mathcal{E}_i$ for $G_i(\Phi)$ ($i=1,2,3$) connected to vertices. More precisely, a vertex $t \in V(G_i(\Phi))$ with $i=2,3$ (resp. $t \in V(G_1(\Phi))$) is incident to a double-circle if and only if $t$ is adjacent to a vertex $v$ of $\mathcal{E}_i$ (resp. to \textit{the} vertex $v$ of $\mathcal{E}_1$). It follows from Lemmas \ref{lem:Ei} and \ref{lem:E1} that in any CPG representation of $G_i(\Phi)$, one endpoint of $P_t$ belongs to $P_v$. In the following, the pairs of vertices $(b,c)$ and $(d,a)$ (resp. $(a,b)$ and $(c,d)$) in $V_x$ are called \textit{positive terminals} (resp. \textit{negative terminals}); the pairs of vertices $(p_i,q_i)$, for $i = 1,2,3$, in $C_c$ are called \textit{terminals}; furthermore, we denote by $\overline{V_x}$ the graph $V_x$ without the end-eating graphs.

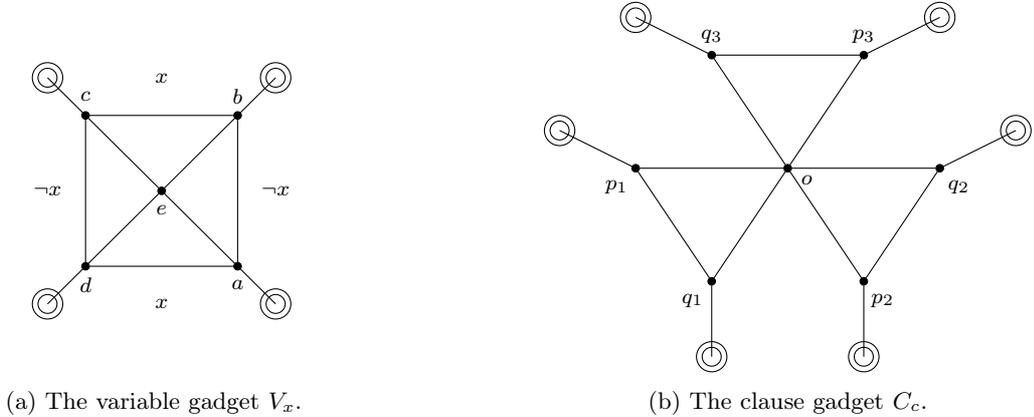
\begin{figure}[h]
\centering
\scalebox{1}
{

\begin{subfigure}{0.5\textwidth}
\centering
\begin{tikzpicture}[scale=0.5,transform shape]

\path pic (v1) {variablebis};
\node [scale=2] at (6,1.5) {\footnotesize $a$};
\node [scale=2] at (6,6.5) {\footnotesize $b$};
\node [scale=2] at (2,6.5) {\footnotesize $c$};
\node [scale=2] at (2,1.5) {\footnotesize $d$};
\node [scale=2] at (4,3.5) {\footnotesize $e$};

\node [scale=2] at (4,1) {\footnotesize $x$};
\node [scale=2] at (4,7) {\footnotesize $x$};
\node [scale=2] at (1,4) {\footnotesize $\neg x$};
\node [scale=2] at (7,4) {\footnotesize $\neg x$};

\node[invisible] at (0,9) {};
\node[invisible] at (0,-.8) {};
\end{tikzpicture}
\caption{The variable gadget $V_x$.}
\label{fig:variablegadget}
\end{subfigure}

\begin{subfigure}{0.5\textwidth}
\centering
\begin{tikzpicture}[scale=0.5, transform shape]

\path pic (clause1) {clausebis};
\node [scale=2] at (6.5,2.7) {\footnotesize $o$};
\node [scale=2] at (1.5,2.5) {\footnotesize $p_1$};
\node [scale=2] at (3.5,-0.5) {\footnotesize $q_1$};
\node [scale=2] at (8.5,-0.5) {\footnotesize $p_2$};
\node [scale=2] at (10.5,2.5) {\footnotesize $q_2$};
\node [scale=2] at (8,6.5) {\footnotesize $p_3$};
\node [scale=2] at (4,6.5) {\footnotesize $q_3$};

\end{tikzpicture}
\caption{The clause gadget $C_c$.}
\label{fig:clausegadget}
\end{subfigure}

}
\caption{The gadgets for the reduction (the double-circles schematically represent the copie of the end-eating gadget $\mathcal{E}_i$ for $G_i(\Phi)$ connected to the vertices).}
\label{fig:gadgets}
\end{figure}

We next replace each edge $xc$ of $E(H)$ by a \textit{connector} $L_{x,c}$ (see Figure \ref{fig:connector1}): the vertex $t_{xc}$ is linked to a positive or negative terminal of $V_x$ depending on whether $x$ occurs positive or negated in $c$; and the vertex $t_{cx}$ is linked to a terminal of $C_c$. These transformations are done in such a way that no (positive, negative) terminal is used twice. Finally, if a clause $c$ has only two literals, we link the last free terminal to a \textit{false terminator} (see Figure \ref{fig:falseterminator}).

\begin{figure}[h]
\centering
\scalebox{1}
{
    \begin{subfigure}{.6\textwidth}
    \centering
    \begin{tikzpicture}[scale=0.3]
    
        \node[circ] (txc) at (-5,0) [label=below:{\footnotesize $t_{xc}$}] {};
        \node[circ] (tcx) at (5,0) [label=below:{\footnotesize $t_{cx}$}] {};
        \node[circ] (x1) at (-8,2) [] {};
        \node[circ] (x2) at (-8,-2) [] {};
        \node[circ] (c1) at (8,2) [] {};
        \node[circ] (c2) at (8,-2) [] {};
        \node[circ] (t1) at (-2.5,0) [] {};
        \node[circ] (t2) at (-1.25,-.5) [] {};
        \node[circ] (t3) at (0,0) [] {};
        \node[circ] (t4) at (1.25,-.5) [] {};
        \node[circ] (t5) at (2.5,0) [] {};
        \draw[thick]
            (txc) -- (x1) -- (x2) -- (txc) -- (t1) -- (t2) -- (t3) -- (t4) -- (t5) -- (tcx) -- (c1) -- (c2) -- (tcx)
            (-9,2.5) -- (x1) -- (-9,1.5)
            (-9,-2.5) -- (x2) -- (-9,-1.5)
            (c1) -- (9,1.5)
            (c2) -- (9,-1.5)
            (t1) -- (-2.5,2)
            (t2) -- (-1.25,-2.5)
            (t3) -- (0,2)
            (t4) -- (1.25,-2.5)
            (t5) -- (2.5,2);
        \draw[dashed]
            (-9,2.5) -- (-10,3) (-10,1) -- (-9,1.5)
            (-9,-2.5) -- (-10,-3) (-10,-1) -- (-9,-1.5)
            (10,1) -- (9,1.5)
            (10,-1) -- (9,-1.5);
        \draw[dotted,thick]
            (-8,-4) arc (-30:30:8)
            (8,4) arc (150:210:8);
	    \draw[]
	        (-2.5,2) circle (0.5) circle (0.8)
	        (-1.25,-2.5) circle (0.5) circle (0.8)
	        (0,2) circle (0.5) circle (0.8)
	        (1.25,-2.5) circle (0.5) circle (0.8)
	        (2.5,2) circle (0.5) circle (0.8);
        \node at (-12,0) {$V_x$};
        \node at (12,0) {$C_c$};
    
    \end{tikzpicture}
    \caption{Connector between $V_x$ and $C_c$.}
    \label{fig:connector1}
    \end{subfigure}

    \begin{subfigure}{.4\textwidth}
    \centering
    \begin{tikzpicture}[scale=0.3]
    
        \node[circ] (tc) at (0,0) [label=below:{\footnotesize $t_f$}] {};
        \node[circ] (c1) at (-3,2) [] {};
        \node[circ] (c2) at (-3,-2) [] {};
        \draw[thick]
            (tc) -- (c1) -- (c2) -- (tc) -- (3,0)
            (c1) -- (-4,1.5)
            (c2) -- (-4,-1.5)
            (3,0) circle (0.5) circle (0.8);
        \draw[dashed]
            (-5,1) -- (-4,1.5)
            (-5,-1) -- (-4,-1.5);
        \draw[dotted,thick]
            (-3,-4) arc (-30:30:8);
        \node at (-7,0) {$C_c$};
    
    \end{tikzpicture}
    \caption{False terminator.}
    \label{fig:falseterminator}
    \end{subfigure}
}
\caption{The connector and false terminator.}
\end{figure}
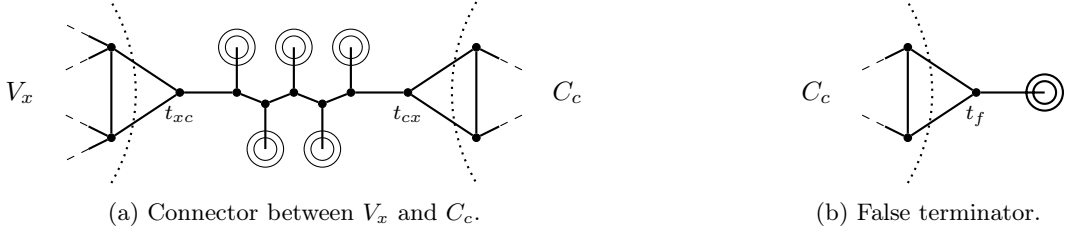

We now define the \textit{$\mathcal{R}$-value} of a (positive, negative) terminal given a CPG representation $\mathcal{R}$ of $G_i(\Phi)$, for $i=1,2,3$, as follows. A terminal $(y,z)$ belonging to a clause gadget has an $\mathcal{R}$-value of $0$ if there is a grid-point $r$ which is an endpoint of both $P_y$ and $P_z$ and an interior point of $P_o$; otherwise, it has an $\mathcal{R}$-value of 1. A (positive, negative) terminal $(y,z)$ belonging to a variable gadget has an $\mathcal{R}$-value of 0 if there exists no grid-point belonging to $(P_y \cap P_z) \setminus P_e$; otherwise, it has an $\mathcal{R}$-value of $1$. We next prove the following claims.

\begin{Claim}
\label{cl:clause_terminals}
Let $\mathcal{R}$ be a CPG (resp. 1-bend CPG) representation of $G_i(\Phi)$ with $i=2,3$ (resp. $G_1(\Phi)$). Then in every clause gadget $C_c$, at least one terminal has an $\mathcal{R}$-value of $0$.
\end{Claim}

\begin{cproof}
Since $P_o$ has two endpoints, there exists at least one terminal $(p_i,q_i)$ such that no endpoint of $P_o$ belongs to $P_{p_i} \cup P_{q_i}$. It follows that one endpoint $r$ of $P_{p_i}$ and one endpoint $r'$ of $P_{q_i}$ belong to $\mathring{P_o}$. On the other hand, since $P_{p_i}$ and $P_{q_i}$ each have one endpoint taken by an end-eating graph (Lemmas \ref{lem:Ei} and \ref{lem:E1}), necessarily $r=r'$, for otherwise $P_{p_i}$ and $P_{q_i}$ would not touch. Thus, the terminal $(p_i,q_i)$ has an $\mathcal{R}$-value of $0$.
\end{cproof}

\begin{Claim}
\label{cl:var_terminals}
Let $\mathcal{R}$ be a CPG (resp. 1-bend CPG) representation of $G_i(\Phi)$ with $i=2,3$ (resp. $G_1(\Phi)$) and let $V_x$ be a variable gadget. Then in the representation induced by $\mathcal{R}$ of a triangle $(x,y,e)$ of $\overline{V_x}$, exactly one endpoint of either $P_x$ or $P_y$ is used. Furthermore, if one endpoint of $P_x$ is used, then one endpoint of $P_z$ is used in the representation induced by $\mathcal{R}$ of the triangle $(x,z,e)$ of $\overline{V_x}$.
\end{Claim}

\begin{cproof}
Clearly, since the edge $xy$ has to be represented in $\mathcal{R}$, either one endpoint of $P_x$ belong to $P_y$ or one endpoint of $P_y$ belongs to $P_x$. Suppose, to the contrary, that $P_x$ and $P_y$ have a common endpoint. It follows from Lemmas \ref{lem:Ei} and \ref{lem:E1} that no endpoint of $P_x$ (resp. $P_y$) belongs to $P_z$ (resp. $P_t$), where $(x,z,e)$ (resp. $(y,t,e)$) is a triangle of $\overline{V_x}$. But then, one endpoint of $P_z$ (resp. $P_t$) must be used to represent the edge $xz$ (resp. $yt$) and we conclude by Lemmas \ref{lem:Ei} and \ref{lem:E1} that no endpoint is available to represent the edge $zt$, a contradiction. Thus, $P_x$ and $P_y$ have distinct endpoints and we may assume without loss of generality that $P_x$ touches $\mathring{P_y}$. Since by Lemmas \ref{lem:Ei} and \ref{lem:E1} the other endpoint of $P_x$ is taken by an end-eating graph, $P_z$ must necessarily touch the interior of $P_x$, where $(x,z,e)$ is the other triangle of $\overline{V_x}$ to which $x$ belongs, which concludes the proof.
\end{cproof}

\begin{Claim}
\label{cl:connector}
Let $\mathcal{R}$ be a CPG (resp. 1-bend CPG) representation of $G_i(\Phi)$ with $i=2,3$ (resp. $G_1(\Phi)$). Then in every connector $L_{x,c}$, at least one of $t_{cx}$ and $t_{xc}$ is connected to a terminal of $\mathcal{R}$-value $1$.
\end{Claim}

\begin{cproof}
Observe first that if a vertex $t$ belonging to $L_{x,c}$ is connected to a terminal $(y,z)$ of $\mathcal{R}$-value~$0$, then one endpoint of $P_t$ is a point of $\mathring{P_y}$ and the other is a point of $\mathring{P_z}$. Indeed, if $(y,z)$ belongs to a clause gadget, then $P_y$ and $P_z$ share a common endpoint which is a point of $\mathring{P_o}$ by definition. Since their other endpoints belong to an end-eating graph, $P_t$ must touch the interior of $P_y$ and the interior of $P_z$. Now if $(y,z)$ belongs to a variable gadget, then either $P_y$ and $P_z$ share a common endpoint which also belongs to $P_e$ by definition and so, $P_t$ must touch $\mathring{P_y}$ and $\mathring{P_z}$, as $P_y$ and $P_z$ have their other endpoints taken by an end-eating graph. Otherwise, we may assume without loss of generality that $P_y$ touches the interior of $P_z$ on a grid-point which is an endpoint of $P_e$ by definition. But then, by repeated applications of Claim \ref{cl:var_terminals}, one endpoint of $P_z$ belongs to some path $P_r$ with $r \neq y,e$ and so, since $P_z$ and $P_y$ have their other endpoints taken by an end-eating graph, $P_t$ must touch $\mathring{P_y}$ and $\mathring{P_z}$. 

Suppose now, to the contrary, that both $t_{cx}$ and $t_{xc}$ are connected to a terminal of $\mathcal{R}$-value $0$. Since the induced path $Q$ connecting $t_{cx}$ and $t_{xc}$ in $L_{x,c}$ has 6 edges, its CPG representation requires at least 6 distinct endpoints of paths corresponding to vertices in $Q$. However, as shown previously, both endpoints of $P_{t_{cx}}$ and $P_{t_{xc}}$ are used and the remaining vertices are connected to end-eating graphs; thus, only 5 endpoints are available, a contradiction.
\end{cproof}

\bigskip

Finally, we prove the following key Lemma.

\begin{lemma}\label{keylemma}
Given an instance $\Phi$ of \textsc{Planar Exactly 3-Bounded 3-Sat}, the following hold.
\begin{itemize}
    \item The graph $G_1(\Phi)$ is $B_1$-CPG if $\Phi$ is satisfiable, but has no $1$-bend CPG representation if $\Phi$ is not satisfiable.
    \item The graph $G_i(\Phi)$, for $i=2,3$, is $B_i$-CPG if $\Phi$ is satisfiable, but has no $k$-bend CPG representation, for any $k \geq 0$, if $\Phi$ is not satisfiable.
\end{itemize}
\end{lemma}

\begin{proof}
Suppose that there exists a $k$-bend CPG (resp. 1-bend CPG) representation $\mathcal{R}$ of $G_i(\Phi)$ with $i=2,3$ (resp. $G_1(\Phi)$) for some $k \geq 0$. We construct a truth assignment satisfying $\Phi$ as follows. By Lemma \ref{lem:freeendpoint}, the representation in $\mathcal{R}$ of $\overline{V_x}$ must contain at least $(8-2 \times 5 + 4)=2$ contact points corresponding to pairwise edge-disjoint triangles in $\overline{V_x}$; and since $\overline{V_x}$ does not contain three pairwise edge-disjoint triangles, the representation in $\mathcal{R}$ of $\overline{V_x}$ contains exactly two such contact points $r$ and $r'$. We claim that $r$ and $r'$ are the two endpoints of $P_e$. Indeed, if say $r$ is not an endpoint of $P_e$, then $r$ is an endpoint of the two other paths, say $P_a$ and $P_b$ without loss of generality. Then, $r'$ necessarily belongs to $P_c \cup P_d \cup P_e$, as $(c,d,e)$ is the only triangle in $\overline{V_x}$ not sharing any edge with $(a,b,e)$. In particular, $r'$ is an endpoint of at least one of $P_c$ and $P_d$, say without loss of generality $P_c$. But then $P_b$ and $P_c$ cannot touch, as their other endpoint is taken by an end-eating graph, a contradiction. Thus, $r$ and $r'$ are the two endpoints of $P_e$. It follows that, for any variable $x \in V$, there are only two ways to represent $\overline{V_x}$: either $r$ and $r'$ correspond to $(a,b,e)$ and $(c,d,e)$, respectively, or $r$ and $r'$ correspond to $(a,d,e)$ and $(b,c,e)$, respectively. 

Consider the first case. We claim that $(a,b)$ and $(c,d)$ have $\mathcal{R}$-value 0, whereas $(a,d)$ and $(b,c)$ have $\mathcal{R}$-value 1. Clearly, $r$ is an endpoint of either $P_a$ or $P_b$. By symmetry, it is enough to show the claim when $r$ is an endpoint of $P_a$. But if $r$ is an endpoint of $P_a$, then $r'$ is an endpoint of $P_c$, for if $r'$ is an endpoint of $P_d$, $P_a$ and $P_d$ cannot touch as their other endpoints are taken by end-eating graphs. Therefore, $P_d$ touches $\mathring{P_a}$ and $P_{b}$ touches $\mathring{P_c}$. This implies that $(a,b)$ and $(c,d)$ have $\mathcal{R}$-value 0, whereas $(a,d)$ and $(b,c)$ have $\mathcal{R}$-value 1.  

In the second case, we conclude similarly that $(a,d)$ and $(b,c)$ have $\mathcal{R}$-value 0, whereas $(a,b)$ and $(c,d)$ have $\mathcal{R}$-value 1 (we leave the verification as an easy exercise). 

We then set $x$ to TRUE if the positive terminals have $\mathcal{R}$-value 1 and set $x$ to FALSE if the negative terminals have $\mathcal{R}$-value 1. By Claim~\ref{cl:clause_terminals}, every clause gadget has at least one terminal of $\mathcal{R}$-value $0$. By Claim \ref{cl:connector}, this terminal is connected to a terminal of $\mathcal{R}$-value $1$ in the variable gadget. Therefore, each clause contains a TRUE literal and so $\Phi$ is satisfiable.\\

Conversely, suppose that $\Phi$ is satisfiable. We build a 1-bend CPG representation $\mathcal{R}$ of $G_1(\Phi)$, $G_2(\Phi)$ and $G_3(\Phi)$ where all the copies of the end-eating graphs have been removed (note that without the copies of the end-eating graphs, $G_1(\Phi)$, $G_2(\Phi)$ and $G_3(\Phi)$ are isomorphic); we then add $i$-bend CPG representations of the copies of the end-eating graph $\mathcal{E}_i$ to the representation $\mathcal{R}$, without adding any bends to the existing paths, to obtain an $i$-bend CPG representation $\mathcal{R}_i$ of $G_i(\Phi)$. In all subsequent figures, double-circles schematically represent a CPG representation of the copy of the appropriate end-eating graph.

Consider a truth assignment of the variables that satisfies $\Phi$. We set the $\mathcal{R}$-values of terminals beforehand as follows. For every variable $x \in V$, if $x$ is set to TRUE (resp. FALSE), then the positive terminals of $V_x$ have an $\mathcal{R}$-value of $1$ (resp. $0$) and the negative terminals have an $\mathcal{R}$-value of $0$ (resp. $1$). For every clause $c \in C$, we then choose a variable $x$ whose occurrence satisfies $c$: the $\mathcal{R}$-value of the terminal of $C_c$ connected to $V_x$ is then set to $0$ while the $\mathcal{R}$-values of the other terminals of $C_c$ are set to $1$. It is worth noticing that in this way no connector is linking two terminals of $\mathcal{R}$-value $0$.

Recall that $H$ is the graph with vertex set $V(H) = V \cup C$ and edge set $E(H) = \{xc : x \in c \text{ or } \neg x \in c\}$. Since $H$ is planar with $\Delta(H) \leq 4$, we can find in linear time an embedding $\mathcal{H}$ on the grid where vertices are mapped to grid-points and edges are mapped to pairwise interiorly disjoint grid-paths with at most 4 bends connecting the two grid-points corresponding to the endvertices \citep{tamassia}. For every vertex $u \in V(H)$, we then replace the corresponding grid-point in $\mathcal{H}$ with an empty square that we assume to be large enough to contain the representations we subsequently construct. 

Given a variable $x \in V$, we construct a 1-bend CPG representation of $V_x$ as follows. We denote by $E_1$, $E_2$ and $E_3$ the grid-paths in $\mathcal{H}$ corresponding to the edges incident to $x$ in $H$. For any $j=1,2,3$, the edge $E_j$ corresponds to a connector incident to $V_x$; we denote by $t_j$ the vertex in the corresponding connector adjacent to vertices in $V_x$, and by $P_j$ the path in the CPG representation we are constructing corresponding to $t_j$. We now explain how to construct the CPG representation of $V_x \cup \{t_1,t_2,t_3\}$ depending on the $\mathcal{R}$-value of the terminal to which $t_j$ ($j=1,2,3$) is connected. Recall that, by construction, $t_1$, $t_2$ and $t_3$ are not all connected to terminals of the same $\mathcal{R}$-value.

Suppose first that exactly two of them are connected to terminals of $V_x$ of $\mathcal{R}$-value $0$, say $t_1$ and $t_2$ without loss of generality. If $E_1$ and $E_2$ are on opposite sides of $x$ in $\mathcal{H}$, then we fill the corresponding square as shown in Figure \ref{fig:001_oppose}. Otherwise, we may assume without loss of generality that $E_1$ and $E_3$ are on opposite sides of $x$ in $\mathcal{H}$ in which case we fill the corresponding square as shown in Figure \ref{fig:001_other} (the blue dotted paths in both figures correspond to the neighbours of $t_1$ and $t_2$ in their respective connector gadget). Note that since the terminal to which $t_3$ is connected has $\mathcal{R}$-value 1, the corresponding paths have a contact point $r$ which belongs to no other path; $P_3$ is then added in such a way that $r$ is one of its endpoints.

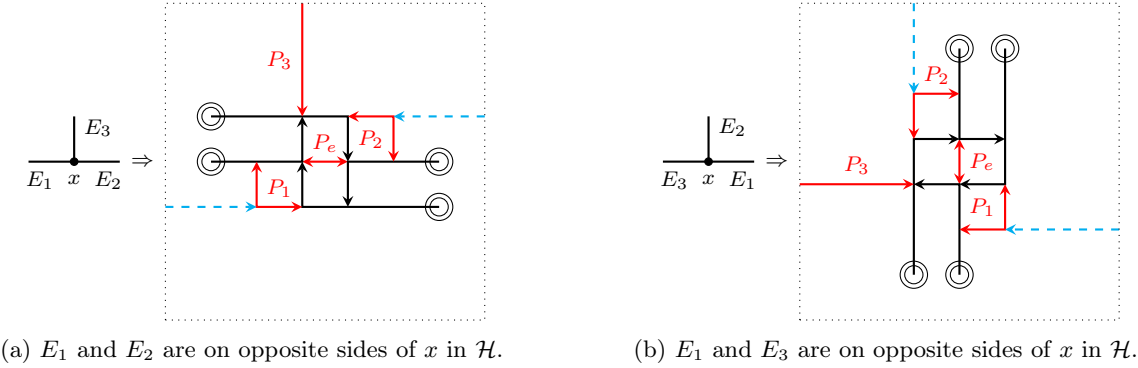
\begin{figure}[h]
\centering
\scalebox{1}
{
    \begin{subfigure}{.5\textwidth}
    \centering
    \begin{tikzpicture}[scale=.6]
 
        \node[circ,label=below:{\footnotesize $x$}] at (-5,1) {};
        \draw[-,thick] (-6,1) -- (-5,1) node[near start,below] {\footnotesize $E_1$};
        \draw[-,thick] (-4,1) -- (-5,1)  node[near start,below] {\footnotesize $E_2$};
        \draw[-,thick] (-5,2) -- (-5,1) node[near start,right] {\footnotesize $E_3$};
        
        \node[draw=none] at (-3.5,1) {\footnotesize $\Rightarrow$};
    
        \draw[red,thick,<->,>=stealth] (0,1) -- (1,1) node[midway,above] {\footnotesize $P_e$};
        \draw[thick,->,>=stealth] (-2,2) -- (1,2) -- (1,1) node[near start,above] {};
        \draw[thick,->,>=stealth] (-2,1) -- (0,1) -- (0,2) node[near start,above] {};
        \draw[thick,->,>=stealth] (3,1) -- (1,1) -- (1,0) node[near start,above] {};
        \draw[thick,->,>=stealth] (3,0) -- (0,0) -- (0,1) node[near start,above] {};
        
        \draw[red,thick,<->,>=stealth] (-1,1) -- (-1,0) -- (0,0) node[midway,above] {\footnotesize $P_1$};
        \draw[red,thick,<->,>=stealth] (1,2) -- (2,2) -- (2,1) node[midway,left] {\footnotesize $P_2$};
        \draw[red,thick,->,>=stealth] (0,4.5) -- (0,2) node[midway,left] {\footnotesize $P_3$};
        \draw[dashed,cyan,thick,->,>=stealth] (4,2) -- (2,2) node[near start,above] {};
        \draw[dashed,cyan,thick,->,>=stealth] (-3,0) -- (-1,0) node[near start,above] {};
        
        \draw
        (-2,2) circle (0.2) circle (0.3)
        (-2,1) circle (0.2) circle (0.3)
        (3,1) circle (0.2) circle (0.3)
        (3,0) circle (0.2) circle (0.3);
        
        \draw[dotted] (-3,-2.5) rectangle (4,4.5);
        
    \end{tikzpicture}
    \caption{$E_1$ and $E_2$ are on opposite sides of $x$ in $\mathcal{H}$.}
    \label{fig:001_oppose}
    \end{subfigure}

    \begin{subfigure}{.5\textwidth}
    \centering
    \begin{tikzpicture}[scale=.6]
    
        \node[circ,label=below:{\footnotesize $x$}] at (-4.5,.5) {};
        \draw[-,thick] (-5.5,.5) -- (-4.5,.5) node[near start,below] {\footnotesize $E_3$};
        \draw[-,thick] (-3.5,.5) -- (-4.5,.5)  node[near start,below] {\footnotesize $E_1$};
        \draw[-,thick] (-4.5,1.5) -- (-4.5,.5) node[near start,right] {\footnotesize $E_2$};
        
        \node[draw=none] at (-3,.5) {\footnotesize $\Rightarrow$};
    
        \draw[red,thick,<->,>=stealth] (1,0) -- (1,1) node[midway,right] {\footnotesize $P_e$};
        \draw[thick,->,>=stealth] (0,-2) -- (0,1) -- (1,1) node[near start,above] {};
        \draw[thick,->,>=stealth] (1,-2) -- (1,0) -- (0,0) node[near start,above] {};
        \draw[thick,->,>=stealth] (1,3) -- (1,1) -- (2,1) node[near start,above] {};
        \draw[thick,->,>=stealth] (2,3) -- (2,0) -- (1,0) node[near start,above] {};
        
        \draw[red,thick,<->,>=stealth] (1,-1) -- (2,-1) -- (2,0) node[midway,left] {\footnotesize $P_1$};
        \draw[red,thick,<->,>=stealth] (0,1) -- (0,2) -- (1,2) node[midway,above] {\footnotesize $P_2$};
        \draw[red,thick,->,>=stealth] (-2.5,0) -- (0,0) node[midway,above] {\footnotesize $P_3$};
        \draw[dashed,cyan,thick,->,>=stealth] (0,4) -- (0,2) node[near start,above] {};
        \draw[dashed,cyan,thick,->,>=stealth] (4.5,-1) -- (2,-1) node[near start,above] {};
        
        \draw
        (0,-2) circle (0.2) circle (0.3)
        (1,-2) circle (0.2) circle (0.3)
        (1,3) circle (0.2) circle (0.3)
        (2,3) circle (0.2) circle (0.3);
        
        \draw[dotted] (-2.5,-3) rectangle (4.5,4);
	
    \end{tikzpicture}
    \caption{$E_1$ and $E_3$ are on opposite sides of $x$ in $\mathcal{H}$.}
    \label{fig:001_other}
    \end{subfigure}
}
\caption{The 1-bend CPG representation of $V_x$ when exactly two terminals have an $\mathcal{R}$-value of $0$ (the dotted square in both figures represents the square corresponding to $x$).}
\label{fig:variable_repr_001}
\end{figure}

Suppose now that exactly one of $t_1$, $t_2$ and $t_3$ is connected to a terminal of $V_x$ of $\mathcal{R}$-value $0$, say $t_3$ without loss of generality. If $E_1$ and $E_2$ are on opposite sides of $x$ in $\mathcal{H}$, then we fill the corresponding square as shown in Figure \ref{fig:110_oppose}. Otherwise, we may assume without loss of generality that $E_1$ and $E_3$ are on opposite sides of $x$ in $\mathcal{H}$ in which case we fill the corresponding square as shown in Figure \ref{fig:110_other} (the blue dotted path in both figures corresponds to the neighbour of $t_3$ in its connector gadget). Note that since the terminal to which $t_1$ (resp. $t_2$) is connected has $\mathcal{R}$-value 1, the corresponding paths have a contact point $r$ (resp. $r'$) which belongs to no other path; $P_1$ (resp. $P_2$) is then added in such a way that $r$ (resp. $r'$) is one of its endpoints.

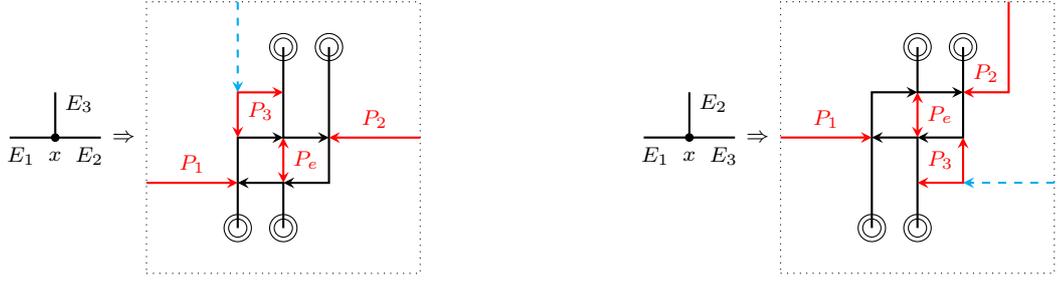
\begin{figure}[h]
\centering
\scalebox{1}
{
    \begin{subfigure}{.5\textwidth}
    \centering
    \begin{tikzpicture}[scale=.6]
    
        \node[circ,label=below:{\footnotesize $x$}] at (-4,1) {};
        \draw[-,thick] (-5,1) -- (-4,1) node[near start,below] {\footnotesize $E_1$};
        \draw[-,thick] (-3,1) -- (-4,1)  node[near start,below] {\footnotesize $E_2$};
        \draw[-,thick] (-4,2) -- (-4,1) node[near start,right] {\footnotesize $E_3$};
        
        \node[draw=none] at (-2.5,1) {\footnotesize $\Rightarrow$};
    
        \draw[red,thick,<->,>=stealth] (1,0) -- (1,1) node[midway,right] {\footnotesize $P_e$};
        \draw[thick,->,>=stealth] (0,-1) -- (0,1) -- (1,1) node[near start,above] {};
        \draw[thick,->,>=stealth] (1,-1) -- (1,0) -- (0,0) node[near start,above] {};
        \draw[thick,->,>=stealth] (1,3) -- (1,1) -- (2,1) node[near start,above] {};
        \draw[thick,->,>=stealth] (2,3) -- (2,0) -- (1,0) node[near start,above] {};
        
        \draw[red,thick,<->,>=stealth] (0,1) -- (0,2) -- (1,2) node[midway,below] {\footnotesize $P_3$};
        \draw[red,thick,->,>=stealth] (-2,0) -- (0,0) node[midway,above] {\footnotesize $P_1$};
        \draw[red,thick,->,>=stealth] (4,1) -- (2,1) node[midway,above] {\footnotesize $P_2$};
        \draw[dashed,cyan,thick,->,>=stealth] (0,4) -- (0,2) node[near start,above] {};
        
        \draw
        (0,-1) circle (0.2) circle (0.3)
        (1,-1) circle (0.2) circle (0.3)
        (1,3) circle (0.2) circle (0.3)
        (2,3) circle (0.2) circle (0.3);
        
        \draw[dotted] (-2,-2) rectangle (4,4);
        
    \end{tikzpicture}
    \caption{$E_1$ and $E_2$ are on opposite sides of $x$ in $\mathcal{H}$.}
    \label{fig:110_oppose}
    \end{subfigure}

    \begin{subfigure}{.5\textwidth}
    \centering
    \begin{tikzpicture}[scale=.6]
    
        \node[circ,label=below:{\footnotesize $x$}] at (-4,0) {};
        \draw[-,thick] (-5,0) -- (-4,0) node[near start,below] {\footnotesize $E_1$};
        \draw[-,thick] (-3,0) -- (-4,0)  node[near start,below] {\footnotesize $E_3$};
        \draw[-,thick] (-4,1) -- (-4,0) node[near start,right] {\footnotesize $E_2$};
        
        \node[draw=none] at (-2.5,0) {\footnotesize $\Rightarrow$};
    
        \draw[red,thick,<->,>=stealth] (1,0) -- (1,1) node[midway,right] {\footnotesize $P_e$};
        \draw[thick,->,>=stealth] (0,-2) -- (0,1) -- (1,1) node[near start,above] {};
        \draw[thick,->,>=stealth] (1,-2) -- (1,0) -- (0,0) node[near start,above] {};
        \draw[thick,->,>=stealth] (1,2) -- (1,1) -- (2,1) node[near start,above] {};
        \draw[thick,->,>=stealth] (2,2) -- (2,0) -- (1,0) node[near start,above] {};
        
        \draw[red,thick,<->,>=stealth] (1,-1) -- (2,-1) -- (2,0) node[midway,left] {\footnotesize $P_3$};
        \draw[red,thick,->,>=stealth] (-2,0) -- (0,0) node[midway,above] {\footnotesize $P_1$};
        \draw[red,thick,->,>=stealth] (3,3) -- (3,1) -- (2,1) node[midway,above] {\footnotesize $P_2$};
        \draw[dashed,cyan,thick,->,>=stealth] (4,-1) -- (2,-1) node[near start,above] {};
        
        \draw
        (0,-2) circle (0.2) circle (0.3)
        (1,-2) circle (0.2) circle (0.3)
        (1,2) circle (0.2) circle (0.3)
        (2,2) circle (0.2) circle (0.3);
        
        \draw[dotted] (-2,-3) rectangle (4,3);
	
    \end{tikzpicture}
    \caption{$E_1$ and $E_3$ are on opposite sides of $x$ in $\mathcal{H}$.}
    \label{fig:110_other}
    \end{subfigure}
}
\caption{The 1-bend CPG representation of $V_x$ when exactly one terminal has an $\mathcal{R}$-value of $0$ (the dotted square in both figures represents the square corresponding to $x$).}
\label{fig:variable_repr_110}
\end{figure}

Next, given a clause $c$, we construct a 1-bend CPG representation of $C_c$ as follows. Denote by $E_1$, $E_2$ and $E_3$ the edges incident to $c$ in $\mathcal{H}$ (if $c$ is of degree 2 in $H$, we proceed similarly by adding an edge consisting of one grid-edge incident to $c$ in $\mathcal{H}$). Each of these edges corresponds either to a connector or to a false terminator; for $i=1,2,3$, we denote by $t_i$ the vertex in the gadget corresponding to $E_i$ with neighbors in $C_c$, and by $P_i$ their corresponding path in the representation we are constructing. We now explain how to construct the representation of $C_c \cup \{t_1,t_2,t_3\}$ depending on the $\mathcal{R}$-value of the terminal to which $t_j$ ($j=1,2,3$) is connected. Recall that, by construction, only one terminal has been assigned an $\mathcal{R}$-value of $0$.

Suppose, without loss of generality, that $t_1$ is connected to a terminal of $C_c$ of $\mathcal{R}$-value $0$. If $E_2$ and $E_3$ are on opposite sides of $c$ in $\mathcal{H}$, then we fill the corresponding square as shown in Figure \ref{fig:clause_oppose}. Otherwise, we may assume without loss of generality that $E_1$ and $E_3$ are on opposite sides of $c$ in $\mathcal{H}$ in which case we fill the corresponding square as shown in Figure \ref{fig:clause_other} (the blue dotted path in both figures corresponds to the neighbour of $t_1$ in its gadget).  Note that since the terminal to which $t_2$ (resp. $t_3$) is connected has $\mathcal{R}$-value 1, the corresponding paths have a contact point $r$ (resp. $r'$) which belongs to no other path; $P_2$ (resp. $P_3$) is then added in such a way that $r$ (resp. $r'$) is one of its endpoint.

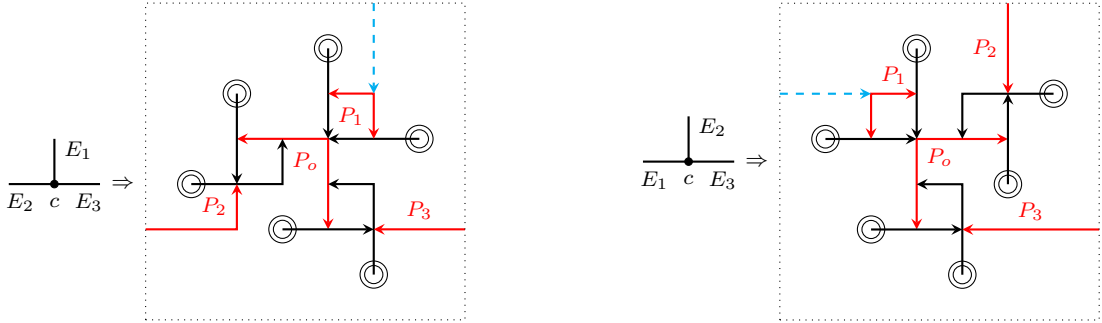
\begin{figure}[h]
\centering
\scalebox{1}
{
    \begin{subfigure}{.5\textwidth}
    \centering
    \begin{tikzpicture}[scale=.6]
    
        \node[circ,label=below:{\footnotesize $c$}] at (-4,1) {};
        \draw[-,thick] (-5,1) -- (-4,1) node[near start,below] {\footnotesize $E_2$};
        \draw[-,thick] (-3,1) -- (-4,1)  node[near start,below] {\footnotesize $E_3$};
        \draw[-,thick] (-4,2) -- (-4,1) node[near start,right] {\footnotesize $E_1$};
        
        \node[draw=none] at (-2.5,1) {\footnotesize $\Rightarrow$};
    
        \draw[red,thick,<->,>=stealth] (0,2) -- (2,2) -- (2,0) node[near start,left] {\footnotesize $P_o$};
        \draw[thick,->,>=stealth] (-1,1) -- (1,1) -- (1,2) node[near start,above] {};
        \draw[thick,->,>=stealth] (0,3) -- (0,1) node[near start,above] {};
        \draw[thick,->,>=stealth] (3,-1) -- (3,1) -- (2,1) node[near start,above] {};
        \draw[thick,->,>=stealth] (1,0) -- (3,0) node[near start,above] {};
        \draw[thick,->,>=stealth] (2,4) -- (2,2) node[near start,above] {};
        \draw[thick,->,>=stealth] (4,2) -- (2,2) node[near start,above] {};
        
        \draw[red,thick,<->,>=stealth] (2,3) -- (3,3) -- (3,2) node[midway,left] {\footnotesize $P_1$};
        \draw[red,thick,->,>=stealth] (-2,0) -- (0,0) -- (0,1) node[midway,left] {\footnotesize $P_2$};
        \draw[red,thick,->,>=stealth] (5,0) -- (3,0) node[midway,above] {\footnotesize $P_3$};
        \draw[dashed,cyan,thick,->,>=stealth] (3,5) -- (3,3) node[near start,above] {};
        
        \draw
        (-1,1) circle (0.2) circle (0.3)
        (0,3) circle (0.2) circle (0.3)
        (3,-1) circle (0.2) circle (0.3)
        (1,0) circle (0.2) circle (0.3)
        (2,4) circle (0.2) circle (0.3)
        (4,2) circle (0.2) circle (0.3);
        
        \draw[dotted] (-2,-2) rectangle (5,5);
	
    \end{tikzpicture}
    \caption{$E_2$ and $E_3$ are on opposite sides of $c$ in $\mathcal{H}$.}
    \label{fig:clause_oppose}
    \end{subfigure}
    
    \begin{subfigure}{.5\textwidth}
    \centering
    \begin{tikzpicture}[scale=.6]
    
        \node[circ,label=below:{\footnotesize $c$}] at (-5,-.5) {};
        \draw[-,thick] (-6,-.5) -- (-5,-.5) node[near start,below] {\footnotesize $E_1$};
        \draw[-,thick] (-4,-.5) -- (-5,-.5)  node[near start,below] {\footnotesize $E_3$};
        \draw[-,thick] (-5,.5) -- (-5,-.5) node[near start,right] {\footnotesize $E_2$};
        
        \node[draw=none] at (-3.5,-.5) {\footnotesize $\Rightarrow$};
    
        \draw[red,thick,<->,>=stealth] (0,-2) -- (0,0) -- (2,0) node[near start,below] {\footnotesize $P_o$};
        \draw[thick,->,>=stealth] (-1,-2) -- (1,-2) node[near start,above] {};
        \draw[thick,->,>=stealth] (1,-3) -- (1,-1) -- (0,-1) node[near start,above] {};
        \draw[thick,->,>=stealth] (2,-1) -- (2,1) node[near start,above] {};
        \draw[thick,->,>=stealth] (3,1) -- (1,1) -- (1,0) node[near start,above] {};
        \draw[thick,->,>=stealth] (0,2) -- (0,0) node[near start,above] {};
        \draw[thick,->,>=stealth] (-2,0) -- (0,0) node[near start,above] {};
        
        \draw[red,thick,<->,>=stealth] (-1,0) -- (-1,1) -- (0,1) node[midway,above] {\footnotesize $P_1$};
        \draw[red,thick,->,>=stealth] (2,3) -- (2,1) node[midway,left] {\footnotesize $P_2$};
        \draw[red,thick,->,>=stealth] (4,-2) -- (1,-2) node[midway,above] {\footnotesize $P_3$};
        \draw[dashed,cyan,thick,->,>=stealth] (-3,1) -- (-1,1) node[near start,above] {};
        
        \draw
        (-1,-2) circle (0.2) circle (0.3)
        (1,-3) circle (0.2) circle (0.3)
        (2,-1) circle (0.2) circle (0.3)
        (3,1) circle (0.2) circle (0.3)
        (0,2) circle (0.2) circle (0.3)
        (-2,0) circle (0.2) circle (0.3);
        
        \draw[dotted] (-3,-4) rectangle (4,3);
        
    \end{tikzpicture}
    \caption{$E_1$ and $E_3$ are on opposite sides of $c$ in $\mathcal{H}$.}
    \label{fig:clause_other}
    \end{subfigure}
}
\caption{The 1-bend CPG representation of $C_c$ (the dotted square in both figures represents the square corresponding to $c$).}
\label{fig:clause_repr}
\end{figure}

It now remains to complete the representation of the connectors. Given an edge $xc \in E(H)$, we construct a 1-bend CPG representation of $L_{x,c}$ as follows. By construction, either $t_{xc}$ or $t_{cx}$ is connected to a terminal in $G_i(\Phi)$ that has an $\mathcal{R}$-value of $1$. Since $x$ and $c$ have a symmetric role, we may assume without loss of generality that $t_{cx}$ is linked to a terminal of $\mathcal{R}$-value $1$, which implies that $P_{t_{cx}}$ has a free endpoint $r$. We denote by $v_1, \ldots, v_5$ the vertices in the path from $t_{cx}$ to $t_{xc}$ in $L_{x,c}$, where $v_1$ is adjacent to $t_{cx}$. Since the grid-path $L$ in $\mathcal{H}$ corresponding to the edge $xc$ has at most 4 bends, we replace each segment of $L$ with at least one $0$-bend path so as to obtain five paths. Then, we either extend or add a bend to each of these paths so that $r$ is a bend-point of $P_{v_1}$ and, for $1\leq i < 5$, one endpoint of $P_{v_i}$ is an interior point (possibly a bend-point) of $P_{v_{i+1}}$ while the other endpoint of $P_{v_i}$ is taken by an end-eating graph (see Figure \ref{fig:connector} for examples).
\end{proof}

\begin{figure}[htb]
\centering
\begin{subfigure}[b]{.45\textwidth}
\centering
\begin{tikzpicture}
\draw[thick,->,>=stealth,red] (-.25,0) -- (.5,0);
\draw[thick,dotted,red] (-.5,0) -- (-.25,0);

\node[circ,label=above:{\footnotesize $r$}] at (.5,0) {};

\draw[thick,->,>=stealth] (.5,-.25) -- (.5,0) -- (1.5,0);
\draw (.5,-.25) circle (0.05) circle (0.1);

\draw[thick,->,>=stealth] (1.5,-.25) -- (1.5,1);
\draw (1.5,-.25) circle (0.05) circle (0.1);

\draw[thick,->,>=stealth] (1.25,1) -- (2.5,1);
\draw (1.25,1) circle (0.05) circle (0.1);

\draw[thick,->,>=stealth] (2.5,.75) -- (2.5,2);
\draw (2.5,.75) circle (0.05) circle (0.1);

\draw[thick,->,>=stealth] (2.25,2) -- (3,2);
\draw (2.25,2) circle (0.05) circle (0.1);

\draw[thick,->,>=stealth] (3,1.75) -- (3,2.75);
\draw (3,1.75) circle (0.05) circle (0.1);

\draw[thick,red] (3,2.75) -- (3,3.25);
\end{tikzpicture}
\caption{The grid-path $L$ in $\mathcal{H}$ corresponding to $xc$ has 4 bends.} 
\end{subfigure}
\hspace*{.5cm}
\begin{subfigure}[b]{.45\textwidth}
\centering
\begin{tikzpicture}
\draw[thick,->,>=stealth,red] (-1.75,1.25) -- (-1,1.25);
\draw[thick,dotted,red] (-2,1.25) -- (-1.75,1.25);

\node[circ,label=above:{\footnotesize $r$}] at (-1,1.25) {};

\draw[thick,->,>=stealth] (-1,1) -- (-1,1.25) -- (0,1.25);
\draw (-1,1) circle (0.05) circle (0.1);

\draw[thick,->,>=stealth] (0,1) -- (0,2.25);
\draw (0,1) circle (0.05) circle (0.1);

\draw[thick,->,>=stealth] (-.25,2.25) -- (0,2.25) -- (0,3.25);
\draw (-.25,2.25) circle (0.05) circle (0.1);

\draw[thick,->,>=stealth] (-.25,3.25) -- (0,3.25) -- (0,4);
\draw (-.25,3.25) circle (0.05) circle (0.1);

\draw[thick,->,>=stealth] (-.25,4) -- (1.5,4);
\draw (-.25,4) circle (0.05) circle (0.1);

\draw[thick,red] (1.5,4) -- (2,4); 
\end{tikzpicture}
\caption{The grid-path $L$ in $\mathcal{H}$ corresponding to $xc$ has 2 bends.}
\end{subfigure}
\caption{A 1-bend representation of $L_{x,c}$ (the paths in red correspond to $t_{xc}$ and $t_{cx}$ in both figures).}
\label{fig:connector}
\end{figure}
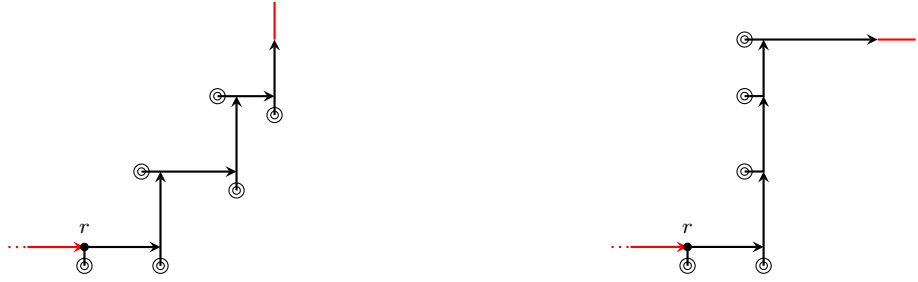 

Now, it is easy to see that the graph $G_3(\Phi)$ is planar. Hence, we have proved the following. 

\begin{theorem}
{\sc Recognition} is $\mathsf{NP}$-complete in $B_k$-CPG for any $k \geq 1$, and remains $\mathsf{NP}$-complete in planar $B_k$-CPG for any $k \geq 3$.
\end{theorem}

Observe that, by \Cref{keylemma}, the graph $G_3(\Phi)$ is CPG if and only if $\Phi$ is satisfiable. Combining this with \Cref{innp}, we immediately obtain the following. 

\begin{theorem}
{\sc Recognition} is $\mathsf{NP}$-complete in CPG and remains $\mathsf{NP}$-complete in planar CPG.
\end{theorem}

%===========================================================

\section{Complexity results for $B_0$-CPG graphs}
\label{sec:complexity}

In this section, we show the $\mathsf{NP}$-completeness of {\sc Independent Set} and {\sc Clique Cover} restricted to $B_0$-CPG graphs. To this end, we first state the following two auxiliary results.

\begin{lemma}
\label{maxdeg3}
For any $k=0,1$ and any subcubic triangle-free $B_k$-CPG graph $G$, we can find in polynomial time a $k$-bend CPG representation of $G$ such that the following hold:
\begin{itemize}
\item[(a)] Paths pairwise touch at most once;
\item[(b)] A path $P$ strictly contains one endpoint of another path if and only if the vertex corresponding to $P$ is cubic;
\item[(c)] No path touches another path at a bend-point.
\end{itemize}
\end{lemma}

\begin{proof}
Let $\mathcal{R} = (\mathcal{G},\mathcal{P})$ be a $k$-bend CPG representation of $G$. Observe first that, since $k \leq 1$, paths pairwise touch at most twice. Consider a path $P \in \mathcal{P}$ with sequence $s(P)$ given by $(x_1,y_1), (x_2,y_2), \dots, (x_{\ell_P},y_{\ell_P})$. We can check in linear time whether there exists $P'$ such that $P$ and $P'$ touch twice. Indeed, this happens if and only if either both endpoints of a path belong to the other (say, without loss of generality, both endpoints of $P$ belong to $P'$) or one endpoint of $P$ belongs to $P'$ and one endpoint of $P'$ belongs to $P$. Consider the first case i.e., $(x_1,y_1)$ and $(x_{\ell_P},y_{\ell_P})$ both belong to $P'$. Since $G$ is triangle-free, these two grid-points belong to no path other than $P$ and $P'$. Since $k \leq 1$, $P \cap P' = \{(x_1,y_1),(x_{\ell_P},y_{\ell_P})\}$ and both $P$ and $P'$ have a bend (in particular, $\ell_P = 3$). Now let $(x,y)$ be the contact point in $P$ closest to $(x_1,y_1)$ i.e., on the portion of $P$ from $(x_1,y_1)$ to $(x,y)$ there is no grid-point belonging to a path distinct from $P$. Notice that such a grid-point can be found in polynomial time by considering all endpoints of paths $P'' \neq P$ belonging to $P$ and sorting them in their order of appearance on $P$, with $(x_1,y_1)$ being the smallest such point. If $(x,y) = (x_{3},y_{3})$, we replace $s(P)$ with the sequence $(x_2,y_2), (x_3,y_3)$. Otherwise, $(x,y) \neq (x_{3},y_{3})$ and if $(x,y) \in [(x_{1},y_{1}),(x_2,y_2)]$, we replace $s(P)$ with the sequence $(x,y), (x_2,y_2), (x_{3},y_{3})$, whereas if $(x,y) \in [(x_{2},y_{2}),(x_3,y_3)]$, we replace $s(P)$ with the sequence $(x,y), (x_3,y_3)$. The case where one endpoint of $P$ belongs to $P'$ and one endpoint of $P'$ belongs to $P$ is addressed similarly. By proceeding in this way for all the paths in $\mathcal{P}$, we update the $k$-bend CPG representation of $G$ to obtain a new representation $\mathcal{R} = (\mathcal{G},\mathcal{P})$ satisfying (a).

For each $P \in \mathcal{P}$, with sequence $s(P)$ given by $(x_1,y_1), \dots, (x_{\ell_P},y_{\ell_P})$, we now check in linear time whether $(x_1,y_1)$ is a free endpoint. If it is not, we leave $s(P)$ as it is. Otherwise, we find in polynomial time the contact point $(x, y)$ in $P$ closest to $(x_1,y_1)$. If $(x,y) = (x_{\ell_P},y_{\ell_P})$, the vertex corresponding to $P$ has degree one and does not strictly contain a grid-point which is an endpoint of some other path and so we leave the sequence $s(P)$ as it is. Otherwise, $(x,y) \notin \{(x_{\ell_P},y_{\ell_P}), (x_1, y_{1})\}$, and we replace $s(P)$ with the sequence $(x,y),(x_{i},y_{i}), (x_{\ell_P},y_{\ell_P})$, where $i \in \{2,3\}$ is such that $(x, y) \in [(x_{i-1},y_{i-1}),(x_i,y_i)]$. For each $P \in \mathcal{P}$, with sequence $s(P)$ given by $(x_1,y_1), \dots, (x_{\ell_P},y_{\ell_P})$, we finally check in linear time whether $(x_{\ell_P},y_{\ell_P})$ is a free endpoint and proceed with the shortenings as above. In this way, we obtain a $k$-bend CPG representation $\mathcal{R}$ satisfying (b). As the only operations done are shortenings of paths, it still satisfies (a).

Given a path $P \in \mathcal{P}$, we can clearly check in linear time whether a path touches $P$ at a bend-point. Suppose that there exists a path $P'$ touching $P$ at a bend-point $(x,y)$ and suppose, without loss of generality, that $P$ lies above and on the right of $(x,y)$ and that $P'$ uses the grid-edge below $(x,y)$ as shown in \Cref{fig:bendtouch} (the other cases are symmetric). 

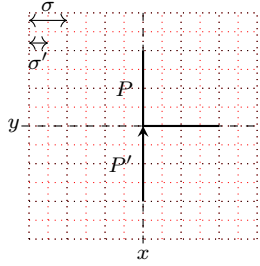
\begin{figure}[htb]
\centering
\begin{tikzpicture}

\draw[step=.25,red,thin,dotted] (-1.5,-1.5) grid (1.5,1.5);

\draw[step=.5,black,thin,dotted] (-1.5,-1.5) grid (1.5,1.5);

\draw[-,thick] (0,1) -- (0,0) node[midway,left] {\footnotesize $P$} -- (1,0);

\draw[->,>=stealth,thick] (0,-1) -- (0,0) node[midway,left] {\footnotesize $P'$};

\draw[<->,thin] (-1.5,1.4) -- (-1,1.4) node[midway,above] {\footnotesize $\sigma$};

\draw[<->,thin] (-1.5,1.1) -- (-1.25,1.1) node[midway,below] {\footnotesize $\sigma'$};

\node[draw=none] at (-1.7,0) {\footnotesize $y$};

\node[draw=none] at (0,-1.7) {\footnotesize $x$};

\draw[dashed] (-1.6,0) -- (1.5,0);

\draw[dashed] (0,1.5) -- (0,-1.6);

\end{tikzpicture}

\caption{$P'$ touches $P$ at one of its bend-point $(x,y)$.}

\label{fig:bendtouch}
\end{figure}

Since $G$ is triangle-free, the grid-point $(x,y)$ belongs to no path other than $P$ and $P'$. We first refine the grid once by setting the grid-step to $\sigma' = \sigma/2$ so as to ensure that every subsequent transformation does not introduce unwanted contacts. Then, for any path having either one endpoint $(x,y')$ on column $x$ with $y' \leq y$, or a segment $[(x,y_1),(x,y_2)]$ on column $x$ with $y_1 < y_2 \leq y$, we do the following. In the first case, we replace $(x,y')$ with $(x + \sigma',y')$ in the sequence describing the considered path, whereas in the second case, we replace $(x,y_1)$ with $(x+\sigma',y_1)$ and $(x,y_2)$ with $(x + \sigma',y_2)$. It is easy to see that repeating these operations for each path $P$, we obtain a $k$-bend CPG representation satisfying (a), (b) and (c). 
\end{proof}

\begin{lemma}[Folklore]
\label{npcomp}
If {\sc Independent Set} is $\mathsf{NP}$-complete for a graph class $\mathcal{G}$, then it is $\mathsf{NP}$-complete for 2-subdivisions of graphs in $\mathcal{G}$.
\end{lemma}

\begin{theorem}
\label{thm:is}
{\sc Independent Set} is $\mathsf{NP}$-complete for triangle-free subcubic $B_0$-CPG graphs.
\end{theorem}

\begin{proof}
We exhibit a polynomial reduction from {\sc Independent Set} restricted to 2-subdivisions of cubic planar graphs. Since {\sc Independent Set} is $\mathsf{NP}$-complete for cubic planar graphs \cite{mohar}, Lemma \ref{npcomp} implies that it remains $\mathsf{NP}$-complete for the considered class.

Given a 2-subdivision $G=(V,E)$ of a cubic planar graph, we construct a $B_0$-CPG graph $G'$ as follows. Since $G$ is planar and triangle-free, it follows from \cite{chaplick12} that we can obtain a $1$-bend CPG representation of $G$ in linear time. By \Cref{maxdeg3}, we can further obtain in polynomial time a $1$-bend CPG representation $\mathcal{R} = (\mathcal{G},\mathcal{P})$ of $G$ in which paths pairwise touch at most once, a path strictly contains one endpoint of another path if and only if its corresponding vertex is cubic and no path touches another path at its bend-point. We then refine the grid four times, so that each segment of a path contains at least $16$ grid-edges and let $\sigma$ be the new grid-step.

For a path $P \in \mathcal{P}$, let $\mathcal{P}_P = \{ P' \in \mathcal{P} : P \cap P' \neq \varnothing\}$ be the subset of paths touching $P$. If $P$ has no bend, consider the middle grid-point $p$ of $P$ with coordinates $(x_p,y_p)$, where $x_p=x_1 + (x_2-x_1)/2$ and $y_p = y_1$, if $P$ is horizontal with $x_1 < x_2$, and $x_p=x_1$ and $y_p=y_1 + (y_2-y_1)/2$, if $P$ is vertical with $y_1 < y_2$. If the grid-point $p$ belongs to a path other than $P$, we set $p$ to be the grid-point with coordinates $(x_p+\sigma,y_p)$, if $P$ is horizontal, and $(x_p,y_p+\sigma)$, if $P$ is vertical. Now $p$ naturally divides $P$ into two line segments $S_1=[(x_1,y_1),(x_p,y_p)]$ and $S_2=[(x_p,y_p),(x_2,y_2)]$ partitioning $\mathcal{P}_P$ into $\mathcal{P}_P^i =\{P' \in \mathcal{P}_P : P' \cap S_i \neq \varnothing\}$, for $i = 1,2$ (see \Cref{fig:nbp}). Note that both $S_1$ and $S_2$ contain at least $4$ grid-edges. If $P$ has a bend, we let $\mathcal{P}_P^1$ and $\mathcal{P}_P^2$ be the subsets of paths touching the horizontal segment $S_{1}$ and the vertical segment $S_{2}$ of $P$, respectively (see \Cref{fig:bp}). $(\mathcal{P}_P^1,\mathcal{P}_P^2)$ is a partition of $\mathcal{P}_P$ in this case as well.

\begin{figure}[htb]

\centering

\begin{subfigure}[b]{.45\textwidth}

\centering

\begin{tikzpicture}[scale=.6]

\draw[<->,thick,>=stealth] (-3.1,0) -- (2.1,0) node[above right] {\footnotesize $P$};

\node[circ,label=below:{\footnotesize $p$}] (p) at (-.6,0) {};

\draw[thick,red,>=stealth] (-3,0) -- (p) node[midway,below] {\footnotesize $S_1$};

\draw[thick,blue,>=stealth] (p) -- (2,0) node[midway,below] {\footnotesize $S_2$};

\draw[->,thick,>=stealth] (-3.1,.5) -- (-3.1,0);

\draw[thick,dotted] (-3.1,.75) -- (-3.1,.5);

\draw[->,thick,>=stealth] (-1,.5) -- (-1,0);

\draw[thick,dotted] (-1,.75) -- (-1,.5);

\draw[->,thick,>=stealth] (2.1,-.5) -- (2.1,0);

\draw[thick,dotted] (2.1,-.75) -- (2.1,-.5);

\end{tikzpicture}

\caption{$P$ has no bend.}

\label{fig:nbp}

\end{subfigure}
\begin{subfigure}[b]{.45\textwidth}

\centering

\begin{tikzpicture}[scale=.6]

\draw[<->,thick,>=stealth] (3.1,0) -- (0,0) -- (0,2.1) node[above right] {\footnotesize $P$};

\draw[thick,blue,>=stealth] (0,2) -- (0,0) node[midway,left] {\footnotesize $S_2$};

\draw[thick,red,>=stealth] (0,0) -- (3,0) node[midway,above] {\footnotesize $S_1$};

\draw[->,thick,>=stealth] (.7,.5) -- (.7,0);

\draw[thick,dotted] (.7,.75) -- (.7,0);

\draw[->,thick,>=stealth] (-.5,2.1) -- (0,2.1);

\draw[thick,dotted] (-.75,2.1) -- (-.5,2.1);

\draw[<-,thick,>=stealth] (3.1,0) -- (3.1,.5);

\draw[thick,dotted] (3.1,.5) -- (3.1,.75);

\end{tikzpicture}

\caption{$P$ has a bend.}

\label{fig:bp}

\end{subfigure}

\caption{Partitioning the set of paths $\mathcal{P}_P$ touching $P$ into the subset of paths $\mathcal{P}_P^1$ touching $S_1$ and the subset of paths $\mathcal{P}_P^2$ touching $S_2$.}

\end{figure}
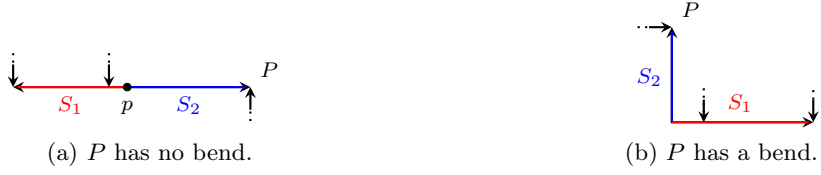

We now subdivide each path $P$ into five 0-bend paths $P_1$, $P_2$, $P^1$, $P^2$, $P^3$, in such a way that $P^1$ (resp. $P^2$; $P^3$) touches only $P_1$ and $P^2$ (resp. $P^1$ and $P^3$; $P^2$ and $P_2$) as depicted in \Cref{fig:sp}. More precisely, if $P$ has no bend, suppose without loss of generality that $P$ is horizontal and $\mathcal{P}_P^2 = \varnothing$; then, $P_1 = [(x_1,y_1),(x_p,y_p)]$, $P^j = [(x_p+(j-1)\sigma,y_p),(x_p+j\sigma,y_p)]$ for $j=1,2,3$, and $P_2 = [(x_p+3\sigma,y_p),(x_2,y_2)]$. On the other hand, if $P$ has a bend, suppose without loss of generality that $\mathcal{P}_P^2 = \varnothing$; then, $P_1=[(x_1,y_1),(x_2,y_2)]$, $P^j=[(x_2,y_2+(j-1)\sigma),(x_2,y_2+j\sigma)]$ for $j=1,2,3$ assuming that $y_2 < y_3$ (otherwise, subtract the multiples of $\sigma$), and $P_2 = [(x_2,y_2+3\sigma),(x_3,y_3)]$.

\begin{figure}[htb]

\centering

\begin{subfigure}[b]{.45\textwidth}

\centering

\begin{tikzpicture}[scale=.7]

\draw[<->,thick,>=stealth] (-2.5,0) -- (0,0) node[midway,below] {\footnotesize $P_1$};

\draw[<->,thick,>=stealth] (1.5,0) -- (2.5,0) node[midway,below] {\footnotesize $P_2$};

\draw[<->,thick,>=stealth] (0,0) -- (.5,0) node[midway,below] {\footnotesize $P^1$};

\draw[<->,thick,>=stealth] (.5,0) -- (1,0) node[midway,below] {\footnotesize $P^2$};

\draw[<->,thick,>=stealth] (1,0) -- (1.5,0) node[midway,below] {\footnotesize $P^3$};

\draw[->,thick,>=stealth] (-2.5,.5) -- (-2.5,0);

\draw[thick,dotted] (-2.5,.75) -- (-2.5,.5);

\draw[->,thick,>=stealth] (-.5,.5) -- (-.5,0);

\draw[thick,dotted] (-.5,.75) -- (-.5,.5);

\draw[->,thick,>=stealth] (2.5,-.5) -- (2.5,0);

\draw[thick,dotted] (2.5,-.75) -- (2.5,-.5);

\end{tikzpicture}

\caption{$P$ has no bend.}

\end{subfigure}
\begin{subfigure}[b]{.45\textwidth}

\centering

\begin{tikzpicture}[scale=.7]

\draw[<->,thick,>=stealth] (0,2.5) -- (0,1.5) node[midway,left] {\footnotesize $P_2$};

\draw[<->,thick,>=stealth] (0,0) -- (2.5,0) node[midway,above] {\footnotesize $P_1$};

\draw[<->,thick,>=stealth] (0,0) -- (0,.5) node[midway,left] {\footnotesize $P^1$};

\draw[<->,thick,>=stealth] (0,.5) -- (0,1) node[midway,left] {\footnotesize $P^2$};

\draw[<->,thick,>=stealth] (0,1) -- (0,1.5) node[midway,left] {\footnotesize $P^3$};

\draw[->,thick,>=stealth] (.5,.5) -- (.5,0);

\draw[thick,dotted] (.5,.75) -- (.5,0);

\draw[->,thick,>=stealth] (-.5,2.5) -- (0,2.5);

\draw[thick,dotted] (-.75,2.5) -- (0,2.5);

\draw[<-,thick,>=stealth] (2.5,0) -- (2.5,.5);

\draw[thick,dotted] (2.5,.5) -- (2.5,.75);

\end{tikzpicture}

\caption{$P$ has a bend.}

\end{subfigure}

\caption{Subdividing a path $P$ into five $0$-bend paths such that, for any $j=1,2,3$, $P^j$ touches no path outside of $\{P^1,P^2,P^3,P_1,P_2\}$ and paths touching $P$ touch either $P_1$ or $P_2$.}

\label{fig:sp}

\end{figure}
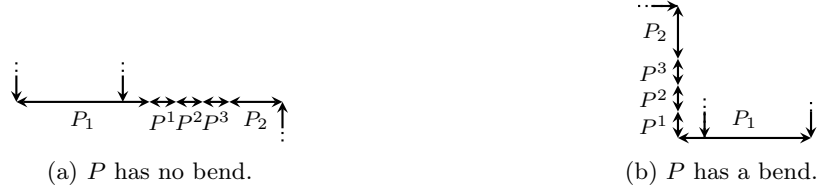

Let $G'$ be the graph corresponding to the resulting 0-bend representation. Clearly, $G'$ is a graph obtained from $G$ by replacing every vertex $u \in V$ with a path $u_1u^1u^2u^3u_2$, where $u_i$ (resp. $u^i$) corresponds to $P_i$ (resp. $P^i$), and apportioning the neighborhood of $u$ among $u_1$ and $u_2$. We now show that $\alpha (G') = \alpha (G) + 2|V|$, which would conclude the proof.

Given a maximum independent set $S$ of $G$, we construct an independent set $S'$ of $G'$ as follows. If $u \in S$, then add $u_1$, $u_2$ and $u^2$ to $S'$. Otherwise, add $u^1$ and $u^3$ to $S'$. Clearly, $S'$ is an independent set of $G'$ and $\alpha (G') \geq |S'| = |S| + 2|V| = \alpha(G) + 2|V|$. 

Conversely, given a maximum independent set $S'$ of $G'$, we construct an independent set $S$ of $G$ as follows. Observe that by maximality of $S'$, for any vertex $u \in V$, at least two vertices of $\{u_1,u_2,u^1,u^2,u^3\}$ are in $S'$. Indeed, if for some $u \in V$ we have that $S' \cap \{u_1,u_2,u^1,u^2,u^3\}$ contains at most one vertex $x$, then $(S'\backslash \{x\}) \cup \{u^1,u^3\}$ is a strictly larger independent set of $G'$. Furthermore, we may assume that for any vertex $u \in V$, either both $u_1$ and $u_2$ are in $S'$ or none of them is. Indeed, if $u_i$ is the only such vertex in $S'$, it suffices to consider the maximum independent set $(S' \backslash \{u_i,x\}) \cup \{u^1,u^3\}$, for the unique $x \in S' \cap \{u^1,u^2,u^3\}$. Also note that if both $u_1$ and $u_2$ are in $S'$, then $u^2$ is also in $S'$. We now add $u \in V$ to $S$ if and only if both $u_1$ and $u_2$ are in $S'$. Clearly, $S$ is an independent set of $G$ and so $\alpha (G) \geq |S| = |S'| - 2|V| = \alpha(G') - 2|V|$.
\end{proof}

\textit{Remark.} Since any triangle-free $B_0$-CPG graph is planar \cite{cpg}, Theorem \ref{thm:is} implies that {\sc Independent Set} is $\mathsf{NP}$-complete for planar $B_0$-CPG graphs.\\

In the following, a \textit{snake} is a $0$-bend CPG representation of a path.

\begin{theorem}
\label{thm:cc}
{\sc Clique Cover} is $\mathsf{NP}$-complete for $B_0$-CPG line graphs.
\end{theorem}

\begin{proof}
We exhibit a polynomial reduction from {\sc Vertex Cover} restricted to 2-subdivisions of triangle-free subcubic $B_0$-CPG graphs, which is $\mathsf{NP}$-complete by Theorem \ref{thm:is} and Lemma \ref{npcomp} (recall that $S$ is an independent set of a graph $G$ if and only if $V(G) \setminus S$ is a vertex cover of $G$). Given a 2-subdivision $G'$ of a triangle-free subcubic $B_0$-CPG graph $G$, we show that its line graph $L(G')$ is $B_0$-CPG. Since for any triangle-free graph $H$ we have $\beta(H) = \theta(L(H))$ (see, e.g., \cite{munaro}), this would conclude the proof. 

Consider a $0$-bend CPG representation $\mathcal{R} = (\mathcal{G},\mathcal{P}) $ of $G$ in which a path strictly contains one endpoint of another path if and only if the corresponding vertex is cubic (see \Cref{maxdeg3}). We first show how to construct a $0$-bend CPG representation $\mathcal{R}'$ of the $2$-subdivision $G'$ of $G$. 

We start by refining the grid twice and let $\sigma$ be the new grid-step. In this way, any path corresponding to a vertex of degree at most $2$ has length at least $4$ and, for any path $P$ corresponding to a cubic vertex, both segments $[(x_1,y_1),(x,y)]$ and $[(x,y),(x_2,y_2)]$, where $(x,y)$ is the contact point contained in $P$, contain at least $4$ grid-edges. Now, for every contact point $p$ of $\mathcal{R}$, arbitrarily choose one path $P$ in $\mathcal{P}$ having $p$ as an endpoint, shorten $P$ so that $p$ is no longer an endpoint of $P$ and add two paths $P_1$ and $P_2$, with $P_1$ having $p$ as endpoint and touching $P_2$, and $P_2$ touching $P$. More precisely, suppose that $P$ is horizontal and that $p = (x_{2}, y_{2})$ is the right endpoint of $P$, the other being $(x_{1}, y_{1})$ (the other cases are symmetric). We set $s(P)=(x_1,y_1),(x_2-2\sigma,y_2)$ and add two paths $P_1$ and $P_2$ such that $s(P_j) = (x_2-j\sigma,y_2),(x_2-(j-1)\sigma,y_2)$, for $j=1,2$ (see \Cref{2sub} for some examples).

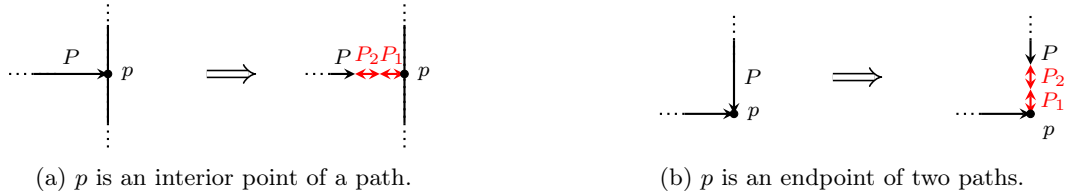
\begin{figure}[htb]

\begin{minipage}[b]{.45\textwidth}

\centering

\begin{subfigure}[b]{\linewidth}

\centering

\begin{tikzpicture}[scale=.65]

\draw[thick] (0,-1) -- (0,1);

\draw[thick,dotted] (0,-1.5) -- (0,1.5);

\node[circ,label=right:{\footnotesize $p$}] at (0,0) {};

\draw[->,thick,>=stealth] (-1.5,0) -- (0,0) node[midway,above] {\footnotesize $P$};

\draw[thick,dotted] (-2,0) -- (-1,0);

\draw[-Implies,line width=.6pt,double distance=2pt] (2,0) -- (3,0);

\draw[thick,dotted] (4,0) -- (4.5,0);

\draw[->,thick,>=stealth] (4.5,0) -- (5,0) node[midway,above] {\footnotesize $P$};

\draw[<->,thick,>=stealth,red] (5,0) -- (5.5,0) node[midway,above] {\footnotesize $P_2$};

\draw[<->,thick,>=stealth,red] (5.5,0) -- (6,0) node[midway,above] {\footnotesize $P_1$};

\draw[thick] (6,-1) -- (6,1);

\draw[thick,dotted] (6,-1.5) -- (6,1.5);

\node[circ,label=right:{\footnotesize $p$}] at (6,0) {};

\end{tikzpicture}

\caption{$p$ is an interior point of a path.}

\end{subfigure}
\end{minipage}
\hspace*{.5cm}
\begin{minipage}[b]{.45\textwidth}
\centering
\begin{subfigure}[b]{\linewidth}

\centering

\begin{tikzpicture}[scale=.65]

\draw[->,thick,>=stealth] (-1,0) -- (0,0);

\draw[thick,dotted] (-1.5,0) -- (-1,0);

\draw[->,thick,>=stealth] (0,1.5) -- (0,0) node[midway,right] {\footnotesize $P$};

\draw[thick,dotted] (0,2) -- (0,1);

\node[circ,label=right:{\footnotesize $p$}] at (0,0) {};

\node[invisible] at (-2.5,0) {};

\draw[-Implies,line width=.6pt,double distance=2pt] (2,.75) -- (3,.75);

\draw[->,thick,>=stealth] (5,0) -- (6,0);

\draw[thick,dotted] (4.5,0) -- (5,0);

\draw[<->,thick,>=stealth,red] (6,0) -- (6,.5) node[midway,right] {\footnotesize $P_1$};

\draw[<->,thick,>=stealth,red] (6,.5) -- (6,1) node[midway,right] {\footnotesize $P_2$};

\draw[<-,thick,>=stealth] (6,1) -- (6,1.5) node[midway,right] {\footnotesize $P$};

\draw[thick,dotted] (6,1.5) -- (6,2);

\node[circ,label=below right:{\footnotesize $p$}] at (6,0) {};

\end{tikzpicture}

\caption{$p$ is an endpoint of two paths.}

\end{subfigure}

\end{minipage}

\caption{2-subdividing the edge $e$ of $G$ represented by $p$.}

\label{2sub}

\end{figure}

Observe now that the $0$-bend representation $\mathcal{R}'$ of $G'$ thus obtained satisfies the following property: every maximal snake in the representation obtained from $\mathcal{R}'$ by removing every path whose corresponding vertex is cubic contains at least two paths lying either on the same row or on the same column. We now derive from $\mathcal{R}'$ a $0$-bend CPG representation of $L(G')$ as follows. We first refine the grid once so that the previously introduced paths have length at least two and let $\sigma$ be the new grid-step. Consider a path $P$ strictly containing an endpoint $p=(x,y)$ of another path $P'$ with $s(P) = (x_{1}, y_{1}), (x_{2}, y_{2})$ and $s(P') = (x'_1,y'_1)(x'_2,y'_2)$. Suppose, without loss of generality, that $P$ is vertical with $y_1 < y_2$ and $P'$ is horizontal with $x'_1 < x'_2$, and that $p$ is the left endpoint of $P'$, that is, $(x,y) = (x'_1,y'_1)$ (the other cases are treated similarly). We then split $P$ into two paths $P_1 = [(x_1,y_1)(x,y)]$ and $P_2=[(x,y),(x_2,y_2)]$ so that $p$ is an endpoint of these two paths, shorten $P'$ so that $p$ is no longer an endpoint of it by updating its sequence to $s(P') = (x + \sigma,y),(x'_2,y'_2)$ and add a path $P_3=[(x,y),(x+\sigma,y)]$ touching $P'$ and having $p$ as an endpoint (see \Cref{splitp}). We refer to this operation as a triangle implant. Note that since no two cubic vertices of $G'$ are adjacent (as $G'$ is the 2-subdivision of $G$), no endpoint of $P$ is strictly contained in another path.

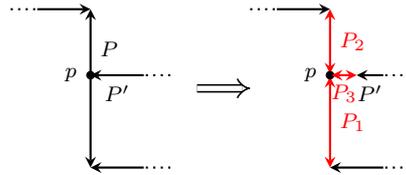
\begin{figure}[htb]

\centering

\begin{tikzpicture}[scale=.7]

\draw[<->,thick,>=stealth] (0,-1.5) -- (0,1.5) node[near end,right] {\footnotesize $P$};

\node[circ,label=left:{\footnotesize $p$}] at (0,.25) {};

\draw[->,thick,>=stealth] (1,-1.5) -- (0,-1.5);

\draw[thick,dotted] (1.5,-1.5) -- (1,-1.5);

\draw[->,thick,>=stealth] (1,.25) -- (0,.25) node[midway,below] {\footnotesize $P'$};

\draw[thick,dotted] (1.5,.25) -- (1,.25);

\draw[->,thick,>=stealth] (-1,1.5) -- (0,1.5);

\draw[thick,dotted] (-1.5,1.5) -- (-1,1.5);

\draw[-Implies,line width=.6pt,double distance=2pt] (2,0) -- (3,0);

\draw[<->,thick,>=stealth,red] (4.5,-1.5) -- (4.5,.22) node[midway,right] {\footnotesize $P_1$};

\draw[<->,thick,>=stealth,red] (4.5,.28) -- (4.5,1.5) node[midway,right] {\footnotesize $P_2$};

\node[circ,label=left:{\footnotesize $p$}] at (4.5,.25) {};

\draw[->,thick,>=stealth] (5.5,-1.5) -- (4.5,-1.5);

\draw[thick,dotted] (6,-1.5) -- (5.5,-1.5);

\draw[<->,thick,>=stealth,red] (5,.25) -- (4.53,.25) node[midway,below] {\footnotesize $P_3$};

\draw[<-,thick,>=stealth] (5,.25) -- (5.5,.25) node[midway,below] {\footnotesize $P'$};

\draw[thick,dotted] (6,.25) -- (5.5,.25);

\draw[->,thick,>=stealth] (3.5,1.5) -- (4.5,1.5);

\draw[thick,dotted] (3,1.5) -- (3.5,1.5);

\end{tikzpicture}

\caption{A triangle implant.}

\label{splitp}

\end{figure}

Next, we remove all the paths introduced by the triangle implants. The resulting representation is then a disjoint union of maximal snakes, each containing at least two paths lying either on the same row or on the same column. For every maximal snake, consider any two such paths and merge them into one path. By then reintroducing all the removed paths from the triangle implants, it is easy to see that we obtain a $0$-bend CPG representation of $L(G')$.
\end{proof}

%===========================================================

\section{$3$-coloring $B_1$-EPG graphs}
\label{sec:b1epg}

It was shown in \cite{cpg} that {\sc 3-Colorability} is $\mathsf{NP}$-complete in $B_0$-CPG. Since $B_0$-CPG is a subclass of $B_2$-EPG, it follows that this problem is $\mathsf{NP}$-complete in $B_2$-EPG. On the other hand, {\sc 3-Colorability} is polynomial-time solvable in $B_0$-EPG as this class coincides with that of interval graphs. We now settle the open case by showing that {\sc 3-Colorability} is $\mathsf{NP}$-complete in $B_1$-EPG.

\begin{theorem}
\label{thm:3color}
{\sc 3-Colorability} is $\mathsf{NP}$-complete for planar $B_1$-EPG graphs.
\end{theorem}

\begin{proof}
We exhibit a polynomial reduction from {\sc 3-Colorability} restricted to planar graphs of maximum degree 4, which was shown to be $\mathsf{NP}$-complete in \cite{garey}. Given a planar graph $G=(V,E)$ of maximum degree 4, we construct a planar $B_1$-EPG graph $G'$ such that $G'$ is 3-colorable if and only if $G$ is 3-colorable. By \citep{tamassia}, we can find in linear time an embedding of $G$ on a grid with area $O(|V|^2)$, such that vertices are mapped to grid-points and edges are mapped to pairwise interiorly disjoint grid-paths with at most 4 bends connecting the two grid-points corresponding to the endvertices. Let $\mathcal{E} = (\mathcal{V}, \mathcal{P})$ be such an embedding of $G$, where $\mathcal{V}$ is the set of grid-points in one-to-one correspondence with $V$ and $\mathcal{P}$ is the set of grid-paths in one-to-one correspondence with $E$. 

For any vertex $u \in V$, we denote by $(x_u,y_u)$ the grid-point in $\mathcal{V}$ corresponding to $u$. For any edge $uv \in E$, we denote by $P_{uv}$ the path in $\mathcal{P}$ corresponding to $uv$ and described by the sequence $s(P_{uv}) = (x_u,y_u),(x_v,y_v)$, if $P_{uv}$ has no bend, or by the sequence $s(P_{uv}) = (x_uy_u),(x_1,y_1), \dots, (x_k,y_k), (x_v,y_v)$ where, for each $i =1,\ldots ,k$, $(x_i,y_i)$ is a bend-point of $P_{uv}$ (note that $k \leq 4$). The graph $G'$ is then obtained as follows. For any edge $uv \in E$, if the path $P_{uv} \in \mathcal{P}$ contains $k$ bends (for some $k = 0, 1, 2 , 3, 4$), we replace the edge $uv$ with a sequence $d_{uv}$ of $k+1$ diamonds by identifying $u$ with the vertex of degree 2 in the first diamond of the sequence and connecting $v$ to the vertex of degree 2 in the last diamond of the sequence (see \Cref{diamonds1} where $P_{uv}$ has 3 bends and $P_{u'v'}$ has no bend). Clearly, $G'$ is planar. We then construct from $\mathcal{E}$ a 1-bend EPG representation of~$G'$ as follows.

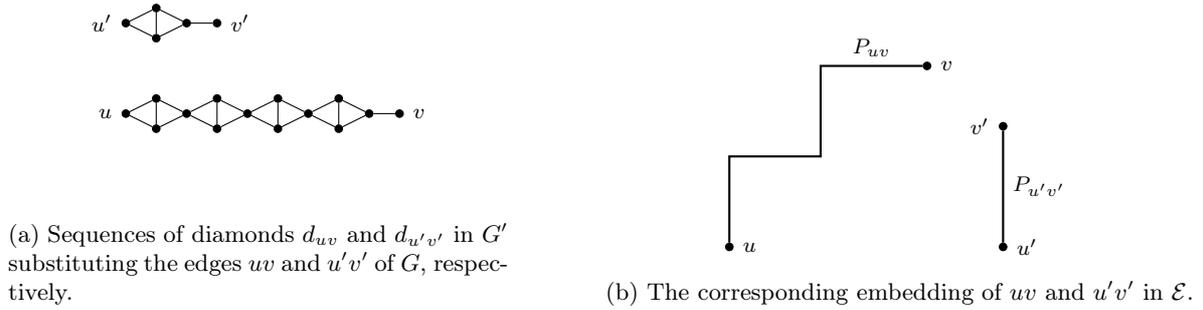
\begin{figure}[htb]

\centering

\begin{subfigure}[b]{.4\textwidth}

\centering

\begin{tikzpicture}[scale=.4]

\node[circ,label=left:{\footnotesize $u'$}] (u') at (0,3) {};

\node[circ] (1b) at (1,3.5) {};

\node[circ] (2b) at (1,2.5) {};

\node[circ] (3b) at (2,3) {};

\node[circ,label=right:{\footnotesize $v'$}] (v') at (3,3) {};

\draw (u') -- (1b)

(u') -- (2b)

(1b) -- (2b)

(1b) -- (3b)

(2b) -- (3b)

(3b) -- (v');

\node[circ,label=left:{\footnotesize $u$}] (u) at (0,0) {};

\node[circ] (1) at (1,.5) {};

\node[circ] (2) at (1,-.5) {};

\node[circ] (3) at (2,0) {};

\node[circ] (4) at (3,.5) {};

\node[circ] (5) at (3,-.5) {};

\node[circ] (6) at (4,0) {};

\node[circ] (a) at (5,.5) {};

\node[circ] (b) at (5,-.5) {};

\node[circ] (7) at (6,0) {};

\node[circ,label=right:{\footnotesize $v$}] (v) at (9,0) {};

\node[invisible] at (0,-3) {};

\node[circ] (8) at (7,.5) {};

\node[circ] (9) at (7,-.5) {};

\node[circ] (10) at (8,0) {};

\draw (u) -- (1)

(u) -- (2)

(1) -- (2)

(1) -- (3)

(2) -- (3)

(3) -- (4)

(3) -- (5)

(4) -- (5)

(4) -- (6)

(5) -- (6)

(6) -- (a)

(6) -- (b)

(a) -- (b)

(a) -- (7)

(b) -- (7)

(7) -- (8)

(7) -- (9)

(8) -- (9)

(8) -- (10)

(9) -- (10)

(10) -- (v);

\end{tikzpicture}

\caption{Sequences of diamonds $d_{uv}$ and $d_{u'v'}$ in~$G'$ substituting the edges $uv$ and $u'v'$ of $G$, respectively.}

\label{diamonds1}

\end{subfigure}
\hspace*{.8cm}
\begin{subfigure}[b]{.5\textwidth}

\centering

\begin{tikzpicture}[scale=.4]

\node[circ,label=right:{\footnotesize $u$}] (u) at (0,1) {};

\node[circ,label=right:{\footnotesize $v$}] (v) at (6.5,7) {};

\draw[thick] (u) -- (0,4) -- (3,4) -- (3,7) -- (v) node[midway,above] {\footnotesize $P_{uv}$};

\node[circ,label=right:{\footnotesize $u'$}] (u') at (9,1) {};

\node[circ,label=left:{\footnotesize $v'$}] (v') at (9,5) {};

\draw[thick] (u') -- (v') node[midway,right] {\footnotesize $P_{u'v'}$};

\end{tikzpicture}

\caption{The corresponding embedding of $uv$ and $u'v'$ in $\mathcal{E}$.}

\label{diamonds2}

\end{subfigure}

\caption{Replacing an edge in $G$ with a sequence of diamonds.}

\end{figure}

We first refine the grid $3$ times and let $\sigma$ be the new grid-step. Clearly, each segment of a path in $\mathcal{P}$ contains now at least $8$ grid-edges and any two grid-points in $\mathcal{V}$ lying on the same column are separated by at least $8$ grid-edges. We then associate with each vertex $u \in V$ a vertical path $P_u$ containing the grid-point $(x_u,y_u)$ as follows. If the grid-edge above (resp. below) $(x_u,y_u)$ is not used by any path in $\mathcal{P}$, then the top (resp. bottom) endpoint of $P_u$ is $(x_u,y_u+\sigma)$ (resp. $(x_u,y_u-\sigma)$); otherwise, the top (resp. bottom) endpoint of $P_u$ is $(x_u,y_u+2\sigma)$ (resp. $(x_u,y_u-2\sigma)$) (see \Cref{fig:Pu}). Then, for any $v \in V$ such that $uv \in E$, we will construct the paths corresponding to vertices in $d_{uv}$ so that the following hold.

\begin{itemize}

\item If $P_{uv}$ uses the grid-edge above $(x_u,y_u)$ then the paths(s) corresponding to the neighbor(s) of $u$ in $d_{uv}$ will intersect $P_u$ on the grid-edge $[(x_u,y_u+\sigma)(x_u,y_u+2\sigma)]$.

\item If $P_{uv}$ uses the grid-edge below $(x_u,y_u)$ then the paths(s) corresponding to the neighbor(s) of $u$ in $d_{uv}$ will intersect $P_u$ on the grid-edge $[(x_u,y_u-\sigma)(x_u,y_u-2\sigma)]$.

\item If $P_{uv}$ uses the grid-edge to the left of $(x_u,y_u)$ then the paths(s) corresponding to the neighbor(s) of $u$ in $d_{uv}$ will intersect $P_u$ on the grid-edge $[(x_u,y_u)(x_u,y_u+\sigma)]$.

\item If $P_{uv}$ uses the grid-edge to the right of $(x_u,y_u)$ then the paths(s) corresponding to the neighbor(s) of $u$ in $d_{uv}$ will intersect $P_u$ on the grid-edge $[(x_u,y_u)(x_u,y_u-\sigma)]$.

\end{itemize}

\begin{figure}[htb]

\centering

\begin{tikzpicture}[scale=.5]

\draw[dotted] (-2,-2) grid (2,5);

\node[circ,label=above left:{\footnotesize $(x_u,y_u)$}] (u) at (0,0) {};

\draw[-,thick] (u) -- (0,3); 
\draw[thick,dashed] (0,3) -- (0,4);
\draw[-,thick] (0,4) -- (0,5) -- (1,5);

\draw[thick,dashed] (1,5) -- (2,5);

\draw[-,thick] (u) -- (1,0);

\draw[thick,dashed] (1,0) -- (1.5,0);

\draw[-,thick] (u) -- (-1,0);

\draw[thick,dashed] (-1.5,0) -- (-1,0);

\draw[thick,blue,<->,>=stealth] (0,2) -- (0,-1) node[pos=.3,right] {\footnotesize $P_u$};

\end{tikzpicture}

\caption{Constructing the path $P_u$ (in blue) corresponding to the vertex $u$ in the case the grid-edge above $(x_u,y_u)$ is used by some path in $\mathcal{P}$ and the grid-edge below $(x_{u}, y_{u})$ is used by no path in $\mathcal{P}$.}

\label{fig:Pu}

\end{figure}
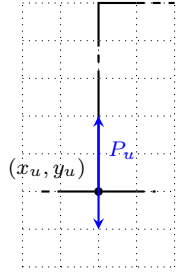

We now explain how the construction satisfying the properties above is done. Consider an edge $uv \in E$ and suppose that $u$ has been identified with a vertex of degree two in $d_{uv}$ and $v$ is adjacent to the other vertex of degree two in $d_{uv}$. Assume first that $P_{uv}$ has no bend. If $P_{uv}$ is horizontal, say $P_{uv}$ lies to the right of $(x_u,y_u)$ (the other case is symmetric), then we add the following three paths corresponding to the other vertices of the diamond: $s(P_1) = s(P_2) = (x_u,y_u-\sigma), (x_u,y_u), (x_u + 2\sigma, y_{u})$ and $s(P_3) = (x_u + \sigma, y_{u}), (x_v,y_v), (x_{v}, y_v +\sigma)$ (see \Cref{fig:nobendhori}). Otherwise $P_{uv}$ is vertical, say $P_{uv}$ lies above $(x_u,y_u)$ (the other case is symmetric), in which case we add the following three paths corresponding to the other vertices of the diamond: $s(P_1) = s(P_2) = (x_u, y_u + \sigma), (x_u,y_u + 3\sigma)$ and $s(P_3) = (x_u,y_u+2\sigma), (x_v,y_v-\sigma)$ (see \Cref{fig:nobendvert}).

\begin{figure}[htb]

\centering

\begin{subfigure}[b]{.45\textwidth}

\centering

\begin{tikzpicture}

\draw[step=.2,dotted] (-.8,-.5) grid (1.2,.5);

\node[circ,label=below:{\footnotesize $u$}] (u) at (-.6,0) {};

\node[circ,label=below:{\footnotesize $v$}] (v) at (1,0) {};

\draw[-,thick] (u) -- (v) node[midway,above] {\footnotesize $P_{uv}$};

\node[draw=none] at (1.7,0) {$\Rightarrow$};

\draw[dotted,step=.5] (2.2,-1) grid (5.2,1);

\node[circ] at (2.5,0) {};

\node[circ] at (5,0) {};

\draw[blue,thick] (2.5,-.5) -- (2.5,.5) node[midway,left] {\footnotesize $P_u$};

\draw[blue,thick] (5,-.5) -- (5,.5) node[midway,right] {\footnotesize $P_v$};

\draw[thick] (2.6,-.5) -- (2.6,-.1) node[midway,right] {\footnotesize $P_1=P_2$} -- (3.5,-.1);

\draw[thick] (3,0) -- (3.8,0) node[pos=1,above] {\footnotesize $P_3$};
\draw[thick,dashed] (3.8,0) -- (4.5,0);
\draw[-,thick]  (4.5,0) -- (4.9,0) -- (4.9,.5);

\end{tikzpicture}

\caption{$P_{uv}$ is horizontal.}

\label{fig:nobendhori}

\end{subfigure}
\hspace*{.5cm}
\begin{subfigure}[b]{.45\textwidth}

\centering

\begin{tikzpicture}

\draw[step=.2,dotted] (-.5,-.6) grid (.5,1.4);

\node[circ,label=left:{\footnotesize $u$}] (u) at (0,-.4) {};

\node[circ,label=left:{\footnotesize $v$}] (v) at (0,1.2) {};

\draw[-,thick] (u) -- (v) node[midway,right] {\footnotesize $P_{uv}$};

\node[draw=none] at (1,.4) {$\Rightarrow$};

\draw[dotted,step=.5] (1.3,-1.2) grid (2.7,2.7);

\node[circ] at (2,-1) {};

\node[circ] at (2,2.5) {};

\draw[blue,thick] (2,-1) -- (2,0) node[midway,left] {\footnotesize $P_u$};

\draw[blue,thick] (2,2.5) -- (2,1.5) node[midway,left] {\footnotesize $P_v$};

\draw[thick] (2.2,-.5) -- (2.2,.5) node[midway,right] {\footnotesize $P_1=P_2$};

\draw[thick] (2.1,0) -- (2.1,.75);
\draw[thick,dashed] (2.1,.6) -- (2.1,1.4);
\draw[thick] (2.1,1.4) -- (2.1,2) node[pos=.1,right] {\footnotesize $P_3$};

\end{tikzpicture}

\caption{$P_{uv}$ is vertical.}

\label{fig:nobendvert}

\end{subfigure}

\caption{$P_{uv}$ has no bend.}

\end{figure}
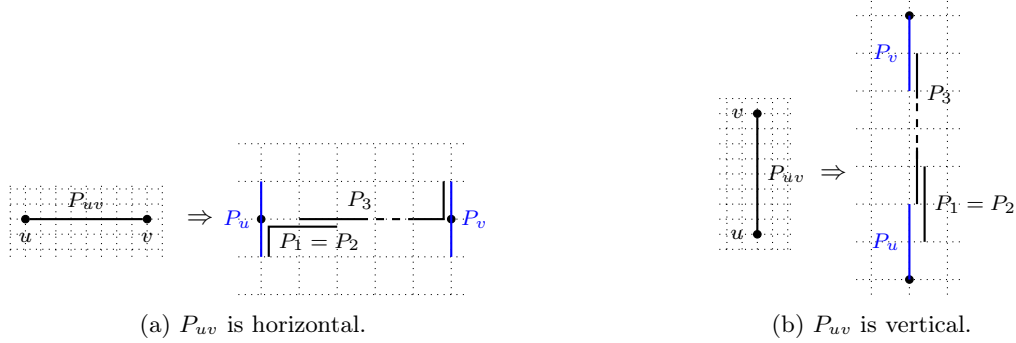

Second, assume that $P_{uv}$ has one bend (the cases where $P_{uv}$ has more than one bend are dealt with similarly). Suppose first that the segment $[(x_u,y_u),(x_1,y_1)]$ of $P_{uv}$ is horizontal, say it lies to the right of $(x_u,y_u)$ (the other case is symmetric). Further suppose that $y_1 < y_v$ (the other case is symmetric). Then, we add the following six paths corresponding to the other vertices in $d_{uv}$: $s(P_1) = s(P_2) = (x_u,y_u-\sigma), (x_u,y_u), (x_u + 2\sigma,y_u)$, $s(P_3) = (x_u+\sigma,y_u), (x_1,y_1), (x_1,y_1+2\sigma)$, $s(P_4) = s(P_5) = (x_1,y_1+\sigma), (x_v,y_v-2\sigma)$ and $s(P_6) = (x_v,y_v-3\sigma), (x_v,y_v-\sigma)$ (see \Cref{fig:1bendhori}). Suppose now that the segment $[(x_u,y_u)(x_1,y_1)]$ is vertical, say it lies above $(x_u,y_u)$ (the other case is symmetric). Further suppose that $x_1 < x_v$ (the other case is symmetric). Then, we add the following six paths corresponding to the other vertices in $d_{uv}$: $s(P_1) = s(P_2) = (x_u,y_u+\sigma), (x_u,y_u+3\sigma)$, $s(P_3) = (x_u,y_u+2\sigma), (x_1,y_1), (x_1+2\sigma,y_1)$, $s(P_4) = s(P_5) = (x_1+\sigma,y_1), (x_v-\sigma,y_v)$ and $s(P_6) = (x_v-2\sigma,y_v), (x_v,y_v), (x_v,y_v+\sigma)$ (see \Cref{fig:1bendvert}).

\begin{figure}[htb]

\centering

\begin{subfigure}[b]{.48\textwidth}

\centering

\begin{tikzpicture}

\draw[step=.2,dotted] (-.8,-.2) grid (1.2,1.8);

\node[circ,label=left:{\footnotesize $u$}] (u) at (-.6,0) {};

\node[circ,label=right:{\footnotesize $v$}] (v) at (1,1.6) {};

\draw[thick] (u) -- (1,0) node[midway,below] {\footnotesize $P_{uv}$} -- (v);

\node[draw=none] at (1.7,.8) {$\Rightarrow$};

\draw[step=.5,dotted] (2.2,-1.2) grid (5.2,2.7);

\node[circ] (u1) at (2.5,-.5) {};

\node[circ] (v1) at (5,2.5) {};

\draw[thick,blue] (2.5,-1) -- (2.5,0) node[midway,left] {\footnotesize $P_u$};

\draw[thick,blue] (v1) -- (5,1.5) node[midway,right] {\footnotesize $P_v$};

\draw[thick] (2.6,-1) -- (2.6,-.4) -- (3.5,-.4) node[midway,above] {\footnotesize $P_1 = P_2$};

\draw[thick] (3,-.5) -- (3.6,-.5);
\draw[thick,dashed] (3.6,-.5) -- (4.6,-.5);
\draw[thick] (4.6,-.5) -- (5,-.5) node[pos=.9,below] {\footnotesize $P_3$} -- (5,.5);

\draw[thick] (5.1,0) -- (5.1,.5);
\draw[thick,dashed] (5.1,.5) -- (5.1,1) node[midway,right] {\footnotesize $P_4 = P_5$};
\draw[thick] (5.1,1) -- (5.1,1.5); 

\draw[thick] (4.9,1) -- (4.9,2) node[midway,left] {\footnotesize $P_6$};

\end{tikzpicture}

\caption{The segment $[(x_u,y_u),(x_1,y_1)]$ of $P_{uv}$ is horizontal.}

\label{fig:1bendhori}

\end{subfigure}
\hspace*{.5cm}
\begin{subfigure}[b]{.45\textwidth}

\centering

\begin{tikzpicture}

\draw[step=.2,dotted] (-.8,-.2) grid (1.2,1.8);

\node[circ,label=left:{\footnotesize $u$}] (u) at (-.6,0) {};

\node[circ,label=right:{\footnotesize $v$}] (v) at (1,1.6) {};

\draw[thick] (u) -- (-.6,1.6) -- (v) node[midway,above] {\footnotesize $P_{uv}$};

\node[draw=none] at (1.7,.8) {$\Rightarrow$};

\draw[step=.5,dotted] (2.2,-.7) grid (5.7,2.7);

\node[circ] (u1) at (2.5,-.5) {};

\node[circ] (v1) at (5.5,2) {};

\draw[thick,blue] (u1) -- (2.5,.45) node[midway,left] {\footnotesize $P_u$};

\draw[thick,blue] (5.5,2.5) -- (5.5,1.5) node[midway,right] {\footnotesize $P_v$};

\draw[thick] (2.6,0) -- (2.6,1) node[midway,right] {\footnotesize $P_1=P_2$};

\draw[thick] (2.5,.55) -- (2.5,1.1);
\draw[thick,dashed] (2.5,1.1) -- (2.5,1.6);
\draw[thick] (2.5,1.6) -- (2.5,2) -- (3.5,2) node[pos=.2,above] {\footnotesize $P_3$};

\draw[thick] (3,1.9) -- (3.6,1.9);
\draw[thick,dashed] (3.6,1.9) -- (4.3,1.9) node[midway,below] {\footnotesize $P_4 = P_5$};
\draw[thick] (4.3,1.9) -- (5,1.9);

\draw[thick] (4.5,2) -- (5.4,2) node[midway,above] {\footnotesize $P_6$} --(5.4,2.5);

\end{tikzpicture}

\caption{The segment $[(x_u,y_u),(x_1,y_1)]$ of $P_{uv}$ is vertical.}

\label{fig:1bendvert}

\end{subfigure}

\caption{$P_{uv}$ has one bend.}

\end{figure}
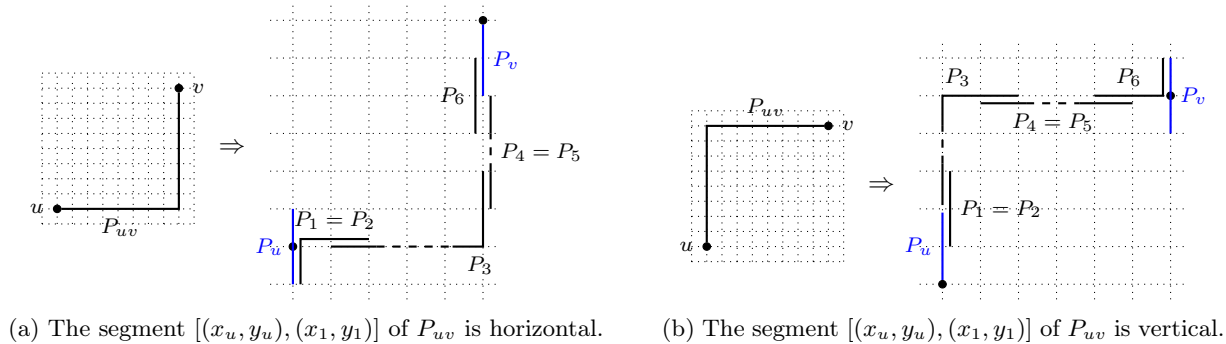

By proceeding in this way for every edge of $G$, we obtain a $1$-bend EPG representation of $G'$. Since in any 3-coloring of the diamond the vertices of degree 2 have the same color, it is then easy to see that $G$ is $3$-colorable if and only if $G'$ is 3-colorable.
\end{proof}

%===========================================================

\section{Concluding remarks}

In this paper, we showed that it is $\mathsf{NP}$-complete to recognize CPG graphs and $B_{k}$-CPG graphs for any $k \geq 0$ (the case $k = 0$ was addressed in \citep{cpg}). Moreover, the problem remains $\mathsf{NP}$-hard even within the class of planar graphs for $k \geq 3$. We leave as an open problem to determine the complexity of recognizing planar $B_{k}$-CPG graphs for $k \leq 2$. 

%===========================================================
\bibliographystyle{plainnat}
\bibliography{references}

\begin{thebibliography}{35}
\providecommand{\natexlab}[1]{#1}
\providecommand{\url}[1]{\texttt{#1}}
\expandafter\ifx\csname urlstyle\endcsname\relax
  \providecommand{\doi}[1]{doi: #1}\else
  \providecommand{\doi}{doi: \begingroup \urlstyle{rm}\Url}\fi

\bibitem[Aerts and Felsner(2015)]{aerts}
N.~Aerts and S.~Felsner.
\newblock Vertex contact representations of paths on a grid.
\newblock \emph{Journal of Graph Algorithms and Applications}, 19\penalty0
  (3):\penalty0 817--849, 2015.

\bibitem[Alc{\'o}n et~al.(2017)Alc{\'o}n, Bonomo, and Mazzoleni]{alcon}
L.~Alc{\'o}n, F.~Bonomo, and M.~P. Mazzoleni.
\newblock Vertex intersection graphs of paths on a grid: Characterization
  within block graphs.
\newblock \emph{Graphs and Combinatorics}, 33\penalty0 (4):\penalty0 653--664,
  2017.

\bibitem[Alc{\'o}n et~al.(2018)Alc{\'o}n, Bonomo, Dur{\'a}n, Gutierrez,
  Mazzoleni, Ries, and Valencia-Pabon]{NCA}
L.~Alc{\'o}n, F.~Bonomo, G.~Dur{\'a}n, M.~Gutierrez, M.~P. Mazzoleni, B.~Ries,
  and M.~Valencia-Pabon.
\newblock On the bend number of circular-arc graphs as edge intersection graphs
  of paths on a grid.
\newblock \emph{Discrete Applied Mathematics}, 234:\penalty0 12--21, 2018.

\bibitem[Asinowski et~al.(2012)Asinowski, Cohen, Golumbic, Limouzy, Lipshteyn,
  and Stern]{asinowski}
A.~Asinowski, E.~Cohen, M.~C. Golumbic, V.~Limouzy, M.~Lipshteyn, and M.~Stern.
\newblock Vertex intersection graphs of paths on a grid.
\newblock \emph{Journal of Graph Algorithms and Applications}, 16\penalty0
  (2):\penalty0 129--150, 2012.

\bibitem[Biedl and Stern(2010)]{biedl}
T.~C. Biedl and M.~Stern.
\newblock On edge-intersection graphs of k-bend paths in grids.
\newblock \emph{Discrete Mathematics {\&} Theoretical Computer Science},
  12\penalty0 (1):\penalty0 1--12, 2010.

\bibitem[Bonomo et~al.(2017)Bonomo, Mazzoleni, and Stein]{bonomo}
F.~Bonomo, M.~P. Mazzoleni, and M.~Stein.
\newblock Clique coloring {B1-EPG} graphs.
\newblock \emph{Discrete Mathematics}, 340\penalty0 (5):\penalty0 1008--1011,
  2017.

\bibitem[Booth and Lueker(1976)]{BL76}
K.~S. Booth and G.~S. Lueker.
\newblock Testing for the consecutive ones property, interval graphs, and graph
  planarity using {PQ}-tree algorithms.
\newblock \emph{Journal of Computer and System Sciences}, 13\penalty0
  (3):\penalty0 335--379, 1976.

\bibitem[Chaplick and Ueckerdt(2013)]{chaplick12}
S.~Chaplick and T.~Ueckerdt.
\newblock Planar graphs as {VPG}-graphs.
\newblock \emph{Journal of Graph Algorithms and Applications}, 17\penalty0
  (4):\penalty0 475--494, 2013.

\bibitem[Chaplick et~al.(2012)Chaplick, Jel{\'i}nek, Kratochv{\'i}l, and
  Vysko{\v{c}}il]{CJKV12}
S.~Chaplick, V.~Jel{\'i}nek, J.~Kratochv{\'i}l, and T.~Vysko{\v{c}}il.
\newblock Bend-bounded path intersection graphs: Sausages, noodles, and waffles
  on a grill.
\newblock In M.~C. Golumbic, M.~Stern, A.~Levy, and G.~Morgenstern, editors,
  \emph{Graph-Theoretic Concepts in Computer Science}, pages 274--285. Springer
  Berlin Heidelberg, 2012.

\bibitem[Chaplick et~al.(2013)Chaplick, Kobourov, and Ueckerdt]{chaplick13}
S.~Chaplick, S.~G. Kobourov, and T.~Ueckerdt.
\newblock Equilateral {L}-contact graphs.
\newblock In A.~Brandst{\"a}dt, K.~Jansen, and R.~Reischuk, editors,
  \emph{Graph-Theoretic Concepts in Computer Science}, pages 139--151, 2013.

\bibitem[Clark et~al.(1990)Clark, Colbourn, and Johnson]{clark}
B.~N. Clark, C.~J. Colbourn, and D.~S. Johnson.
\newblock Unit disk graphs.
\newblock \emph{Discrete Mathematics}, 86\penalty0 (1):\penalty0 165--177,
  1990.

\bibitem[Cohen et~al.(2016)Cohen, Golumbic, Trotter, and Wang]{cohen}
E.~Cohen, M.~C. Golumbic, W.~T. Trotter, and R.~Wang.
\newblock Posets and {VPG} graphs.
\newblock \emph{Order}, 33\penalty0 (1):\penalty0 39--49, 2016.

\bibitem[Dahlhaus et~al.(1994)Dahlhaus, Johnson, Papadimitriou, Seymour, and
  Yannakakis]{Dahlhaus}
E.~Dahlhaus, D.~S. Johnson, C.~H. Papadimitriou, P.~D. Seymour, and
  M.~Yannakakis.
\newblock The complexity of multiterminal cuts.
\newblock \emph{SIAM Journal on Computing}, 23:\penalty0 864--894, 1994.

\bibitem[Deniz et~al.(2018)Deniz, Galby, Munaro, and Ries]{cpg}
Z.~Deniz, E.~Galby, A.~Munaro, and B.~Ries.
\newblock On contact graphs of paths on a grid.
\newblock In T.~Biedl and A.~Kerren, editors, \emph{Graph Drawing and Network
  Visualization}, pages 317--330, 2018.

\bibitem[Diestel(2005)]{diestel}
R.~Diestel.
\newblock \emph{{G}raph {T}heory}.
\newblock Graduate Texts in Mathematics. Springer, 2005.

\bibitem[Epstein et~al.(2013)Epstein, Golumbic, and Morgenstern]{epstein}
D.~Epstein, M.~C. Golumbic, and G.~Morgenstern.
\newblock Approximation algorithms for {$B_1$-EPG} graphs.
\newblock In F.~Dehne, R.~Solis-Oba, and J.-R. Sack, editors, \emph{Algorithms
  and Data Structures}, pages 328--340, 2013.

\bibitem[Felsner et~al.(2016)Felsner, Knauer, Mertzios, and Ueckerdt]{felsner}
S.~Felsner, K.~Knauer, G.~B. Mertzios, and T.~Ueckerdt.
\newblock Intersection graphs of {L}-shapes and segments in the plane.
\newblock \emph{Discrete Applied Mathematics}, 206:\penalty0 48--55, 2016.

\bibitem[Francis and Lahiri(2016)]{francis}
M.~C. Francis and A.~Lahiri.
\newblock {VPG} and {EPG} bend-numbers of {H}alin graphs.
\newblock \emph{Discrete Applied Mathematics}, 215:\penalty0 95--105, 2016.

\bibitem[Fraysseix et~al.(1991)Fraysseix, Mendez, and Pach]{planbip}
H.~De Fraysseix, P.~Ossona~De Mendez, and J.~Pach.
\newblock Representation of planar graphs by segments.
\newblock \emph{Intuitive Geometry}, 63:\penalty0 109--117, 1991.

\bibitem[Garey et~al.(1976)Garey, Johnson, and Stockmeyer]{garey}
M.~R. Garey, D.~S. Johnson, and L.~Stockmeyer.
\newblock Some simplified {NP}-complete graph problems.
\newblock \emph{Theoretical Computer Science}, 1\penalty0 (3):\penalty0
  237--267, 1976.

\bibitem[Golumbic et~al.(2009)Golumbic, Lipshteyn, and Stern]{golumbic}
M.~C. Golumbic, M.~Lipshteyn, and M.~Stern.
\newblock Edge intersection graphs of single bend paths on a grid.
\newblock \emph{Networks}, 54\penalty0 (3):\penalty0 130--138, 2009.

\bibitem[Gon\c{c}alves et~al.(2018)Gon\c{c}alves, Isenmann, and Pennarun]{gonc}
D.~Gon\c{c}alves, L.~Isenmann, and C.~Pennarun.
\newblock Planar graphs as {L}-intersection or {L}-contact graphs.
\newblock In \emph{Proceedings of the Twenty-Ninth Annual ACM-SIAM Symposium on
  Discrete Algorithms}, SODA '18, pages 172--184, 2018.

\bibitem[Heldt et~al.(2014{\natexlab{a}})Heldt, Knauer, and Ueckerdt]{heldt1}
D.~Heldt, K.~Knauer, and T.~Ueckerdt.
\newblock On the bend-number of planar and outerplanar graphs.
\newblock \emph{Discrete Applied Mathematics}, 179:\penalty0 109--119,
  2014{\natexlab{a}}.

\bibitem[Heldt et~al.(2014{\natexlab{b}})Heldt, Knauer, and Ueckerdt]{heldt2}
D.~Heldt, K.~Knauer, and T.~Ueckerdt.
\newblock Edge-intersection graphs of grid paths: {T}he bend-number.
\newblock \emph{Discrete Applied Mathematics}, 167:\penalty0 144--162,
  2014{\natexlab{b}}.

\bibitem[{Hlin\u{e}n\'y}(1998)]{Hlineny}
Petr {Hlin\u{e}n\'y}.
\newblock Classes and recognition of curve contact graphs.
\newblock \emph{Journal of Combinatorial Theory, Series B}, 74\penalty0
  (1):\penalty0 87--103, 1998.

\bibitem[Kobourov et~al.(2013)Kobourov, Ueckerdt, and Verbeek]{kobourov}
S.~Kobourov, T.~Ueckerdt, and K.~Verbeek.
\newblock Combinatorial and geometric properties of planar {L}aman graphs.
\newblock In \emph{Proceedings of the Twenty-fourth Annual ACM-SIAM Symposium
  on Discrete Algorithms}, SODA '13, pages 1668--1678, 2013.

\bibitem[Kratochv\'{i}l(1991)]{Kra91}
J.~Kratochv\'{i}l.
\newblock String graphs. {II}. {R}ecognizing string graphs is
  $\mathsf{NP}$-hard.
\newblock \emph{Journal of Combinatorial Theory, Series B}, 52\penalty0
  (1):\penalty0 67--78, 1991.

\bibitem[Kratochv\'{i}l(1994)]{Kra94}
J.~Kratochv\'{i}l.
\newblock A special planar satisfiability problem and a consequence of its
  $\mathsf{NP}$-completeness.
\newblock \emph{Discrete Applied Mathematics}, 52\penalty0 (3):\penalty0
  233--252, 1994.

\bibitem[Kratochv\'{i}l and Matou\v{s}ek(1991)]{KM91}
J.~Kratochv\'{i}l and J.~Matou\v{s}ek.
\newblock String graphs requiring exponential representations.
\newblock \emph{Journal of Combinatorial Theory, Series B}, 53\penalty0
  (1):\penalty0 1--4, 1991.

\bibitem[Mohar(2001)]{mohar}
B.~Mohar.
\newblock Face covers and the genus problem for apex graphs.
\newblock \emph{Journal of Combinatorial Theory, Series B}, 82\penalty0
  (1):\penalty0 102--117, 2001.

\bibitem[Munaro(2017{\natexlab{a}})]{munaro}
A.~Munaro.
\newblock Bounded clique cover of some sparse graphs.
\newblock \emph{Discrete Mathematics}, 340\penalty0 (9):\penalty0 2208--2216,
  2017{\natexlab{a}}.

\bibitem[Munaro(2017{\natexlab{b}})]{munaro1}
A.~Munaro.
\newblock Boundary classes for graph problems involving non-local properties.
\newblock \emph{Theoretical Computer Science}, 692:\penalty0 46--71,
  2017{\natexlab{b}}.

\bibitem[Pergel and Rz\k{a}\.{z}ewski(2017)]{PR17}
M.~Pergel and P.~Rz\k{a}\.{z}ewski.
\newblock On edge intersection graphs of paths with 2 bends.
\newblock \emph{Discrete Applied Mathematics}, 226:\penalty0 106--116, 2017.

\bibitem[Schaefer et~al.(2003)Schaefer, Sedgwick, and
  \v{S}tefankovi\v{c}]{SSS03}
M.~Schaefer, E.~Sedgwick, and Daniel \v{S}tefankovi\v{c}.
\newblock Recognizing string graphs in $\mathsf{NP}$.
\newblock \emph{Journal of Computer and System Sciences}, 67\penalty0
  (2):\penalty0 365--380, 2003.

\bibitem[Tamassia and Tollis(1989)]{tamassia}
R.~Tamassia and I.~G. Tollis.
\newblock Planar grid embedding in linear time.
\newblock \emph{IEEE Transactions on Circuits and Systems}, 36\penalty0
  (9):\penalty0 1230--1234, 1989.

\end{thebibliography}

\end{document}